\documentclass[10pt, journal]{IEEEtran}
\usepackage{comment}
\usepackage{amsfonts}
\usepackage{amsthm}
\usepackage{amsmath}
\usepackage{amssymb}
\usepackage{array}
\usepackage{bm}         
\usepackage{bbm}
\usepackage{cite}
\usepackage{dsfont}
\usepackage{float}
\usepackage[T1]{fontenc}
\usepackage{graphicx}
\usepackage{indentfirst}
\usepackage{multicol}
\usepackage{blkarray}
\usepackage{subfigure}
\usepackage[export]{adjustbox}
\usepackage{footnote}
\makesavenoteenv{tabular}
\makesavenoteenv{table}
\usepackage{multirow}
\usepackage{mathtools}
\usepackage{wasysym}
\usepackage{xr}
\newcommand{\set}[1]{\mathcal{#1}} 
\newcommand{\eqdef}{\triangleq} 
\newcommand{\mat}[1]{\mathsf{#1}} 
\newcommand{\vect}[1]{\bm{#1}} 
\newcommand{\card}[1]{\left|#1\right|} 
\newcommand{\trans}[1]{#1^{\textup{\textsf{\tiny T}}}} 

\newtheorem{thm}{Theorem}
\newtheorem{lem}{Lemma}
\newtheorem{mydef}{Definition}

\newtheorem{cor}{Corollary}

\usepackage[ruled,linesnumbered]{algorithm2e}
\SetAlFnt{\small}

\SetCommentSty{mycommfont}
\makesavenoteenv{algorithm}

\ifCLASSINFOpdf
\else
\fi
\begin{document}
%
\title{Optimal and Approximation Algorithms for Joint Routing and Scheduling in Millimeter-Wave Cellular Networks}

%
%
%

\author{Dingwen~Yuan,
        Hsuan-Yin~Lin,
        J\"{o}rg~Widmer and
	Matthias Hollick
\thanks{D. Yuan and M. Hollick are with Secure Mobile Networking
Lab (SEEMOO), Technische Universit\"{a}t Darmstadt, Darmstadt, Germany 
(e-mail: dyuan@seemoo.tu-darmstadt.de; mhollick@seemoo.tu-darmstadt.de).}
\thanks{H. Lin is with Simula UiB Research Lab, Bergen, Norway
(e-mail: hsuan-yin.lin@ieee.org)}
\thanks{J. Widmer is with Institute IMDEA Networks, Madrid, Spain
(e-mail: joerg.widmer@imdea.org)}}
\maketitle


\begin{abstract}
	
	Millimeter-wave (mmWave) communication is a promising technology to cope with the exponential increase in 5G data traffic.
	Such networks typically require a very dense deployment of base stations. 
	A subset of those, so-called macro base stations, feature high-bandwidth connection to the core network, while relay base stations are connected wirelessly.
	To reduce cost and increase flexibility, wireless backhauling is needed to connect both macro to relay as well as relay  to relay base stations. 
	The characteristics of mmWave communication mandates new paradigms for routing and scheduling. 
	The paper investigates scheduling algorithms under different interference models. 
	To showcase the scheduling methods, we study the maximum throughput fair scheduling  problem. Yet the proposed algorithms can be easily extended to other problems.
	For a full-duplex network under the no interference model, we propose an efficient polynomial-time scheduling method, the {\em schedule-oriented optimization}. Further, we prove that the problem is NP-hard if we assume pairwise link interference model or half-duplex radios.
	Fractional weighted coloring based approximation algorithms are proposed for these NP-hard cases.
	Moreover, the approximation algorithm parallel data stream scheduling is proposed for the case of half-duplex network under the no interference model. It has better approximation ratio than the fractional weighted coloring based algorithms and even attains the optimal solution for the special case of uniform orthogonal backhaul networks.
\end{abstract}


%
\IEEEpeerreviewmaketitle

\section{Introduction}
5G and beyond cellular systems are embracing millimeter wave (mmWave) communication in the $10$-$300$ GHz band where abundant bandwidth is available to achieve Gbps data rate. One of the main challenges for mmWave systems is the high propagation loss. 
Although it can be partially compensated by directional antennas~\cite{Rappaport:2013jk, Rangan:2014ia}, the effective communication range of a mmWave base station (BS) is around $100$ meters for typical use cases. Thus, base station deployment density in 5G will be significantly higher than in 4G~\cite{Singh:2015eh, ghosh2014millimeter}. This leads to high infrastructure cost for the operators. In fact, besides the cost of site lease, backhaul link provisioning is the main contributor to this expense because the mmWave access network may require multi-Gbps backhaul links to the core network. Currently, such a high data rate can only be accommodated by fiber-optic links which have high installation costs and are inflexible for reconfiguration.  

Recent studies show that mmWave self-backhauling is a  cost-effective alternative to wired backhauling~\cite{Dehos14,Singh:2015eh}. 
This approach is particularly interesting in an NG-RAN (Next Generation Radio Access Network) where one or more {\em 5G base stations} (gNBs)
have fiber backhaul to the core network and act as gateways for the other gNBs~\cite{NGRAN}.
We refer to the gateway gNBs as {\em macro BSs} and the other gNBs as {\em relay BSs}.
Fig.~\ref{fig:net-model} illustrates such a setup in which each relay BS can be reached by at least one macro BS directly or via other relay~BSs.
Moreover, the directionality of mmWave communication reduces or removes the wireless backhaul interference and allows simultaneous scheduling of multiple links over the same channel as long as their antenna beams do not overlap. However, the number of simultaneous data streams a base station can handle is limited by the number of its radio frequency (RF) chains. Furthermore, a base station may only support half-duplex communication, i.e. it {\em can not} work as a transmitter and a receiver at the same time. 

\begin{figure}[t!]
\centering
		\includegraphics[width=8cm]{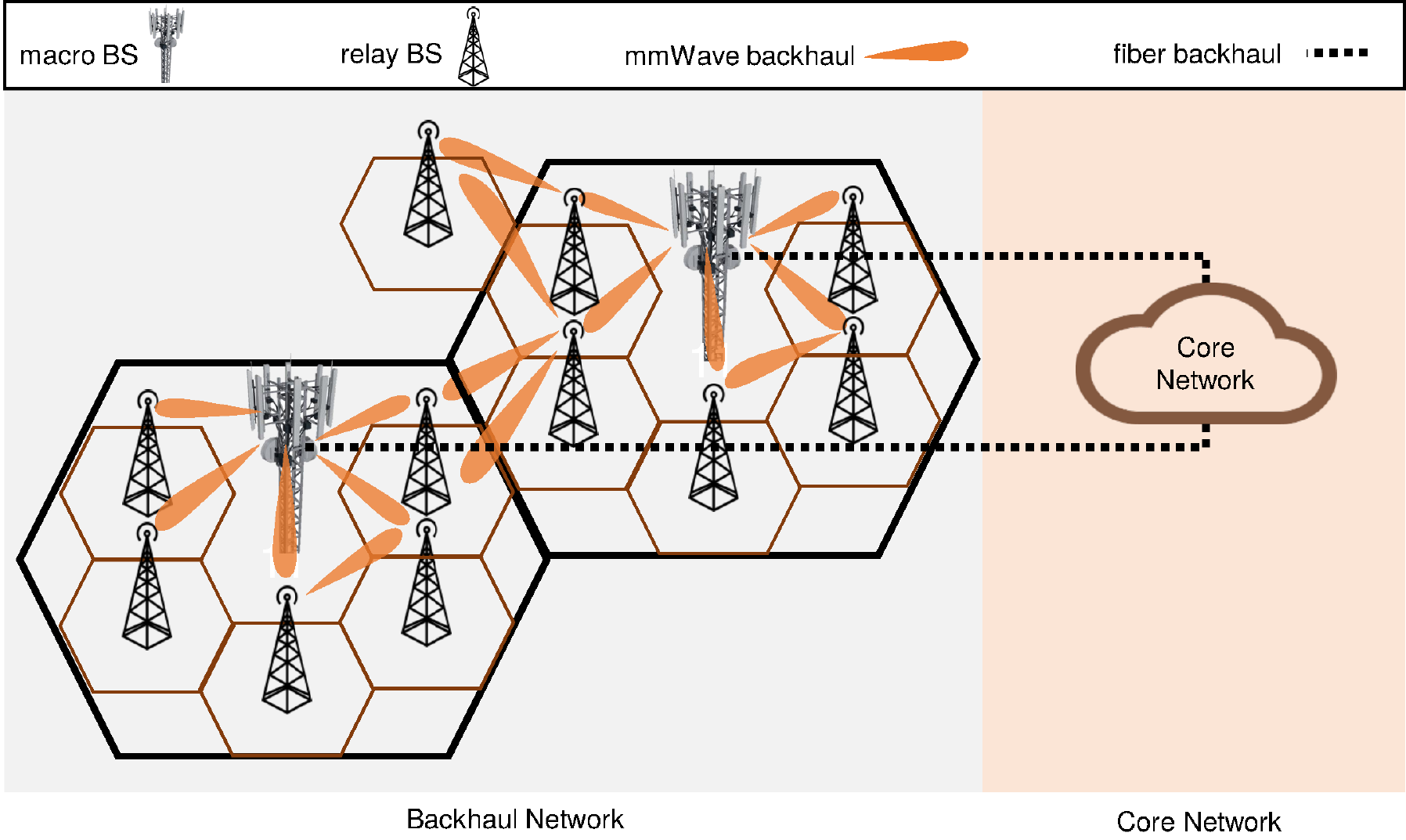} 
		\caption{mmWave self-backhauling setup.}
		\label{fig:net-model}
\end{figure}

As of now, much of the research on mmWave communication has been dedicated to issues that the mobile users (UEs) face in the access networks. 
How to maximize performance such as throughput and energy efficiency in mmWave backhaul networks has received less attention.
Transmission scheduling that incorporates the possibility of multi-hop routing is a most crucial research question to be addressed.


A naive scheduling which lets a macro BS serve all the relay BSs in its macrocell in a round robin fashion is neither practical nor efficient.
If a relay BS' link to a macro BS is weaker than its links to other nearby relay BSs (which in turn have high-capacity links to a macro BS), a schedule allowing multi-hop routing can be more favorable since it alleviates the bottleneck at the macro BSs. At the same time, the limited interference at mmWave frequencies makes it efficient to maximize spatial reuse and operate as many data streams simultaneously as possible.
The goal of this paper is to design a scheduler that exploits these characteristics to optimize mmWave backhaul efficiency for full-duplex and half-duplex radios assuming pairwise link interference or not, and realistic or maximum single-user spatial multiplexing.

The paper is organized as follows: related work is discussed in \S \ref{s:related}.
The system model  is described in \S \ref{s:sysmodel}.
The {\em maximum throughput fair scheduling} (MTFS) serves as a concrete problem for presenting our methods for backhaul scheduling.
In \S\ref{s:fd-sched-fair}, we present a polynomial time optimal algorithm for the full-duplex MTFS problem assuming no interference (NI) between links.
In \S\ref{sec:complexity}, we show that the MTFS problem is NP-hard for both cases of pairwise link interference (PI) model and half-duplex radios. Yet it is solvable in polynomial time for a special case---the so-called {\em uniform orthogonal backhaul network}.
In \S \ref{sec:app-fra-color}, we propose two general approximation algorithms based on fractional weighted coloring for the NP-hard cases.
For half-duplex radios under the NI model, we propose the approximation algorithm---parallel data stream scheduling (PDS) in \S\ref{sec:pds}, which has better a performance bound than the fractional weighted coloring based algorithms. It provides the optimal scheduling for the case of uniform orthogonal backhaul networks.
\S\ref{sec:extension} elaborates on the extension of the scheduling algorithms to more general scenarios.
\S\ref{s:eval} demonstrates the efficiency and effectiveness of the algorithms through numerical evaluations. Finally, \S \ref{s:conclude} concludes the paper.

\section{Related Work}
\label{s:related}

Our previous paper~\cite{YuanLWH18} discussed the optimal full-duplex scheduling of mmWave backhaul networks and its approximation algorithm assuming zero interference between any pair of links and the linearity of link capacity in terms of RF chains.
This paper is a continuation of the work, which includes both full-duplex and half-duplex scheduling of mmWave backhaul networks for pairwise link interference and realistic single-user spatial-multiplexing model.
A few works studying mmWave network scheduling~\cite{Niu15, Zhu16, Feng16, Li17, Niu17} share the assumption that the traffic demand is measured in discrete units of timeslots or packets.
 The resulting optimization problems are all formulated as mixed integer programs (MIP). As MIPs are in general NP-hard, optimal solutions can only be computed for very small networks. For practical use, they all rely on heuristics, which are based for example on greedy edge coloring~\cite{Niu15, Feng16} or finding the maximum independent set in a graph~\cite{Zhu16, Li17, Niu17}. Furthermore, \cite{Niu15, Zhu16, Li17, Niu17} assume that routing is pre-determined, which does not fully exploit the freedom given by a reconfigurable backhaul, and may thus limit performance.

In contrast, we relax the constraint of in-order flow scheduling (i.e., if needed, packets are queued for a short time) which does not harm the long-term throughput, and allows a timeslot to be of any length (a following step can be used to discretize the length).
 Based on these assumptions, we propose polynomial time optimal and approximation scheduling algorithms. We show by simulation that they are practical for mmWave cellular networks. 
Moreover, the scheduling takes QoS optimization goals or QoS requirements as input and finds an efficient routing automatically.
The first attempt to solve the problem of joint routing and scheduling in a network with Edmonds' matching formulation goes back to~\cite{Hajek88}. Hajek's scheduling algorithm differs from ours in that it minimizes schedule length. Furthermore, we use a one-step {\em schedule-oriented} approach while they first compute the optimal link time and then compute the minimum length schedule given the link time. 
The matching-based approach is also used in~\cite{Sinha19} to study the problem of throughput-optimal routing and scheduling in a wireless directed acyclic graph (DAG) with a time-varying connectivity, to which a throughput-optimal dynamic broadcast policy is proposed.

Two of our approximation algorithms (F$^3$WC) are adapted from~\cite{Wan09} which proposed scheduling algorithms for multihop wireless networks based on fractional weighted vertex coloring of conflict graphs. 
We show the generality of F$^3$WC algorithms by demonstrating that all constraints of half-duplex, RF chain number and pairwise link interference can be modeled in a conflict graph.

 While current mmWave radios only support half-duplex operation, full-duplex is feasible through proper analog and digital cancellation~\cite{LiJT14} and is likely to be used in the future.
More details about the design of a full-duplex mmWave radio are discussed in~\cite{Dinc18}.
In addition, we assume that multiple data streams can be transmitted concurrently between a single pair or multiple pairs of BSs. This assumption is backed by the feasibility of single-user and multi-user MIMO in mmWave communication~\cite{SunRHNR14, RaghavanPSSKRCM18, XueFW17}, exploiting beamforming, multipath, spatial and polarization diversity.

Existing research works on mmWave backhaul scheduling favor centralized solutions for various optimization goals, including throughput~\cite{ Feng16, Li17, YuanLWH18}, delay~\cite{Kilpi17}, energy consumption~\cite{Niu17}, makespan~\cite{Niu15, Arribas19}, wireless bandwidth~\cite{DuOCVV17} and flows with satisfied QoS requirements~\cite{Zhu16}.
To achieve these optimization goals, most of the works perform link scheduling to maximize spatial reuse while some control routing~\cite{Niu15, Feng16, DuOCVV17, YuanLWH18, Arribas19}, transmission power~\cite{Li17, Niu17, DuOCVV17} and bandwidth allocation~\cite{DuOCVV17}.
3GPP is currently performing a study on Integrated Access and Backhaul (IAB)~\cite{3GPP-IAB} which describes mechanisms for sharing radio resources between access and backhaul links.  Yet designing a high-performance IAB is still an open problem~\cite{PoleseGZRGCZ20, SahaD19}.
Our algorithms can be equally applied to IAB networks because of our graph-based scheduling method, thus contributing to this emerging field.

\section{System model}
\label{s:sysmodel}
The system model considers a backhaul network which consists of one or more macro BSs and multiple relay BSs. 
The macro BSs act as the backhaul gateways for the relay BSs.
We assume that any two macro BSs are connected by fiber optics with infinite capacity (data rate). Thus, data can be exchanged between two macro BSs with zero delay\footnote{We can equally support non-zero communication delay between macro BSs. This can be done by adding wired links of fixed capacity between macro BSs. These links do not interfere with each other or with the wireless links. For presentation clarity, we assume that wired links introduce zero delay.}.

A node (BS) is equipped with multiple RF chains. Two nodes may have a different number of RF chains. We assume analog or hybrid beamforming which allows up to the same number of {\em simultaneous data streams} incident to a node as its number of RF chains.
In addition, the maximum number of simultaneous data streams supported by a link  is determined by the spatial diversity of that link (see. \S\ref{sec:su-sm}). These data streams may have different capacities.
Finally, the network is either {\em full-duplex} or {\em half-duplex}, meaning all nodes are equipped with full-duplex or half-duplex radios, respectively. 

We use the convention that {\em arc} and {\em edge} refer to {\em directed} and {\em undirected edge}  of a graph.
The notations $(u, v)$ and $\{u, v\}$ are used to represent an arc from $u$ to $v$ and an edge between $u$ and $v$, respectively.
To model the backhaul network, we start with a {\em directed network} $D$, which is an {\em arc weighted directed multigraph}\footnote{A directed multigraph allows parallel arcs with the same head node and the same tail node, but disallows loops.}. 
In this paper, we use the notation $V(G)$ and $E(G)$ to refer to the vertex set and arc (edge) set of a directed (or undirected) graph $G$. Thus, $V(D)$ and $E(D)$ are the vertex set and arc set of $D$, representing the set of all BSs and the set of all possible data streams.
Let $B(D)$ and $M(D)$ be the set of macro BS vertices and the set of relay BS vertices, respectively.
Let $L$ be the {\em link network} of $D$. It is a subgraph of $D$ which is constructed by replacing each set of parallel arcs by a single arc. 
Let $l = (u, v) \in E(L)$ denote the link from node $u$ to node $v$ and $d(l)$ the maximum number of simultaneous data streams allowed by $l$. 
Then there are $d(l)$ parallel arcs from $u$ to $v$ in $D$, which are labeled as $l_1 \dots l_{d(l)}$.
Let $c: E(D) \mapsto \mathbb{Q}_+$ be the {\em capacity} function defined for each arc (data stream) where $\mathbb{Q}_+$ is the set of non-negative rational numbers. The value $c(e)$ is the capacity of a data stream $e$.
Fig.~\ref{fig:sys_model} illustrates a toy example for a directed network and the corresponding link network. 
\begin{figure}[hbpt]
\centering
\includegraphics[width=8cm]{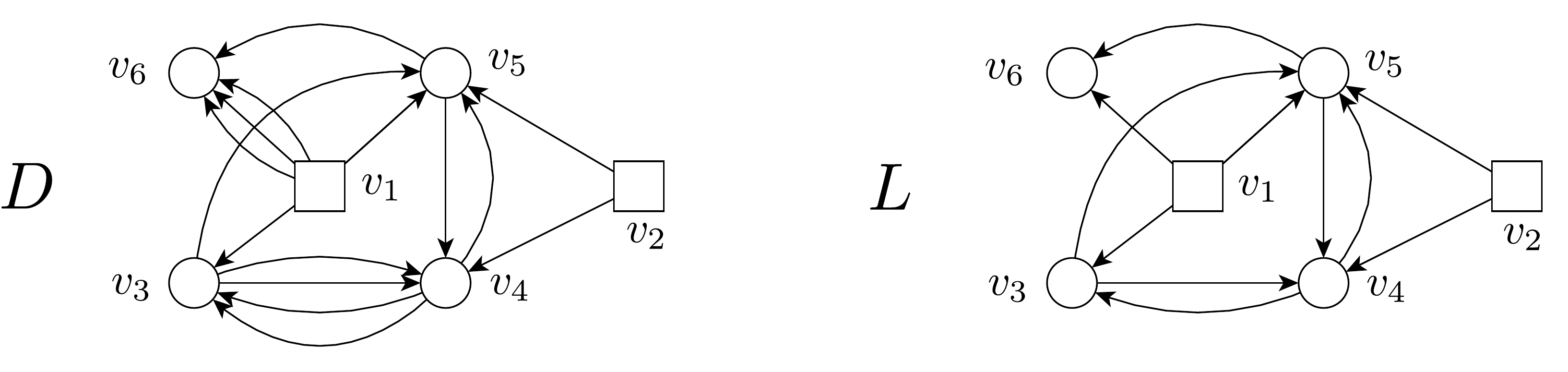}
\caption{Example of {\em directed network} and {\em link network}.
$\Square$ and $\Circle$ represent a macro and a relay BS, respectively.}
\label{fig:sys_model}	
\end{figure}
 
In addition, we  define an {\em RF chain number} function $r: V(D) \mapsto \mathbb{N}$ where $\mathbb{N}$ is the set of natural numbers $\{1, 2, \dots \}$. $r(u)$ is the number of RF chains of node $u$.  We denote the total number of RF chains of  $D$ by $r(D)$, i.e., $r(D) = \sum_{v \in V(D)} r(v)$. Scheduling is made based on the directed network $D$.

\subsection{Single User Spatial Multiplexing Model} \label{sec:su-sm}
Parallel data streams from a transmitter $u$ to a receiver $v \ne u$ may be scheduled simultaneously if both $u$ and $v$ use multiple RF chains. This is referred to as {\em single user spatial multiplexing} (SU-SM)~\cite{SunRHNR14}. In this paper, we assume that every RF chain of a node $v$ is assigned the same transmission power $p_{\text{tx}}(v)$. $p_{\text{tx}}(u)$ and $p_{\text{tx}}(v)$ may be different if $u \ne v$.
This is a reasonable assumption for a multi-transceiver node that has multiple RF chains, each with their own amplifiers and phased arrays (see~\cite{ZhaoWWQZ20} for an example MIMO system).
We investigate two models for SU-SM:
\begin{enumerate}
	\item{REAL-SU-SM:} This is a general and realistic model. A link $l = (u, v) \in E(L)$ supports at most $d(l)$ data streams where $d(l) \le \min \big(r(u), r(v)\big)$ and is determined by the spatial diversity of link $l$ such as the {\em channel matrix rank}. 
	Since only the total capacity of the data streams is relevant in our scheduling problem, we can equivalently model their capacities such that $c(l_1) \ge \dots \ge c(l_{d(l)}) > 0$. 
	This means that we have capacity $c(l_1)$ if only one data stream is active. With each additional data stream we have a more marginal increase in total capacity.
	Using $k \le d(l)$ RF chains at both ends, we get a total capacity of $\sum_{i=1}^k c(l_i)$.
	
	\item{MAX-SU-SM:} This model  achieves the maximum possible capacity of SU-SM, where each link provides sufficient spatial diversity.
	It is a special case of the REAL-SU-SM model with $d(l) = \min \big(r(u), r(v)\big)$ and $c(l_1) = \dots = c(l_{d(l)}) > 0$.
	In this model, capacity is also defined for a link such that $c(l) = c(l_1)$.

\end{enumerate}

\subsection{RF Chain Number Constraint and Half-duplex Constraint}

The constraint due to the limited number of RF chains is that given a set of simultaneously scheduled data streams $E \in E(D)$, for each node $v$, the number of arcs in $E$ that are incident to $v$ is at most $r(v)$.
Furthermore, a schedule can be either {\em full-duplex} (FD) or {\em half-duplex} (HD). FD allows incoming and outgoing arcs of a vertex to be scheduled simultaneously, whereas HD does not. We will prove in the paper that the duplex state is a key parameter for scheduling. It determines whether the  computational complexity is polynomial time  or NP-hard.

\subsection{Interference Model} \label{sec:sys-intf}
Compared to lower frequency omnidirectional communication, mmWave frequencies significantly reduce the interference between concurrent links  due to the short communication range and the directionality brought by beamforming. We will study and compare two interference models in this paper: 
\begin{enumerate}
	\item{NI:} no interference between concurrent links.
	\item{PI:} pairwise link interference.
\end{enumerate}

\begin{figure}[hbpt]
	\centering
	\includegraphics[width=3.7cm]{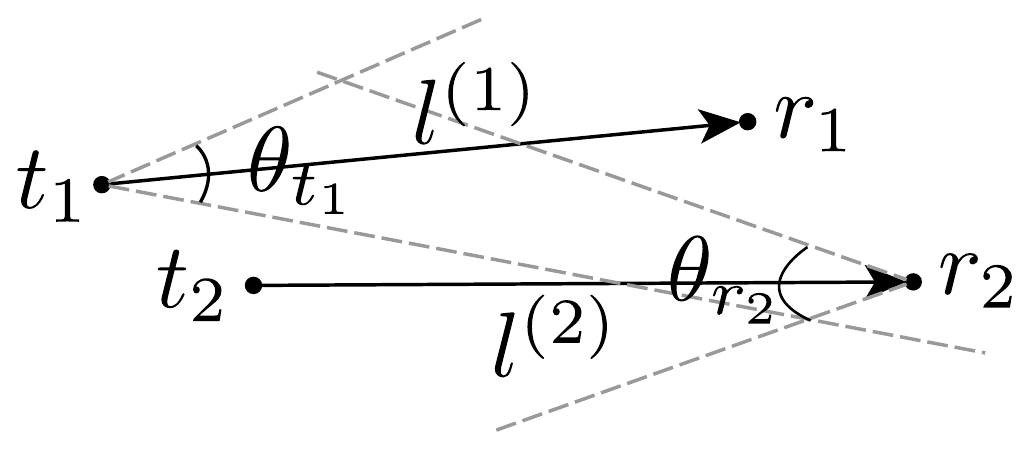}
	\caption{Two interfering links $l^{(1)}$ and $l^{(2)}$. An angle between two dashed rays starting from a node is the beamwidth of that node.}
	\label{fig:intf-links}	
\end{figure}

The NI model (no interference assumption) is reasonable for the situation of narrow beams and sparse deployments.
The PI model specifies that given two links $l^{(1)}, l^{(2)} \in E(L)$ sharing no common vertex, the function $\mathrm{intf}(l^{(1)}, l^{(2)}) = 1$ if there is interference between $l^{(1)}$ and $l^{(2)}$, i.e., they cannot be scheduled simultaneously; otherwise $\mathrm{intf}(l^{(1)}, l^{(2)}) = 0$. As shown in Fig.~\ref{fig:intf-links}, we say that {\em $l^{(1)}$ interferes $l^{(2)}$} if (1) $r_2$ and $t_1$ are within the beamwidths of each other, and (2) $\mathrm{SINR} = \frac{p_{r_2}^{t_2}}{p_{r_2}^{t_1} + N_0} < \tau$, where $p_r^t$ is the received power at receiver $r$ due to transmitter $t$;  $N_0$ is the noise power and $\tau$ is a SINR threshold. $\mathrm{intf}(l^{(1)}, l^{(2)}) = 1$ if $l^{(1)}$ interferes $l^{(2)}$ or $l^{(2)}$ interferes $l^{(1)}$. Otherwise, $\mathrm{intf}(l^{(1)}, l^{(2)}) = 0$.
For the case that two links have vertices in common, we assume that advanced signal processing (MIMO precoding and combining) can be  performed at the common vertices (transmitters or receivers) to remove the interference.

According to~\cite{GhadikolaeiFM16}, the {\em protocol model} is a sufficiently accurate interference model for  directional mmwave communication.
Due to the high gain of directional antennas, interference from a single node is typically either negligible or destructive. 
Obviously, the protocol model with an arbitrary interference range can be translated into pairwise link interference. This justifies that the PI model is reasonably accurate and the NI model is optimistic.

\subsection{Optimization Goal}
Our paper focuses on downlink communication. With simple adaptation, the proposed algorithms can also be applied to an uplink or a joint uplink and downlink optimization. We will explain this in \S\ref{sec:extension}, after presenting the algorithms.
 For downlink scheduling, the only data sources are the macro BSs, and we can thus remove all incoming arcs to macro BSs in $D$. We can observe in Fig.~\ref{fig:sys_model} that there are many ways to schedule downlink communication among the macro and relay BSs. Our goal is to find a schedule that is of unit length and is optimal with respect to a QoS metric, while satisfying given QoS requirements and the constraints on simultaneous transmissions (number of RF chains at a node, radio duplexity, spatial multiplexing model and interference model).  In practice, the unit time length corresponds to the duration of a radio frame. 

\section{Optimal Maximum Throughput Fair Scheduling For A Full-Duplex Network Under the NI Model} 
\label{s:fd-sched-fair}

This section assumes a full-duplex backhaul network under the NI (No Interference) model.
Given these assumptions, we provide a polynomial-time optimal  algorithm for the maximum throughput fair scheduling problem (MTFS).

The goal of the MTFS problem is to maximize the downlink {\em network throughput} under the condition that the {\em max-min fairness}~\cite{Tang06, Tassiulas02} in throughput is achieved at the relay BSs.

\begin{mydef}[Maximum Throughput Fair Schedule] 
Given a directed network $D$ and a unit time schedule $S$, let $\vect{h}^S = [h^S_v | v \in M(D)]$ be
the {\em throughput vector} of $S$, where $h^{S}_v$ denotes the throughput of a relay BS $v$, that is the total
amount of data entering $v$ minus that leaving $v$ when $S$ is applied.

\noindent (i) A feasible unit time schedule $S_f$ is said to satisfy the
    max-min fairness criterion if
    $\min_{v \in M(D)} h_v^{S_f}\geq
    \min_{v \in M(D)}h_v^S$ for any feasible unit time schedule
    $S$. The value $\min_{v \in M(D)} h_v^{S_f}$ is called the {\em max-min throughput}.
  
\noindent (ii) A feasible unit time schedule $S^*$ is an optimal solution of the MTFS
    problem if $S^*$ satisfies the max-min fairness criterion in (i) and the network throughput $\sum_{v \in M(D)}h^{S^*}_v$ is maximum.
\end{mydef}

In the following, we present our generally applicable optimization method---{\em schedule-oriented optimization}.

\subsection{Schedule-Oriented Optimization}
\label{ss:sched-ori-opt}

To characterize the scheduled data streams in each timeslot, we need the definition of {\em simple $b$-matching} of a graph~\cite{Korte12}.

\begin{mydef} Let $G$ be an undirected graph with numbers $b: V(G) \mapsto \mathbb{N}$ and weights $c: E(G) \mapsto \mathbb{R}$, then a {\em simple $b$-matching} in $G$ is a function $f: E(G) \mapsto \{0, 1\}$ and $\sum_{e \in \delta(v)} f(e) \le b(v)$ for all $v \in V(G)$ where $\delta(v)$ is the set of edges incident to $v$. A {\em maximum weight simple $b$-matching} $f$ is a simple $b$-matching whose weight $\sum_{e \in E(G)} c(e) f(e)$ is maximum.
\end{mydef}

Let $b = [r(v) | v \in V(D)]$, then it is obvious the arc set scheduled in a timeslot is a simple $b$-matching of $D$. 
On the other hand, a simple $b$-matching of $D$ is an arc set that can be scheduled simultaneously in a timeslot.
The schedule-oriented optimization solves a linear optimization problem, the solution to which is exactly the optimal schedule. For the mathematical formulation, we construct the {\em node-matching matrix}.

\begin{mydef}[Node-matching Matrix] Given a directed network $D$, suppose the number of all possible simple $b$-matchings
  of $D$ is $K$ where $b = [r(v) | v \in V(D)]$. Then the node-matching matrix $\mat{A}=[a_{i,j}]$ is
  a $\card{V(D)} \times K$ matrix.
 Each element $a_{i,j}$  is equal to the 
  sum capacity of all arcs in the $j$-th simple $b$-matching that enter the $i$-th vertex of $D$ 
  minus the sum capacity of all arcs in the $j$-th
  simple $b$-matching that leave the $i$-th vertex of $D$.
  \label{def:node-matching-matrix}
\end{mydef}

\begin{figure}[!hbpt]
	\begin{minipage}{0.3\textwidth}
		\begin{figure}[H]
			\includegraphics[width=4cm]{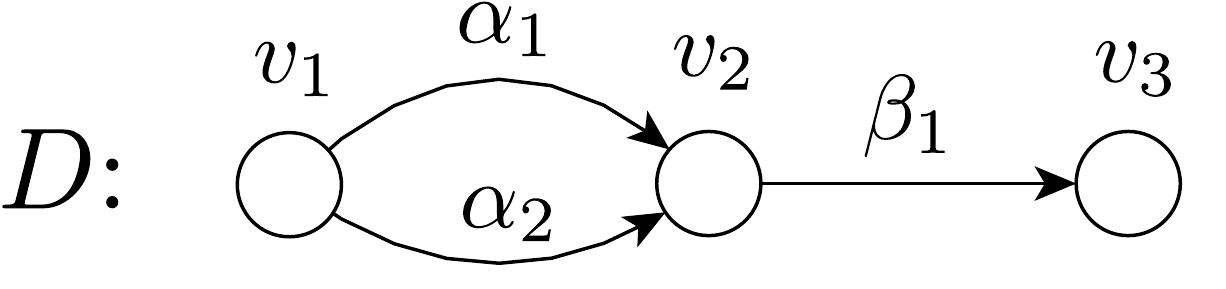}	
		\end{figure}
	\end{minipage} \\ \\ \\ 
	$\mat{A:}$
	\begin{minipage}{0.7\textwidth} \footnotesize
		$\begin{blockarray}{ccccccccccccccccc}
		& \alpha_1  & \alpha_2  & \beta_1  & \alpha_1,\alpha_2  & \alpha_1, \beta_1  & \alpha_2, \beta_1 \\
		\begin{block}{c[cccccccccccccccc]}
		v_1  & -8 & -6 & 0   & -14 & -8 & -6  \\
		v_2  & 8  & 6  & -3  & 14  & 5  & 3 \\
		v_3  & 0  & 0  & 3   & 0   & 3  & 3  \\
		\end{block}
		\end{blockarray}$
	\end{minipage}
	\caption{Node-$b$-matching matrix $\mat{A}$. Every node in the backhaul network has 2 RF chains.
		Let two links $\alpha = (v_1, v_2)$ and $\beta = (v_2, v_3)$.
		The maximum number of data streams of each link is $d(\alpha) = 2, d(\beta) = 1$. 
		The arc capacities in the directed network $D$ are: $c(\alpha_1) = 8$, $c(\alpha_2) = 6$, $c(\beta_1) = 3$.}
	\label{fig:node-matching}
\end{figure}

As we will see, the node-matching matrix helps to formulate the throughput constraints at each relay BS.
Fig.~\ref{fig:node-matching} gives an example of node-matching matrix for a directed network. 

Let $\mat{A}$ be the node-matching matrix of $D$. We define $\mat{A}^M$ as a submatrix of $\mat{A}$, which consists of the rows related to relay BSs. 
As the set of arcs scheduled in each timeslot of a schedule must be a simple $b$-matching in $D$, we define $\vect{t}^S$ as a $K\times 1$ {\em length vector}, each element of which is the length of a potential timeslot corresponding to a simple $b$-matching. Let the minimum throughput among all relay BSs be $\theta$. Then we can solve the MTFS problem in two steps: (i)
maximizing $\theta$; the solution $\theta^*$ is the
max-min throughput, and (ii)
computing the optimal schedule $S^*$ that offers the highest network throughput subject to the constraint $\theta \ge \theta^*$.

\textbf{Linear programs for MTFS.} 
The linear program to
maximize $\theta$ in step (i) is
\begin{IEEEeqnarray}{lrCl}
  \IEEEyesnumber\label{eq:mtf-theta}\IEEEyessubnumber*
  \textnormal{max} & \theta & &
  \\
  \textnormal{s.t.}\quad& \mat{A}^{M}\vect{t}^S& \geq & \vect{1}\theta\label{eq:mtf-theta1}
  \\
  & \trans{\vect{1}}\vect{t}^S& = &1\label{eq:mtf-theta2}
  \textnormal{ and }
   \vect{t}^S \geq \vect{0},
\end{IEEEeqnarray}
where $\vect{1}$ and $\vect{0}$ represent the all-one and all-zero
column vectors. The superscript `$\trans{}$' denotes the
vector transpose. \eqref{eq:mtf-theta1} is the constraint that the
throughput at each relay BS should be at least
$\theta$. \eqref{eq:mtf-theta2} is the constraint that the schedule
 should be of unit length. The feasibility of the schedule is
implicitly guaranteed by the formulation in terms of simple
$b$-matchings.
 
After we have obtained the solution $\theta^*$ from \eqref{eq:mtf-theta}, we can formulate the linear program that maximizes the  network throughput, under the condition that each relay BS has throughput at least $\theta^*$:
\begin{IEEEeqnarray}{lrCl}
  \IEEEyesnumber\label{eq:mtf}\IEEEyessubnumber*
  \textnormal{max} &\trans{\vect{c}}\vect{t}^S & &
  \\
  \textnormal{s.t.}\quad& \mat{A}^{M}\vect{t}^S& \geq & \vect{1} \theta^*
  \\
  & \trans{\vect{1}}\vect{t}^S& = &1
  \textnormal{ and }
   \vect{t}^S \geq \vect{0},
\end{IEEEeqnarray}
Here, $\vect{c}$ is the {\em capacity vector} whose element $c_j$ is
the cumulative capacity of all macro-BS-to-relay-BS data streams in the $j$-th
simple $b$-matching $M_j$, i.e.,
$c_j = \sum_{\{e | e \in M_j, \text{tail}(e) \in B(D)\}} c(e)$, where
$B(D)$ is the set of all macro BSs.

The difficulty in solving \eqref{eq:mtf-theta} and
\eqref{eq:mtf} is due to the huge number of elements in $\vect{t}^S$
(equal to the number of simple $b$-matchings of $D$, which is exponential in $\card{V(D)}$). 
Yet, we show that it is unnecessary to enumerate all of them, and both \eqref{eq:mtf-theta} and
\eqref{eq:mtf} can be solved in polynomial time.

\begin{thm}
\label{thm:MTFS-polynomial-time}
	Under the assumption of a full-duplex backhaul network and the NI model, the MTFS problem can be solved in polynomial time with the ellipsoid algorithm.
\end{thm}
\begin{proof}
  See Appendix~\ref{sec:proof_MTFS-polynomial-time} for the proof.
\end{proof}

Although polynomial, in practice the ellipsoid algorithm~\cite{Khachiyan80} almost always runs much slower than the {\em simplex algorithm}. Therefore, we propose algorithms based on the {\em revised simplex algorithm}~\cite{Dantzig55} which does not require the generation of all columns of $\mat{A}^M$. Conceptually, the algorithms first create a feasible schedule. In each iteration, to improve the optimization goal, we replace a timeslot in the schedule by another simple $b$-matching (a set of simultaneous data streams) while keeping the schedule feasible, until  the optimum is reached. The optimum is guaranteed to be reachable due to the correctness of the simplex algorithm in solving linear programs.
The maximum weight simple $b$-matching algorithm~\cite{Gabow83} is used to choose a better simple $b$-matching (column) to enter the schedule (basis). 

\subsection{Solving the MTFS Problem}
\label{sec:solve-mtfs}

To optimize $\theta$, we need an initial basic feasible solution to \eqref{eq:mtf-theta}. Suppose that each relay BS is reachable from at least one macro BS by following a sequence of arcs in $D$.
Let $D_1$ be a subgraph of $D$ such that only the first data stream of each link (the first arc of each set of parallel arcs in $D$)  is kept in $D_1$. 
We add a root vertex $v_r$ to $D_1$ and add an arc from $v_r$ to each macro BS vertex in $D_1$.
We perform a {\em breadth-first-search} (BFS) in $D_1$ starting from $v_r$. The result is a tree $T$ spanning $v_r$ and all BSs. 
Removing $v_r$ from $T$, we get a forest $T'$ that has exactly $\card{M(D)}$ arcs.
The initial schedule $S_0$ is constructed as follows: $S_0$ has $\card{M(D)}$ timeslots, each of which contains a different arc in $T'$. Moreover, it is required that the throughput of every relay BS is the same and the schedule takes unit time. This initial solution is unique. We convert the linear
program \eqref{eq:mtf-theta} to the standard form \eqref{eq:mtf-theta_std} by introducing $\card{M(D)}$
\emph{surplus variables} $s_i$.

\begin{IEEEeqnarray}{lrCl}
  \IEEEyesnumber\label{eq:mtf-theta_std}\IEEEyessubnumber*
  \textnormal{min} & &\trans{\vect{f}}\vect{x}&
  \\
  \text{s.t.}\quad& \mat{U}\vect{x}& = &\vect{g}
  \textnormal{ and }
   \vect{x} \ge \vect{0},
\end{IEEEeqnarray}
where
$
  \mat{U} \eqdef [\mat{U}^1|\mat{U}^2|\mat{U}^3]\eqdef
  \left[\begin{array}{c|c|c}
          \mat{A}^M & -\vect{1} &  -\mat{I}
          \\
          \trans{\vect{1}}  & 0 & \trans{\vect{0}}
        \end{array}\right],
$
$\trans{\vect{f}}=\bigl[\trans{\vect{0}} \,|\, {-1} \,|\, \trans{\vect{0}}\bigr]$,
$\trans{\vect{x}}\eqdef \bigl[\trans{(\vect{t}^S)}\,|\, \theta
\,|\,\trans{\vect{s}}\bigr]$, and
$\trans{\vect{g}}\eqdef\bigl[\trans{\vect{0}}\,|\,1\bigr]$. 
Alg.~\ref{alg:mtf-theta} shows the computation of the max-min throughput $\theta$.

\begin{algorithm}[!t]
\SetAlgoLined
Set the basis $\mat{B}$ according to the initial schedule $S_0$\;
\While{True} {
	Compute the dual variable $\trans{\vect{p}} = \trans{\vect{f}_{\mat{B}}} \mat{B}^{-1} $\;
	
	Set weight $w(e)$ to each arc $e = (v_i, v_j)_l \in E(D)$ where
	\begin{equation*}
	w(e) \eqdef
	\begin{cases}
	c(e) (p_j - p_i) & \text{if }v_i \in M(D)
	\\
	c(e) p_j         & \text{otherwise}.
	\end{cases} 
	\end{equation*}
	Do max weight simple $b$-matching on $D$ and let the max weight be $z$.  Compute $\eta_1 = -z - p_{\card{M(D)}+1}$\label{alg:mtf-a1}\;
Compute $\eta_2 = -1 + \sum_{k = 1}^{\card{M(D)}}p_k$\label{alg:mtf-a2}\;
Compute $\eta_3 = \min_{1 \le k \le \card{M(D)}} {p_k}$\label{alg:mtf-a3}\;
Compute $\eta = \min(\eta_1, \eta_2, \eta_3)$ and let the corresponding column be $\vect{u} \in \mat{U}$\;
\eIf {$\eta \ge 0$} {
\Return $\theta^* = \theta$ and $\mat{B}_{\theta^*} = \mat{B}$\;}
{Update $\mat{B}$ by replacing a column of $\mat{B}$ with $\vect{u}$ according to the simplex algorithm\;}
}
\caption{Compute the max-min throughput $\theta^*$}
\label{alg:mtf-theta}
\end{algorithm}

The basis $\mat{B}$ is a square matrix that consists of $\card{M(D)}+1$ columns from $\mat{U}$. $\vect{f}_{\mat{B}}$ are the elements of $\vect{f}$ corresponding to $\mat{B}$.
Lines \ref{alg:mtf-a1}, \ref{alg:mtf-a2} and \ref{alg:mtf-a3} compute the minimum reduced cost of a column in the matrices $\mat{U}^1, \mat{U}^2$ and $\mat{U}^3$ respectively. 
To decrease $-\theta$, we need to find a column of $\mat{U}$, $\vect{u}_{k}$ that has negative reduced cost $f_k - \trans{\vect{p}} \vect{u}_k < 0$ to enter the basis, according to the simplex algorithm. In each iteration of the algorithm, we find the column $\vect{u}$ in $\mat{U}$ that produces the minimum reduced cost $\eta$. If $\eta \ge 0$, then no columns can be used to decrease $-\theta$, thus we have reached the optimum. 

Let the max-min throughput be $\theta^*$ and the related basis be $\mat{B}_{\theta^*}$. To directly use $\mat{B}_{\theta^*}$ as the initial basis for the solution \eqref{eq:mtf}, we add an artificial scalar variable $y \ge 0$ to \eqref{eq:mtf} and replace the constraint $\mat{A}^M  \vect{t}^{S} \ge \vect{1} \theta^*$ with $\mat{A}^M  \vect{t}^{S} - \vect{1} y \ge \vect{1} \theta^* $. 
Since $\theta^*$ is the max-min throughput, the feasible $y$ must be 0. Hence, the optimal solution to \eqref{eq:mtf} is unaffected.
Again, we convert \eqref{eq:mtf} into the standard form of \eqref{eq:mtf-theta_std}, which is solvable with the revised simplex algorithm. 

In the standard form, $\mat{U}$ remains unchanged. We redefine
$\trans{\vect{f}} \eqdef \bigl[-\trans{\vect{c}} \,|\, 0 \,|\,
\trans{\vect{0}}\bigr]$, 
$\trans{\vect{x}}\eqdef \bigl[\trans{(\vect{t}^S)}\,|\, y
\,|\,\trans{\vect{s}}\bigr]$, and
$\trans{\vect{g}}\eqdef\bigl[\trans{\vect{1}}\theta^*\,|\,1\bigr]$.
The optimization algorithm is similar to
Alg.~\ref{alg:mtf-theta} and is outlined in Alg.~\ref{alg:mtf}. 
Since the basis $\mat{B}$ is a square matrix of dimension $\card{M(D)}+1$, it follows that the optimal schedule $S^*$ contains
no more than $\card{M(D)}+1$ timeslots. Additionally, since the links of a 
flow from a macro BS to a destination relay BS may not be scheduled in sequential
order, some transmission opportunities of the flow in the first few frames may be wasted. Therefore, maximum throughput is achieved in the long-term.

\begin{algorithm}[!t]
\SetAlgoLined
Set the basis $\mat{B} = \mat{B}_{\theta^*}$\;
\While{True} {
	Compute the dual variable $\trans{\vect{p}} = \trans{\vect{f}_{\mat{B}}} \mat{B}^{-1}$\;
	Set weight $w(e)$ to each arc $e = (v_i, v_j)_l \in E(D)$ where
	\begin{equation*}
	w(e) \eqdef 
	\begin{cases}
	c(e) (p_j - p_i) & \text{if }v_i \in M(D) \\
	c(e) (p_j + 1)   & \text{otherwise}.
	\end{cases}
	\end{equation*}
	Do max weight simple $b$-matching on $D$ and let the max weight be $z$. Compute $\eta_1 = -z - p_{\card{M(D)}+1}$\label{alg:mtf-b1}\;
	Compute $\eta_2 =  \sum_{k = 1}^{\card{M(D)}}p_k$\;
	Compute $\eta_3 = \min_{1 \le k \le \card{M(D)}} {p_k}$\;
	Compute $\eta = \min(\eta_1, \eta_2, \eta_3)$ and let the corresponding column be $\vect{u} \in \mat{U}$\;
	\eIf {$\eta \ge 0$} 
	{\Return the optimal schedule $S^*$ corresponding to $\mat{B}$\;}
	{Update $\mat{B}$ by replacing a column of $\mat{B}$ with $\vect{u}$\;}

}
\caption{Solving the MTFS problem}
\label{alg:mtf}
\end{algorithm}

\subsection{Reducing Simple $b$-matching to Matching}
To do maximum weight simple $b$-matching, we can either use dedicated algorithms such as~\cite{Gabow83}, or reduce it to a matching (equivalent to simple $1$-matching) problem~\cite{Edmonds65b}, for which highly efficient algorithms and implementations are available.
In this work, we use the state-of-the-art C++ implementation for {\em maximum weight perfect matching} on a general graph~\cite{Kolmogorov09}, where a perfect matching matches all vertices of a graph.
Thus, we need to reduce a maximum weight simple $b$-matching problem to a maximum weight perfect matching problem.

\subsubsection{Reduction for the case of MAX-SU-SM}
We use the reduction by Tutte~\cite{Tutte54}. Given a {\em strict graph}\footnote{A graph is {\em strict} if it has no loops or parallel edges between any pair of vertices.} $G$ whose edges have positive weights $\vect{w}$ and whose vertices have numbers  $b = \vect{r} \eqdef [r(v) | v \in V(G)]$, we create a graph $G'$ as follows. 
Each vertex $v \in V(G)$ is mapped to $r(v)$ vertices $v^{(1)} \dots v^{(r(v))}$.
Each edge $e = \{u, v\} \in E(G)$ with weight $w(e)$, is mapped to $r(u) \cdot r(v)$  edges $\big\{ \{u^{(i)}, v^{(j)}\} | \forall i, \forall j \big \}$ in $G'$, each of which is assigned the weight $w(e)$. 
Then a graph $G''$ is created by duplicating $G'$ and connecting each pair of symmetric vertices by an edge of zero weight.
An example of the reduction is shown in Fig.~\ref{fig:b-match-reduct}.
It is obvious that a maximum weight $b$-matching in $G$ can be deduced from a maximum weight perfect matching in $G''$.
An edge  $\{u^{(i)}, v^{(j)}\}$ in a matching is mapped to the edge $\{u, v\}$ in a $b$-matching.
\begin{figure}[!bpht]
	\centering
	\includegraphics[width=6cm]{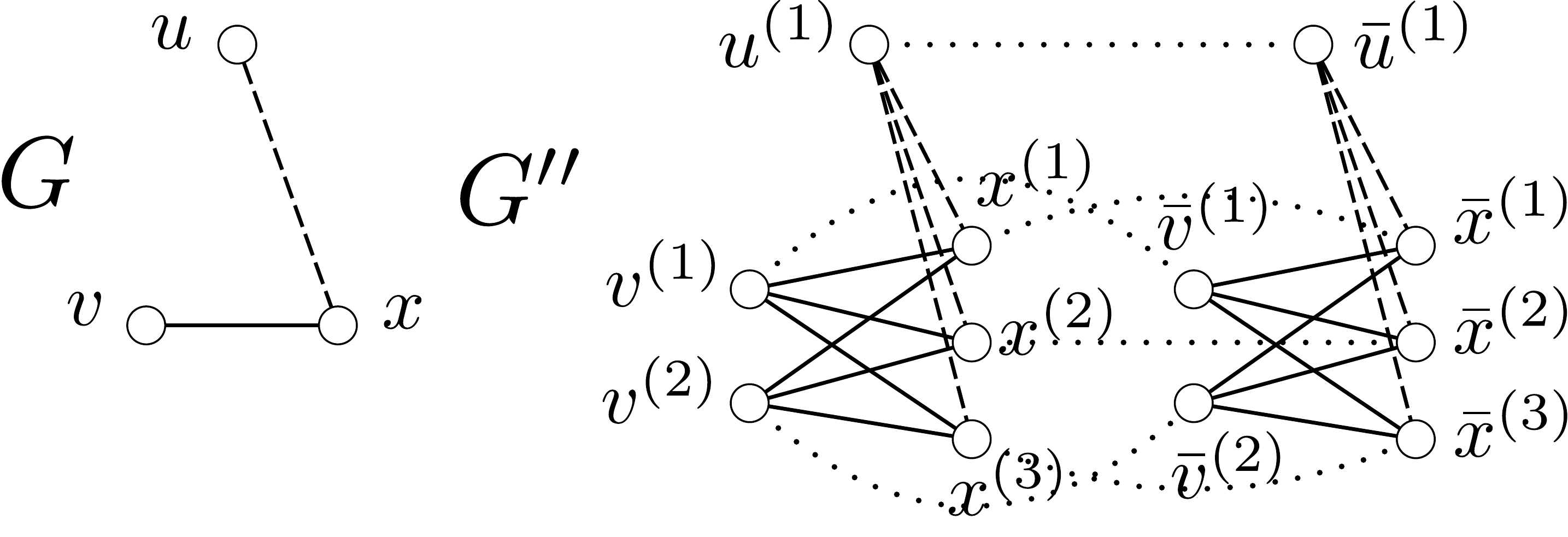}
	\caption{Reduction for the case of MAX-SU-SM. $u, v, x$ have 1, 2, 3 RF chains, respectively. The edges of the same style have the same weight. 
	Dotted edges have zero weights.}
	\label{fig:b-match-reduct}	
\end{figure}

\subsubsection{Reduction for the case of REAL-SU-SM}
We adapt the reduction in~\cite{Schrijver03} for this case. 
Given a multigraph $G$ whose edges have positive weights $\vect{w}$ and whose vertices have numbers $b = \vect{r} \eqdef [r(v) | v \in V(G)]$, we assume that $r(v) \le \deg(v)$ for any $v \in V(G)$ where $\deg(v)$ is the degree of $v$; otherwise we set $r(v) = \deg(v)$. We create a graph $G'$ as follows.
For each vertex $v$ in $G$, we create $r(v)$ vertices in $G'$, labeled as $v^{(1)} \dots v^{(r(v))}$. These vertices are called {\em outer vertices}.
For each edge $e = \{u, v\}$ in $G$, we add two {\em inner vertices} $e_u, e_v$ to $G'$, and $r(u) + r(v) + 1$ edges of weight $w(e)$, which are $\{u^{(i)}, e_u\}$,
 $\forall i$; $\{e_u, e_v\}$ and $\{e_v, v^{(j)}\}$, $\forall j$.
Then $G''$ is created by duplicating $G'$ and connecting each pair of symmetric outer vertices by an edge of zero weight.
An example of the reduction is shown in Fig.~\ref{fig:simple-match-reduct}.
It is obvious that a maximum weight simple $b$-matching in $G$ can be deduced from a maximum weight perfect matching in $G''$.
An edge $\{u^{(i)}, e_u\}$ in a matching is mapped to the edge $e$ in a simple $b$-matching.
\begin{figure}[!bpht]
	\centering
	\includegraphics[width=8cm]{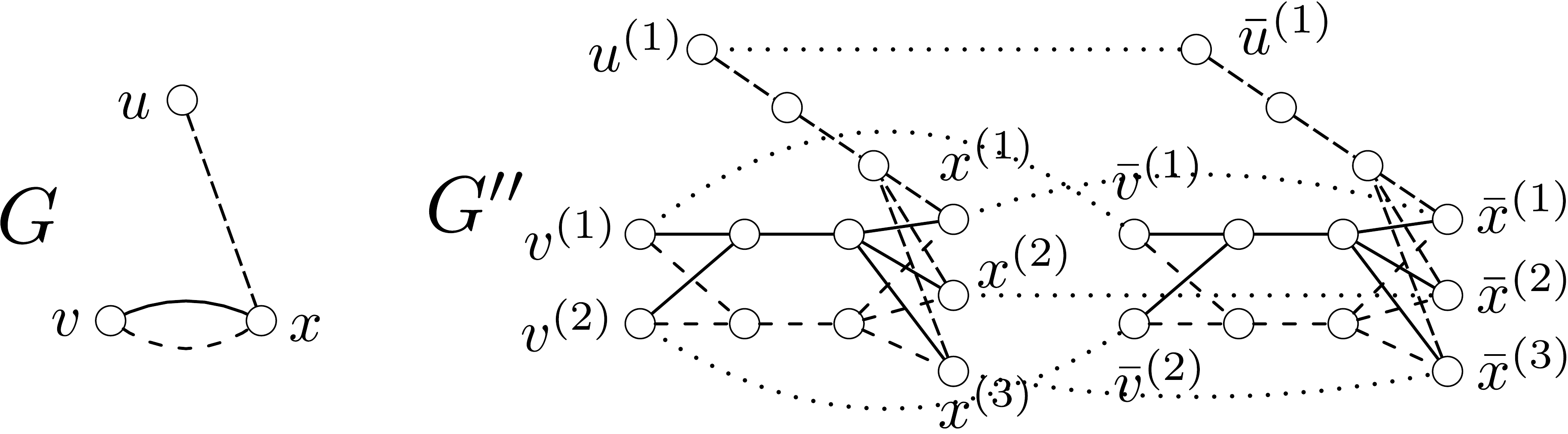}
	\caption{Reduction for the case of REAL-SU-SM. $u, v, x$ have 1, 2, 3 RF chains, respectively. The edges of the same style have the same weight. 
		Dotted edges have zero weights.}
	\label{fig:simple-match-reduct}	
\end{figure}

\section{Further Complexity Results On MTFS Scheduling}
\label{sec:complexity}

In this section, we will develop further computational complexity results on MTFS scheduling, for different duplexity modes, interference models and single-user spatial multiplexing (SU-SM) models.

\subsection{MTFS Under The Pairwise Link Interference Model is NP-hard}
\label{sec:mtfs-pi}
For both full-duplex and half-duplex scheduling, the MTFS problem is NP-hard if we assume an arbitrary pairwise link interference (PI) model.
\begin{thm}
   The MTFS problem is NP-hard under the pairwise link interference model for both half-duplex and full-duplex backhaul networks.
	\label{thm:mtfs-intf}
\end{thm}
\begin{proof}
	The proof can be done by relating the MTFS problem to computing the fractional chromatic number. See Appendix~\ref{sec:thm:mtfs-intf} for details.
\end{proof}

\subsection{Half-duplex MTFS is NP-hard}
\label{s:hd-mtfs}
Different from the polynomial-time solvable problem of full-duplex MTFS under the NI model, half-duplex MTFS is NP-hard, which will be proved in the following.
Then we will show a special case that allows a polynomial-time optimal solution to the half-duplex MTFS problem.
Since it is proved in \S\ref{sec:mtfs-pi} that half-duplex MTFS is NP-hard under the PI model, we assume the NI model in this subsection.

\subsubsection{Linear Programs for Half-duplex MTFS Problem}
Compared to the full-duplex case, the half-duplex scheduling has the additional half-duplex constraint---a node cannot work as transmitter and receiver simultaneously. Therefore, the matching-based optimization method that works successfully for the full-duplex scheduling cannot be applied directly.
For the half-duplex case, the data streams scheduled in each timeslot must be a {\em half-duplex subgraph} $J \subseteq D$ which is defined as follows.
\begin{mydef}[Half-duplex Subgraph] A half-duplex subgraph of a directed network $D$ is a subgraph $J \subseteq D$ such that (i) $J$ is a simple $b$-matching of $D$ with $b = [r(v) | v \in V(D)]$ and (ii) $J$ is a {\em directed bipartite graph}, i.e., $V(J)$ can be divided into two disjoint sets $V_1, V_2$ where each arc of $E(J)$ has the head in $V_2$ and the tail in $V_1$.
\end{mydef}
The constraint \emph{(i)} is due to the number of RF chains.
The constraint \emph{(ii)} reflects the half-duplex property, because the active nodes in a timeslot can be divided into the sender set $V_1$ and the receiver set $V_2$, where a data stream only goes from a sender to a receiver. 
Analogous to the node-matching matrix, we define the {\em node-hd-subgraph matrix} for the formulation of the half-duplex MTFS problem.
\begin{mydef} [Node-hd-subgraph Matrix] Given a directed network $D$, suppose that the number of all half-duplex subgraphs of $D$ is $K$. Then the node-hd-subgraph matrix $\mat{L} = [l_{i, j}]$ is a $|V(D)| \times K$ matrix. Denote the $i$-th vertex of $D$ as $v_i$, which is related to the $i$-th row of $\mat{L}$. 
	Denote the $j$-th half-duplex subgraph of $D$ as $J_j$, which is related to the $j$-th column of $\mat{L}$.
	Each element $l_{i, j}$ is equal to the sum capacity of all arcs in $J_j$ that enter $v_i$ minus the sum capacity of all arcs in $J_j$ that leave $v_i$.
\end{mydef}
\begin{figure}[!t]
	\begin{minipage}{0.12\textwidth}
		\begin{figure}[H]
			\includegraphics[width=2cm,left]{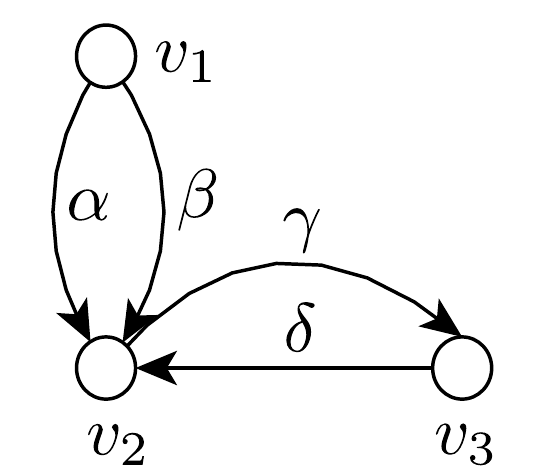}	
		\end{figure}
	\end{minipage} \hspace{-0.5cm}
	$\mat{L}$
	\begin{minipage}{0.3\textwidth} \footnotesize
		$\begin{blockarray}{cccccccc}
		& \alpha  & \beta  & \gamma   & \delta  & \alpha, \beta & \alpha, \delta & \beta, \delta   \\
		\begin{block}{c[ccccccc]}
		v_1  & -8 & -8 & 0   & 0  & -16 & -8  & -8 \\
		v_2  & 8  & 8  & -3  & 3  & 16  & 11  &  11 \\
		v_3  & 0  & 0  & 3   & -3 & 0   & -3  &  -3 \\
		\end{block}
		\end{blockarray}$
	\end{minipage}
	\caption{Node-hd-subgraph matrix $\mat{L}$. $v_1, v_2$ and $v_3$ have 2, 2 and 1 RF chain respectively. The arc capacities are: 
		$c(\alpha) = c(\beta) = 8$, $c(\gamma) = c(\delta) = 3$.}
	\label{node-hd-subgraph}
\end{figure}
Fig.~\ref{node-hd-subgraph} gives an example of the node-hd-subgraph matrix for a directed network. Similar to the definition of $\mat{A}^M$ in \S\ref{ss:sched-ori-opt}, $\mat{L}^M$ is the submatrix of $\mat{L}$ that only consists
of the rows related to relay BSs. 
The linear program formulation of the half-duplex MTFS problem is the same as \eqref{eq:mtf-theta} and \eqref{eq:mtf} except that $\mat{A}^M$ is replaced by $\mat{L}^M$.
Yet, different from the full-duplex MTFS problem, the half-duplex MTFS problem is NP-hard.
The intuitive reason is that the solution of these two problems requires computing a maximum weight simple $b$-matching and  a maximum weight half-duplex subgraph, respectively. The first can be done in polynomial time while the second is NP-hard. We will give a formal proof in the following.

\subsubsection{The Half-duplex MTFS Problem is NP-hard} \label{ss:hd-mtfs-np-comp}
\begin{mydef} [Maximum Weight Half-duplex Subgraph (MWHS) Problem] Given a directed network $D$ and a weight function $w(e)$ defined for each arc $e$, find a half-duplex subgraph $J \subseteq D$ such that $\sum_{e \in E(J)} w(e)$ is maximum.
\end{mydef}
Since the maximum weight simple $b$-matching problem can be solved in polynomial time, so can the full-duplex MTFS problem (see \S\ref{s:fd-sched-fair}). Analogously, if the MWHS problem could be solved in polynomial time, so could be the half-duplex MTFS problem.
Unfortunately, this is not the case. The MWHS problem is NP-hard for a directed network even if it is a directed acyclic graph (DAG). Furthermore, by extending the technique for proving NP-hardness of MWHS, we can prove that the half-duplex MTFS problem is also NP-hard.

\begin{lem} The MWHS problem is NP-hard for a directed network that is a DAG.
	\label{lem:MCHS-DAG}
\end{lem}
\begin{proof} The proof is done by reduction from the SAT problem. For further details, see Appendix~\ref{sec:proof-lem-MCHS-DAG}.
\end{proof}

\begin{thm}
	The half-duplex MTFS problem is NP-hard for a general directed network.
	\label{thm:HD-MTFS-NP-complete}	
\end{thm}
\begin{proof}
	See Appendix~\ref{sec:proof-thm-HD-MTFS} for the proof.
\end{proof}
In summary, assuming the NI model, the optimal scheduling problem of a mmWave backhaul networks can be solved in polynomial time when all BSs are full-duplex.
In contrast, the problem is NP-hard when all BSs are half-duplex, i.e., it is impossible to obtain an optimal schedule in polynomial time.

\subsection{Special Case: Half-duplex MTFS is Solvable in Polynomial Time} \label{sec:hd-mtfs-special-case}
We now study a special case of the half-duplex MTFS that is solvable in polynomial time. 
We refer to such backhaul networks as {\em uniform orthogonal backhaul networks}.

\begin{mydef}[Uniform Orthogonal Backhaul Network] A backhaul network that is represented by the directed network $D$ that satisfies the following  conditions:
	\begin{enumerate}
		\item There is no interference (NI model) between any pair of links.
		\item The MAX-SU-SM model is assumed for single-user spatial multiplexing.
		\item Each relay BS has the same number of RF chains, i.e., $r(v) \equiv r^M \in \mathbb{N}, \forall v \in M(D)$. 
		In addition, any macro BS has the RF chain number that is a multiple of $r^M$, i.e., for each $i=1 \dots |B(D)|$ and $n_i \in B(D)$,
		$r(n_i) = k_i \cdot r^M$, for some $k_i \in \mathbb{N}$.
	\end{enumerate}
\end{mydef}

Assuming the NI model, half-duplex MTFS is solvable in polynomial time if every node has a {\em single RF chain}, because the half-duplex constraint is automatically satisfied if every node can serve only one data stream.
Thus, the optimal schedule can be obtained with the optimal matching-based algorithm.
Next we will prove that half-duplex MTFS is also solvable in polynomial time for uniform orthogonal backhaul networks and provide an optimal algorithm.
First we look at the case that each node has the same number of RF chains $R$. In this case, the directed network $D$ is a multi-digraph, each arc of which belongs to a set of $R$ equivalent (same head, tail and capacity) parallel arcs. 
\begin{thm}
	Given a uniform orthogonal backhaul network $D$, each node of which has the same number of RF chains: $r(v) \equiv R$, $\forall v$, the half-duplex MTFS problem can be solved in polynomial time as follows:
	\begin{enumerate}
		\item Compute the optimal schedule $S$ with the matching-based optimal MTFS algorithm on $D$'s link network $L$ where the capacity of an arc in $L$ is the same as that of an arc in $D$ with the same head and tail, assuming each node has one RF chain.
		\item The optimal schedule $S^*$ consists of $R$ copies of $S$ running in parallel.
	\end{enumerate}
	\label{thm:hd-mtfs-uniform-rf}
\end{thm}
\begin{proof}
	See Appendix~\ref{sec:thm:hd-mtfs-uniform-rf} for the proof.
\end{proof}

We can now relax the condition that all nodes in $D$ have the same number of RF chains.

\begin{cor}
	Given a uniform orthogonal backhaul network $D$, the optimal solution to the half-duplex MTFS problem can be solved in polynomial time as follows. Assume that each relay BS has the same number of RF chains, $r(v) \equiv r^M \in \mathbb{N}, \forall v \in M(D)$. In addition, any macro BS has an RF chain number that is a multiple of $r^M$, i.e., for each $i = 1 \dots \card{B(D)}$ and $n_i \in B(D)$, $r(n_i) = k_i \cdot r^M$ with $k_i \in \mathbb{N}$.
	\begin{enumerate}
		\item
		Replace each macro BS vertex $n_i$ by $k_i$ vertices $n_i^{(1)} \dots n_i^{(k_i)}$,  each of which has $r^M$ RF chains. The connection of $n_i^{(j)}$ to the relay BSs is the same as that of $n_i$ (same number of arcs with the same head and capacity). Let the resulting directed network be $D'$.
		\item The optimal half-duplex MTFS schedule for $D$ is equivalent to that for $D'$, which is obtained with the algorithm in Theorem~\ref{thm:hd-mtfs-uniform-rf}.
	\end{enumerate}
	\label{cor:uni-rf-sched}
\end{cor}
\begin{proof}
	The optimal half-duplex schedule for $D$ is equivalent to the optimal one for $D'$, since a half-duplex schedule for $D$ can be translated into one for $D'$ and vice versa.
		Moreover $D'$ satisfies the condition of a uniform orthogonal backhaul network.
\end{proof}

\section{Approximation Algorithms based on Fractional Weighted Coloring}
\label{sec:app-fra-color}
As explained in \S\ref{sec:complexity}, the MTFS problem is NP-hard if there is pairwise link interference or the backhaul network is half-duplex. For both cases, we must rely on approximation algorithms.
Two such algorithms F$^3$WC-FAO and F$^3$WC-LSLO are based on the method {\em fractional weighted vertex coloring} of {\em conflict graphs}, as proposed by Wan	~\cite{Wan09}. The MTFS problem can be transformed into a fractional weighted vertex coloring problem since we can embody all four types of constraints in a conflict graph: 1) pairwise link interference, 2) number of data streams for a link restricted by spatial diversity, 3) number of data streams incident to a node restricted by the number of RF chains, and 4) half-duplex.

\subsection{Conflict Graph}\label{sec:conflict_graph}
Conflict graph is a powerful tool for modeling scheduling constraints. It is an undirected simple graph (it has neither loops nor parallel edges) denoted by $C$, in which each vertex represents a data stream and each edge represents that two data streams cannot be scheduled simultaneously. The conflict graph is derived from the {\em expanded network} $H$, a directed graph that explicitly models the RF chains. 
$H$ is the collective notation of four variants: $H_{\text{FD}}^\text{R}, H_{\text{FD}}^\text{M}, H_{\text{HD}}^\text{R}$ and $H_{\text{HD}}^\text{M}$, depending on the modeling.
Each vertex of $C$ is one-to-one mapped to each arc of $H$,  $V(C) = E(H)$.

\subsubsection{Full-Duplex Scheduling}
Let us first consider the most general backhaul network which is subjected to the PI model and the REAL-SU-SM model. Given a backhaul network $D$, we assume that the RF chain number $r(v) \le \deg(v), \forall v \in V(D)$. Otherwise, we set $r(v) = \deg(v)$ as the extra RF chains are redundant. The expanded network $H_{\text{FD}}^\text{R}$ is created by mapping each vertex $v \in V(D)$ into vertices $v^{(1)} \dots v^{(r(v))}$. 
Each arc $l_i = (u, v)_i \in E(D)$ (the $i$-th data stream from $u$ to $v$) is mapped to $r(u) \cdot r(v)$ arcs $\{(u^{(j)}, v^{(k)})_i | \forall j, k\}$, each with capacity $c(l_i)$. We define the {\em expanded arcs of a link} $(u, v)$ in $L$ as $X\big((u, v) \big) = \{(u^{(j)}, v^{(k)})_i |  \forall i, j, k\}$.
We define the {\em expanded arcs of a data stream} $(u, v)_i$ in $D$ as $X\big((u, v)_i\big) = \{(u^{(j)}, v^{(k)})_i | \forall j, k\}$. An example for the expanded network is shown in the middle of Fig.~\ref{fig:ex-net-real}.

\begin{figure*}[hbpt]
	\centering
	\includegraphics[width=13cm]{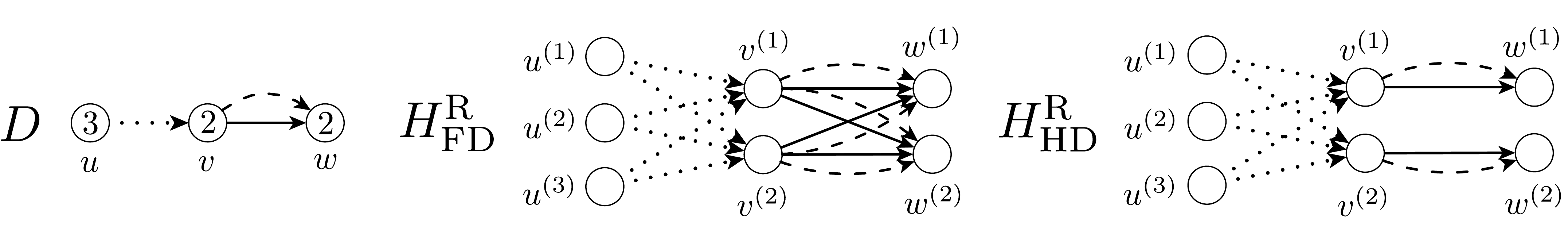}
	\caption{Transforming a directed network $D$ into an expanded network $H_{\text{FD}}^\text{R}$ and  $H_{\text{HD}}^\text{R}$ (full-duplex/half-duplex and REAL-SU-SM model).
		 The number inside a node is its number of RF chains. The arcs of the same style have the same capacity.}
	\label{fig:ex-net-real}	
\end{figure*}
The conflict graph $C$ is constructed as follows. Let a {\em complete graph} $K(V)$ be an undirected graph such that there is an edge between each pair of vertices in $V$. The edges of $C$ are constructed by first adding the union of the edge sets of a number of complete graphs. They are 1) the ones formed by the arcs incident to each vertex in $H_{\text{FD}}^\text{R}$, $K(\delta_{H_{\text{FD}}^\text{R}}(v)), \forall v$, and
2) the ones formed by the expanded arcs of each data stream $X(e), \forall e \in E(D)$.
Then we add the edges representing pairwise link interference. For each pair of interfering links in $L$, say $l$ and $l'$, we add to $C$ the edges $\big\{\{ e, e' \} |  e \in X(l),  e' \in X(l') \big \}$.

If the MAX-SU-SM model is assumed instead of the REAL-SU-SM model, then the expanded network $H_{\text{FD}}^\text{M}$ contains fewer arcs than $H_{\text{FD}}^\text{R}$.
Again each vertex $v \in V(D)$ is mapped into $r(v)$ vertices. If there is a link $(u, v) \in E(L)$ of capacity $c$,
then $D$ contains $\min(r(u), r(v))$ arcs (data streams) from $u$ to $v$ with the same capacity $c$. 
The link is mapped into $r(u) \cdot r(v)$ arcs $\{(u^{(j)}, v^{(k)}) | \forall j, k\}$ in $H_{\text{FD}}^\text{M}$, all having capacity $c$.
The {\em expanded arcs of a link} $(u, v)$ in $L$ are defined as $X\big((u, v) \big) = \{ (u^{(j)}, v^{(k)}) | \forall j, k \}$.
An example for the expanded network is shown in the middle of Fig.~\ref{fig:ex-net-max}.
\begin{figure*}[hbpt]
	\centering
	\includegraphics[width=13cm]{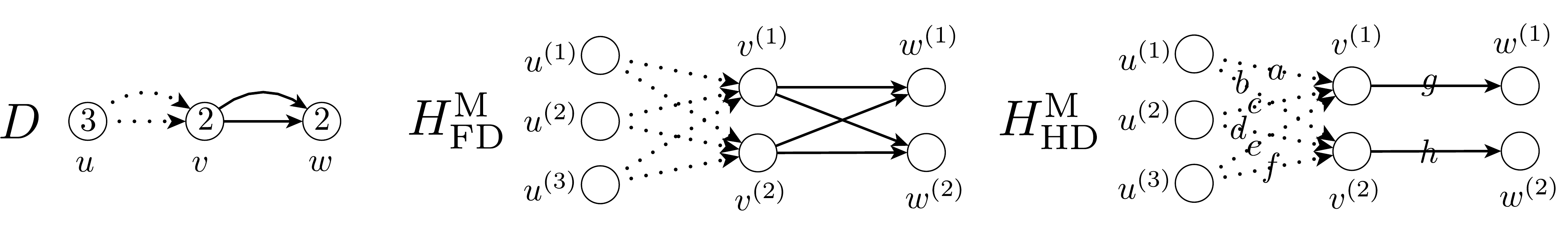}
	\caption{Transforming a directed network $D$ into an expanded network $H_{\text{FD}}^\text{M}$ and $H_{\text{HD}}^\text{M}$ (full-duplex/half-duplex and the MAX-SU-SM model).
		The number inside a node is its number of RF chains. The arcs of the same style have the same capacity.}
	\label{fig:ex-net-max}	
\end{figure*}

 We first add the edges of the complete graphs formed by the 
arcs incident to each vertex in $H_{\text{FD}}^\text{M}$. Then for each pair of interfering links in $L$, say $l$ and $l'$, we add to $C$ the edges $\big\{\{ e, e' \} | e \in X(l),  e' \in X(l') \big\}$.

\subsubsection{Half-Duplex Scheduling}
Again let us first consider the most general backhaul network subject to the PI model and the REAL-SU-SM model.
The expanded network $H_{\text{HD}}^\text{R}$ is more sparse than the full-duplex counterpart $H_{\text{FD}}^\text{R}$.
Given a backhaul network $D$, we set $r(v) = \min \Big(\max\big(\deg_-(v), \deg_+(v)\big), r(v) \Big)$ where $\deg_-(v)$ and $\deg_+(v)$ are the number of incoming and outgoing arcs of $v$ in $D$. The reason is that a higher number of RF chains is unnecessary.
$H_{\text{HD}}^\text{R}$ is created by first mapping the vertices in $D$ the same way as before.
Then each arc $l_i = (u, v)_i \in E(D)$ is mapped as follows. 
If $r(u) \ne r(v)$, $l_i$ is mapped to $r(u) \cdot r(v)$ arcs $\{(u^{(j)}, v^{(k)})_i | \forall j, k\}$, each with capacity $c(l_i)$ the same way as for $H_{\text{FD}}^\text{R}$. 
Otherwise, $r(u) = r(v)$, $l_i$ is mapped to $r(u)$ arcs $\{(u^{(j)}, v^{(j)})_i | \forall j\}$, each with capacity $c(l_i)$.
An example for the expanded network is shown in the right side of Fig.~\ref{fig:ex-net-real}.

For a vertex $v$ in $D$, we denote $\delta_{H_{\text{HD}}^\text{R}}^-(v)$ and $\delta_{H_{\text{HD}}^\text{R}}^+(v)$
as the arcs in $H_{\text{HD}}^\text{R}$ that enter or leave the vertices $v^{(j)}$ for all $j$, respectively.
The conflict graph $C$ is first constructed with the method for $H_{\text{FD}}^\text{R}$.
Then we add to $C$ the edges $\big \{\{e, e'\} | e \in \delta_{H_{\text{HD}}^\text{R}}^-(v),  e' \in \delta_{H_{\text{HD}}^\text{R}}^+(v), \forall v \in V(D)\big \}$. These edges model the half-duplex constraint.

If the MAX-SU-SM model is assumed, then the expanded network $H_{\text{HD}}^\text{M}$ is even more sparse than $H_{\text{HD}}^\text{R}$. 
Each link $(u, v) \in E(L)$ with capacity $c$ is mapped as follows.
If $r(u) \ne r(v)$, $(u, v)$ is mapped to $r(u) \cdot r(v)$ arcs $\{(u^{(j)}, v^{(k)}) | \forall j, k\}$, each with capacity $c$. Otherwise, it is mapped to $r(u)$ arcs $\{(u^{(j)}, v^{(j)}) | \forall j\}$, each with capacity $c$.
An example for the expanded network is shown in the right side of Fig.~\ref{fig:ex-net-max}.

\begin{figure}[hbpt]
	\centering
	\includegraphics[width=4cm]{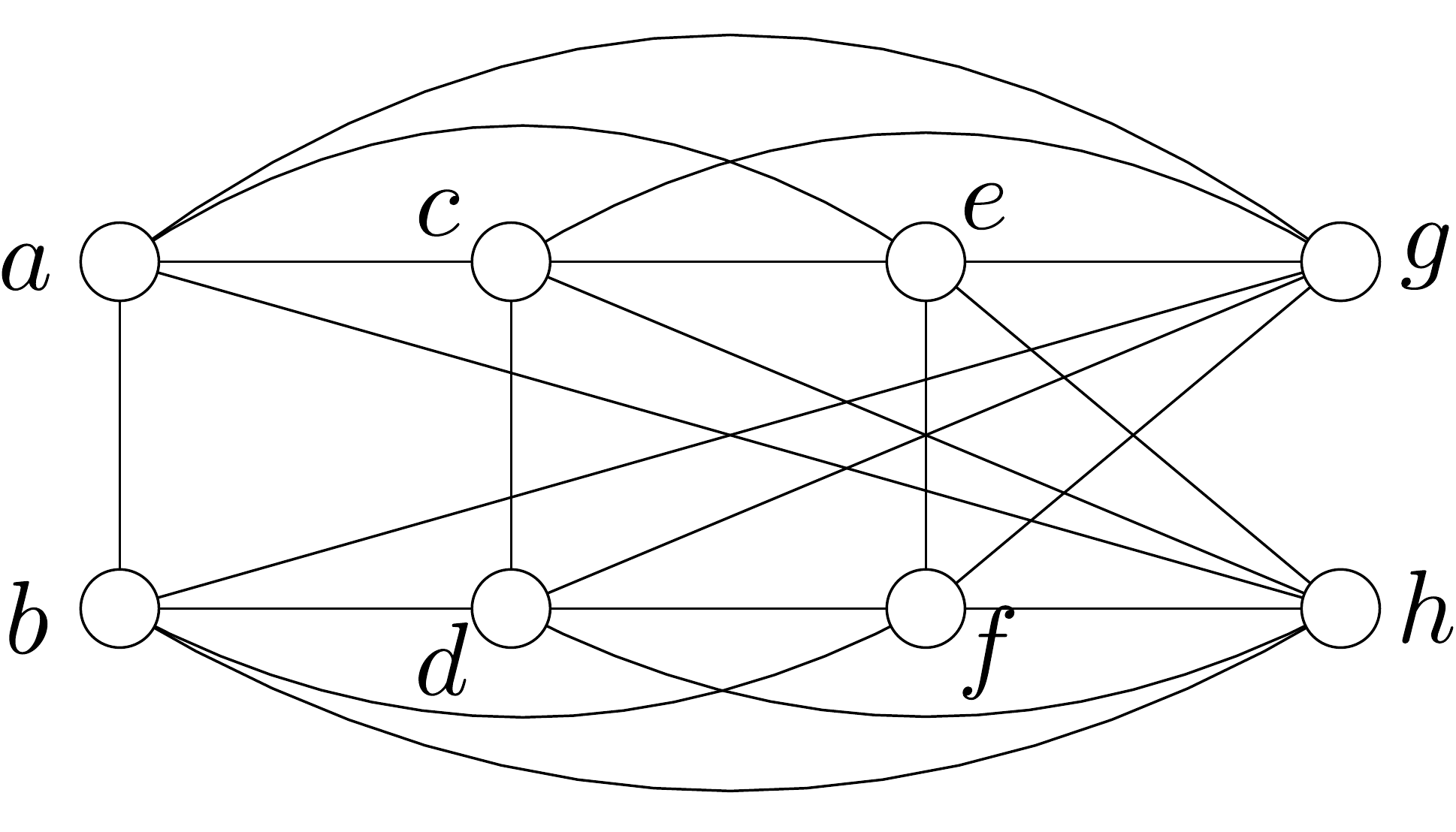}
	\caption{The conflict graph for $H_{\text{HD}}^\text{M}$ in Fig.~\ref{fig:ex-net-max}}
	\label{fig:conflict-graph}	
\end{figure}

The conflict graph $C$ is first constructed with the method for $H_{\text{FD}}^\text{M}$.
Then we add to $C$ the edges $\big\{\{e, e'\} | e \in \delta_{H_{\text{HD}}^\text{M}}^-(v),  e' \in \delta_{H_{\text{HD}}^\text{M}}^+(v), \forall v \in V(D) \big\}$. The conflict graph for $H_{\text{HD}}^\text{M}$ in Fig.~\ref{fig:ex-net-max} is shown in Fig.~\ref{fig:conflict-graph}.	

A sparse expanded network leads to a conflict graph with fewer vertices and hence shorter execution time for the algorithms. 
We will prove in the following why the sparse expanded networks $H_{\text{HD}}^\text{R}$ and $H_{\text{HD}}^\text{M}$ can be used for half-duplex scheduling.

\begin{thm}
	Given a directed network $D$, then any half-duplex subgraph of $D$ can be represented by a matching in $H_{\text{HD}}^\text{M}$.
	\label{thm:sparse-exp-net}
\end{thm}
\begin{proof}
	See Appendix~\ref{sec:thm:sparse-exp-net} for the proof.
\end{proof}

\subsection{General Procedure of Fractional Weighted Coloring Based Approximation Algorithms}
A fractional weighted coloring based approximation algorithm consists of three steps: (i) computing the {\em data stream time vector} 
$\vect{t} = [t_e | e \in E(H)] = [t_v | v \in V(C)]$, (ii) sorting the vertices $V(C)$ and performing F$^3$WC with the given ordering, and (iii) scaling the schedule.
We use the results of~\cite{Wan09} and adapt two approximation algorithms based on fractional weighted coloring. 
The difference of the two algorithms lies in the linear programs for computing the link time vector $\vect{t}$ and the ordering of $V(C)$ for coloring. The coloring step uses the so-called {\em first-fit fractional weighted coloring} (F$^3$WC) algorithm from~\cite{Wan09}, listed in Alg.~\ref{alg:f3wc}.

\begin{algorithm}[!t]
	\SetAlgoLined
	\SetKwInOut{Input}{Input}
	\SetKwInOut{Output}{Output}
	
	\Input{$C$, $\vect{t} \in \mathbb{R}_+^{V(C)}$, and an ordering of $V(C)$.}
	\Output{A fractional weighted coloring $\Pi$ of $(C, \vect{t})$.}
	
	$\Pi \leftarrow \emptyset$\;
	$U \leftarrow \{v \in V(C) | t_v > 0\}$\;
	
	\While{$U \ne \emptyset$} {
		$I \leftarrow$ the first-fit MIS (maximal independent set) of $U$\;
		$\lambda \leftarrow \min_{v \in I} t_v$\;
		add $(I, \lambda)$ to $\Pi$\;
		\For{each $v \in I$} {
			$t_v \leftarrow t_v - \lambda$\;
			\If{$t_v = 0$}
			{remove $v$ from $U$\;}
		} 
	}
	output $\Pi$\;
	\caption{First-fit fractional weighted coloring.}
	\label{alg:f3wc}
\end{algorithm}

How to compute $\vect{t}$ depends on the specific algorithm.
The {\em minimum makespan} scheduling for $\vect{t}$ is the same as the {\em minimum fractional weighted coloring} of $(C, \vect{t})$. The latter is defined as a set of $K \in \mathbb{N}$ pairs $(I_i, \lambda_i)$ where each $I_i$ is an independent set (a set of nonadjacent vertices) of $C$  and $\lambda_i \in \mathbb{R}_+$ for $1 \le i \le K$ satisfying that $\sum_{1 \le i \le K, v \in I_i} \lambda_i = t_v, \forall v \in V(C)$ and the sum $\sum_{i = 1}^K \lambda_i$ is the minimum. But the problem of finding a minimum fractional weighted coloring is NP-hard~\cite{Groetschel81}. Let $P$ be the {\em independence polytope} of $C$, i.e., the convex hull of the incidence vectors
of the independent sets of $C$. Then any point in $P$ corresponds to a feasible unit time schedule. The minimum fractional weighted coloring problem can be expressed as a linear program with the help of $P$.

We assume that an algorithm provides a $\gamma$-approximate ($\gamma > 1$) independent polytope $Q^\circ$, 
i.e., $Q^\circ \subseteq P \subseteq \gamma Q^\circ$. 
Specifically, we have two options---F$^3$WC-FAO or F$^3$WC-LSLO with $\gamma = \alpha^*, Q^\circ = Q$ 
and $\gamma = 2\beta^*, Q^\circ = Q'$, respectively (the definition of these variables will be clear in the following). Step (i) is to solve the following two linear programs.
\begin{IEEEeqnarray}{lrClr}
	\label{eq:q0app1}
	\IEEEyesnumber \IEEEyessubnumber* 
	\textnormal{max} & \theta & & & \\
	\textnormal{s.t.} &\sum_{\mathclap{e\in\delta_{H}^-(U(v))}} c(e) t_e
	-\sum_{\mathclap{e\in\delta_{H}^+(U(v))}} c(e) t_e& \ge & \theta
	\quad\forall\,v\in M(D)\IEEEeqnarraynumspace \label{eq:q0app-maxmin-tput}
	\\
	& \vect{t} & \in & Q^\circ \IEEEeqnarraynumspace \label{eq:q0app-in-q0},
\end{IEEEeqnarray}
where $U(v) = \{v^{(i)} \in V(H) | \forall i\}$.

With the max-min throughput solution $\theta$, we go on to compute $\vect{t}$ for the maximum network throughput.
\begin{IEEEeqnarray}{lC}
	\IEEEyesnumber \IEEEyessubnumber*
	\textnormal{max} &\sum_{\mathclap{e \in \{\delta_{H}^+(U(v)) | v \in B(D) \}}} c(e) t_e
	\nonumber\\
	\textnormal{s.t.}& \eqref{eq:q0app-maxmin-tput} \text{ and } \eqref{eq:q0app-in-q0}. \nonumber \IEEEeqnarraynumspace 
\end{IEEEeqnarray}
Step (ii) is to sort $V(C)$ with the given method and then to perform the F$^3$WC algorithm (Alg.~\ref{alg:f3wc}) 
with the computed $\vect{t}$ and vertex ordering. 
Since, it is guaranteed by step (i) and (ii) that the schedule length after performing the F$^3$WC algorithm is no more than one, we perform the last step to scale the schedule length to exactly unit time. The goal is to  improve performance by fully utilizing the available time resource.

\subsection{Fixed and Arbitrary Ordering (F$^3$WC-FAO)} \label{sec:hd-f3wc-fao}
Assume that $\langle v_1 \dots  v_n \rangle$ is an arbitrary but fixed ordering of $V(C)$ where $n = \card{V(C)}$.
We denote $v_i < v_j$ if $i < j$.
Let $V_i$ be the set of vertices of $v_i$ and all its {\em smaller} neighbors (neighbors in $\{v_1 \dots  v_{i-1}\}$).
Define the {\em inductive independence polytope} $Q$ of $C$ by the ordering $\langle v_1 \dots v_n \rangle$ as
\begin{equation}
Q \eqdef \Big \{ \vect{t} \in \mathbb{R}_+^{V(C)} \Big| \max_{1 \le i \le n} t(V_i) \le 1 \Big \},
\end{equation}
where $t(V_i) = \sum_{v \in V_i} t_v$. $Q$ is an approximation of the independence polytope $P$.

\subsection{Largest Surplus Last Ordering (F$^3$WC-LSLO)} \label{sec:hd-f3wc-lslo}
The largest surplus last ordering of $V(C)$ is done by first transforming the undirected graph $C$ into a directed graph 
$C^d$ by imposing a certain orientation on each edge. We specify the following orientation.

Suppose that the vertices of the directed network $D$ have an ordering. That is, given two different vertices $w, w' \in V(D)$, if $w$ comes before $w'$ in the ordering, we denote $w < w'$. 
Given two different vertices $u^{(i)}$ and $v^{(j)}$ of the expanded network $H$, 
we denote $u^{(i)} < v^{(j)}$ if and only if $u < v$ or ($u = v$ and $i < j$).
For the MAX-SU-SM model, given two different vertices  $(u^{(i)}, v^{(j)})$ and $(s^{(k)}, t^{(l)}) \in V(C^d)$, $(u^{(i)}, v^{(j)}) < (s^{(k)}, t^{(l)})$, if and only if $u^{(i)} < s^{(k)}$ or 
$(u^{(i)} = s^{(k)}$ and $v^{(j)} <  t^{(l)})$.
For the REAL-SU-SM model, given two different vertices $(u^{(i)}, v^{(j)})_m$ and $(s^{(k)}, t^{(l)})_n \in V(C^d)$,
$(u^{(i)}, v^{(j)})_m < (s^{(k)}, t^{(l)})_n$ if and only if $(u^{(i)}, v^{(j)}) < (s^{(k)}, t^{(l)})$ or
$\big( (u^{(i)}, v^{(j)}) = (s^{(k)}, t^{(l)})$ and $m < n \big)$.

The orientation is chosen according to the following rules for each edge in $C$. Note, the subscripts $m, n$ are taken as empty for the MAX-SU-SM model.
\begin{enumerate}
	\item An edge between $(u^{(i)}, v^{(j)})_m$ and a vertex of the form $(v^{(k)}, x^{(l)})_n$ 
	such that $u \ne x$ has the orientation from the first to the second.
	\item Otherwise, an edge between two vertices has the orientation from the small one to the large one.
\end{enumerate}

Let $D'$ be a digraph. 
For a vertex $u \in V(D')$, let $N^{in}(u)$ denote the set of in-neighbors of $u$ in $D'$, and let $N^{in}[u]$ denote $\{u\} \cup N^{in}(u)$.
$N^{out}(u)$ and $N^{out}[u]$ are defined correspondingly.
For any $\vect{t} \in \mathbb{R}_+^{V(D')}$, the {\em surplus} of a vertex $u$ is defined as $t(N^{in}(u)) - t(N^{out}(u))$.
The largest surplus last ordering is constructed as follows. Let $\vect{t} \in \mathbb{R}_+^{V(C^d)}$. Initialize $D'$ to $C^d$. For $i = n$ down to 1, let $v_i$ be a vertex of the largest surplus in $(D', \vect{t})$ and then delete $v_i$ from $D'$ and the element $t_{v_i}$ from $\vect{t}$.
The ordering of $\langle v_1 \dots  v_n \rangle$ is the largest surplus last ordering of $(C^d, \vect{t})$.
The {\em independence polytope} $Q'$ of $C^d$ is defined as 
\begin{equation}
Q' = \Big \{ \vect{t} \in \mathbb{R}_+^{V(C^d)} \Big | \max_{u \in V(C^d)} t (N^{in}[u]) \le 1/2 \Big \},
\end{equation}
which is another approximation of the independence polytope $P$.

\subsection{Approximation Ratios In Terms of Max-Min Throughput}
The following theorem presents the worst-case approximation ratios in terms of max-min throughput of the two algorithms F$^3$WC-FAO and F$^3$WC-LSLO.
\begin{thm}
	The algorithms F$^3$WC-FAO and F$^3$WC-LSLO solve the MTFS problem by producing a unit-time schedule.
	They achieve a max-min throughput $\theta' \ge \theta^*/\alpha^*$ and $\theta' \ge \theta^* / (2 \beta^*)$ respectively,
	where $\theta^*$ is the optimum and
	\begin{itemize}
		\item for the case of a full-duplex network, PI and REAL-SU-SM model:
		\begin{flalign*}
			\alpha^* &\le \max \Big( 1, \max_{l \in E(L)} \big( \sum_{l' | \mathrm{intf}(l', l) = 1} d(l') \big) \Big) + 2 &\\
			\beta^* &\le \max \Big( 1, \max_{l \in E(L)} \big( \sum_{l' |l' < l, \mathrm{intf}(l', l) = 1} d(l') \big) \Big) + 2,
		\end{flalign*}
		\item for the case of a full-duplex network, PI and MAX-SU-SM model:
		\begin{flalign*}
			\alpha^*&\le \max_{l \in E(L)} \big( \sum_{l' | \mathrm{intf}(l', l) = 1} d(l') \big) + 2 &\\
			\beta^* &\le \max_{l \in E(L)} \big( \sum_{l' |l' < l, \mathrm{intf}(l', l) = 1} d(l') \big) + 2,
		\end{flalign*}
		\item for the case of a half-duplex network and PI model:
		\begin{flalign*}
					\alpha^* &\le \max_{l \in E(L)} \big(r(l) +  \sum_{l' | \mathrm{intf}(l', l) = 1} d(l') \big) & \\
					\beta^* &\le \max_{l = (u, v) \in E(L)} \big(r(u) +  \sum_{l' | l' < l, \mathrm{intf}(l', l) = 1} d(l') \big) +1
		\end{flalign*}
		where $r(l) = r(u) + r(v)$.
	\end{itemize}
	\label{thm:perf-f3wc}
\end{thm}
\begin{proof}
	See Appendix~\ref{sec:thm:perf-f3wc} for the proof.
\end{proof}

\section{Approximation Algorithm of Parallel Data Stream Scheduling}
\label{sec:pds}
This section proposes an effective approximation algorithm for half-duplex MTFS scheduling under the NI model.
The PDS (Parallel Data Stream Scheduling) approximation algorithm (listed in Alg.~\ref{alg:pls}) extends the optimal half-duplex MTFS algorithm in \S\ref{sec:hd-mtfs-special-case} to cover the situation that an optimal MTFS schedule cannot be found in polynomial time. It is based on the idea that the parallel data streams between a pair of BSs are always scheduled simultaneously.
An example of the graph transformation step (Line~\ref{alg:graph-trans-begin} to \ref{alg:graph-trans-end}) of the PDS algorithm is shown in Fig.~\ref{fig:exp-PLS}.
\begin{algorithm}[htbp]
	Create a network $D^e$ based on the directed network $D$. 
	$D^e$ copies the relay BS vertices and the arcs between them from $D$ while keeping the values of RF chain number and capacity unchanged.
	Let the minimum data stream number of any link in the link network $L$ be $d_{min} = \min_{l \in E(L)} d(l)$. Each macro BS vertex $v$ in $D$ is mapped into $s(v)$ macro BS vertices $v^{(1)} \ldots v^{(s(v))}$ in $D^e$ where 
	\begin{equation}
	s(v) \eqdef \lfloor {r_D(v)} / {d_{min}} \rfloor.  
	\label{eq:eNB-split-PLS}
	\end{equation}
	The RF chain number of each macro BS vertex $v^{(i)}$ is defined as
	\begin{equation}
	r_{D^e}(v^{(i)}) \eqdef \begin{cases}
	d_{min} & \text{if } i < s(v)\\
	m(v) = r_D(v) - [s(v) - 1] d_{min} & \text{otherwise}.  \\ 
	\end{cases}
	\label{eq:r_De-PLS}
	\end{equation}
	For each $i$, create $d = \min(r_{D^e}(v^{(i)}), r_{D^e}(w), d(v, w))$ arcs $(v^{(i)}, w)_j$ in $D^e$
	such that $c_{D^e}((v^{(i)}, w)_j) = c_D((v, w)_j),$ $\forall j = 1 \dots d$, for each neighbor $w$ of $v$ in $D$\; \label{alg:graph-trans-begin}
	Make a copy $D^s$ of $D^e$ and replace each set of parallel arc with a single arc. For each arc $(u, v)$ in $D^s$, define the 
	capacity function associated with $D^s$ as $c_{D^s}((u, v)) \eqdef \sum_j c_{D^e}((u, v)_j)$.
	Define the RF chain number function associated with $D^s$ as $r_{D^s}(u) \eqdef 1, \forall u \in V(D^s)$\; 
	Compute the optimal MTFS schedule $S$ for $D^s$ with the capacity  $c_{D^s}$ and RF chain number  $r_{D^s}$ using the method in \S \ref{sec:solve-mtfs}\; \label{alg:graph-trans-end}
	Create the final schedule $S^*$ based on $S$, by mapping the activation of an arc in $D^s$ into 
	the simultaneous activation of parallel arcs in $D$\;
	\caption{PDS algorithm.}
	\label{alg:pls}
\end{algorithm}

\begin{figure}[hbpt]
	\centering
	\includegraphics[width=8.8cm]{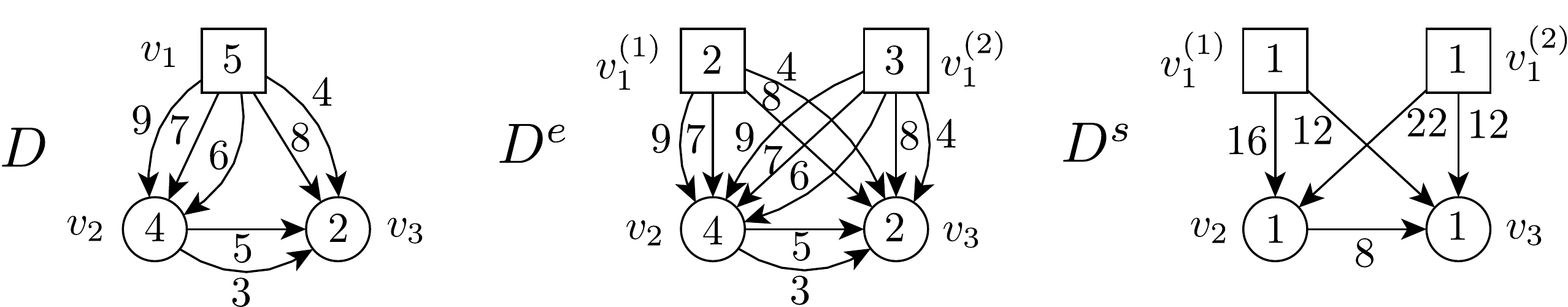}
	\caption{Example of the graph transformations in PDS. 
		$\Square$ and $\Circle$ represents macro BS and relay BS, respectively.
		The number inside a node $v$ is $r_G(v)$ and the number next to an arc $e$ is $c_G(e)$ where $G$ is the related graph.}
	\label{fig:exp-PLS}	
\end{figure}

\begin{thm}
	Suppose that the optimal max-min throughput of the half-duplex MTFS problem on a directed network $D$ is $\theta^*$ under the NI model.
	Let the max-min throughput obtained with the PDS algorithm be $\theta$ and $r_{min} = \min_{u \in V(D)}{r_D(u)}$ be the minimum RF chain number of any BS. 
	Let $r_{max}^M = \max_{u \in M(D)}{r_D(u)}$ be the maximum RF chain number of any relay BS,
	and $d_{min} = \min_{l \in E(L)} d(l)$ be the minimum data stream number of any link.
	We have $\theta \ge {\theta^*} / {\gamma^*}$, where 
	\begin{itemize}
		\item $\gamma^* = \max(r_{max}^M, \max_{v \in B(D)} m(v) ) \le \\ \max(r_{max}^M, 2d_{min} - 1)$, where $m(v)$ is defined in~\eqref{eq:r_De-PLS}, if the REAL-SU-SM model is assumed;
		\item $\gamma^* = \frac{\max(r_{max}^M, \max_{v \in B(D)} m(v) )} {r_{min}} \le \frac{\max(r_{max}^M, 2r_{min} - 1 )} {r_{min}}$, if the MAX-SU-SM model is assumed.
	\end{itemize}
	\label{thm:pls-perf}
\end{thm}
\begin{proof}
	See Appendix~\ref{sec:thm:pls-perf} for the proof.
\end{proof}

Let us consider some special cases of the MAX-SU-SM model.
If each relay BS in $D$ has the same RF chain number $r^M$ and each macro BS has an RF chain number that is a multiple of $r^M$, then the PDS algorithm attains the optimal MTFS schedule. On the other hand, if each relay BS has $r^M$ RF chains and any macro BS has at least $r^M$ RF chains, then PDS has a worst-case performance ratio of $1/2$ for the max-min throughput.

\begin{cor}
	Assume the NI model and the MAX-SU-SM model. 
	Given a directed network $D$, assume that each relay BS has the same number of RF chains $r^M$ and any macro BS has at least $r^M$ RF chains. The PDS algorithm achieves the max-min throughput $\theta > \theta^*/2$ where $\theta^*$ is the optimum for the half-duplex MTFS problem.
\end{cor}
\begin{proof}
	According to Theorem~\ref{thm:pls-perf}, $\theta \ge \frac{r^M}{2r^M-1} \theta^* > \theta^*/2.$
\end{proof}

In summary, under the NI model, the three algorithms: PDS, 
 F$^3$WC-FAO and F$^3$WC-LSLO are respectively $\frac{1}{\max(r_{max}^M, \max_{v \in B(D)} m(v))} \ge \frac{1}{r_{max}}$,
$1/\max_{l \in E(L)} (r(l))$ and $1/(2r_{max}+2)$-approximate algorithms for the half-duplex MTFS problem, where $r_{max} = \max_{v \in V(D)} r_D(v)$. A $\rho$-approximate ($\rho \le 1$) algorithm achieves a max-min throughput $\theta$ that is at least $\rho$ times that of the optimal value $\theta^*$, 
$\theta \ge \rho \theta^*$.
Theoretically, PDS has the best performance and F$^3$WC-LSLO has the worst.
\section{Extension To Integrated Access and Backhaul}
\label{sec:extension}
To date, 3GPP is investigating the standardization of Integrated Access and Backhaul (IAB) for mmWave cellular networks~\cite{3GPP-IAB}.
Yet designing a high-performance IAB network is still an open problem~\cite{PoleseGZRGCZ20, SahaD19}.
This paper offers joint routing and scheduling algorithms for IAB networks with optimal or guaranteed QoS.
Both the optimal algorithm and approximation algorithms proposed in this paper can readily be applied to the scenario of integrated backhaul and access (IAB) networks. 
Due to the graph-based network modeling, our approach is applicable to both IAB networks and backhaul networks.
 However, the runtime efficiency may be an issue, if the IAB network includes numerous user equipments (UE).
 
The proposed algorithms in this paper solve a downlink optimization problem. Yet with slight modification, they can solve an uplink or a joint uplink and downlink optimization problem. 
By doing so, the optimal algorithms still retain their optimality while the approximation algorithms keep their approximation ratios. 
A joint uplink and downlink optimization may use the resources better than two separate optimizations.
Conceptually, every algorithm in this paper has two parts, the routing part and the data stream conflict resolving part, which are performed either sequentially or intertwined. 
The routing part is a linear program that finds an efficient routing scheme for arbitrary throughput requirements on sources and destinations. So it naturally supports an uplink or a joint uplink and downlink optimization. 
The data stream conflict resolving part uses either the matching technique for the optimal algorithms or the conflict graph technique for the approximation algorithms.

In addition, our algorithms can be extended to solve other problems than MTFS. These include problems that can be formulated as a linear program whose variables are the active time of data streams and QoS metrics.
For example, we can optimize for the constraint that each relay BS has a minimum throughput requirement.
Another example is to optimize the energy consumption as it can be translated into the minimization of total transmission time in a schedule. We do not further elaborate on them as the extension is straightforward. 

\section{Numerical Evaluation} 
\label{s:eval}
In this section, we evaluate the proposed optimal and approximation algorithms for the MTFS problem in terms of max-min throughput, network throughput and execution time.

\subsection{Evaluation Setting}

We simulate an mmWave backhaul network, which consists of $n \times n$ relay BSs and $j \times k$ macro BSs.
The relay BSs are placed on the intersections of $n$ horizontal and $n$ vertical grid lines. The distance between two neighboring grid lines is $d^g$. The grid plane is divided into $j \times k$ equal rectangles and a macro BS is placed at each rectangle center (see Fig.~\ref{fig:backhaul-grid} for an example). We assume channel reciprocity in the simulation.
The capacity of a link is computed with the formula of Shannon capacity.
This is the value if one RF chain is used to serve the link on both ends.
 The received power is given by 
$p_{\text{rx}} = p_{\text{tx}} + g_\text{x} - PL$ where $p_{\text{tx}}$ is the transmission power, $g_\text{x}$ is the directivity gain and $PL$ is the path loss. 
We assume a carrier frequency of 28 GHz. The channel state of a link is simulated according to the statistical model derived from the real-world measurement~\cite{Akdeniz14}. There are three possible channel states---LOS (line-of sight), NLOS (non line-of-sight) or outage.
We only keep the links that are in LOS or NLOS state and have an SNR higher than $5$ dB.
The simulation parameters are listed in Tab.~\ref{tb:sim-param}.

\begin{figure}
	\centering
	\includegraphics[width=2.5cm,bb=0 0 20 20]{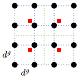}
	\caption{An example backhaul network consisting of $4 \times 4$ relay BSs and $2 \times 2$ macro BSs.}
	\label{fig:backhaul-grid}	
\end{figure}

\begin{table}
	\setlength\extrarowheight{2pt}
	\caption{Simulation parameters}
	\centering
	\footnotesize
	\label{tb:sim-param}
	\begin{tabular}{|p{3.8cm}| p{4.2cm} |}
		\hline
		\textbf{ Parameter }	 &	\textbf{ Value } 		\\
		\hline
		\hline
		Distance between 2 grid lines, $d^g$ 				& 80 m	\\ \hline
		Carrier frequency, $f$ & 28 GHz \\ \hline
		\multirow{2}{*}{\shortstack[l]
			{Path loss parameters\footnote{$\xi$ represents the shadowing effect.
					It is a normally distributed random variable with zero mean and
					$\sigma$ standard deviation.} $\alpha, \beta, \sigma$ in \\ $PL ( d )  = \alpha + 10 \beta \log _ { 10 } { d } + \xi$}} 	& LOS: $\alpha = 61.4, \beta = 2, \sigma = 5.8$ \\ 
		& NLOS: $\alpha = 72, \beta = 2.92, \sigma = 8.7$ \\ \hline
		Transmission power, $p_{\text{tx}}$						& 30 dB \\ \hline
		Directivity gain, $g_\text{x}$		& 30 dB \\ \hline
		Bandwidth, $b$		& 1 GHz \\ \hline
		Noise $N_0 = kT_0 +F + 10 \log_{10} b$ & $kT_0 = -174$ dBm/Hz, $F = 4$ dB \\ \hline
		Min SINR threshold for reception, $\tau$	& $5$ dB \\ \hline
		Number of data streams, $K$ &  $K \sim \max\{\textrm{Poisson}(\lambda), 1\}, \lambda = 1.8$  \\ \hline
		Beamwidth, $\phi$ & $\phi = 20^{\circ}$  \\ \hline
		Correlation coefficient in the exponential correlation matrix, $r$ & $r = 0.9$  \\ \hline
	\end{tabular}
\end{table}

For the PI model, we simulate the pairwise link interference according to the model in \S\ref{sec:sys-intf}. As illustrated in Fig.~\ref{fig:intf-links}, the 4 links $(t_1, r_1), (t_1, r_2), (t_2, r_1), (t_2, r_2)$ are assumed to be independent.

To simulate the REAL-SU-SM model, we assume that the maximum number of data streams supported by a link is Poisson distributed with the mean value $1.8$ (Tab.~\ref{tb:sim-param}), following the empirical model of~\cite{Akdeniz14}. 
The total capacity of a link increases sublinearly to the number of data streams and is simulated according to the exponential correlation matrix model in~\cite{Loyka01} by choosing the correlation coefficient $r = 0.9$. A comparison of the total capacity
of parallel data streams for the REAL-SU-SM and MAX-SU-SM models is shown in Fig.~\ref{fig:cap-max-real}.
\begin{figure}
	\centering
	\includegraphics[width=5cm]{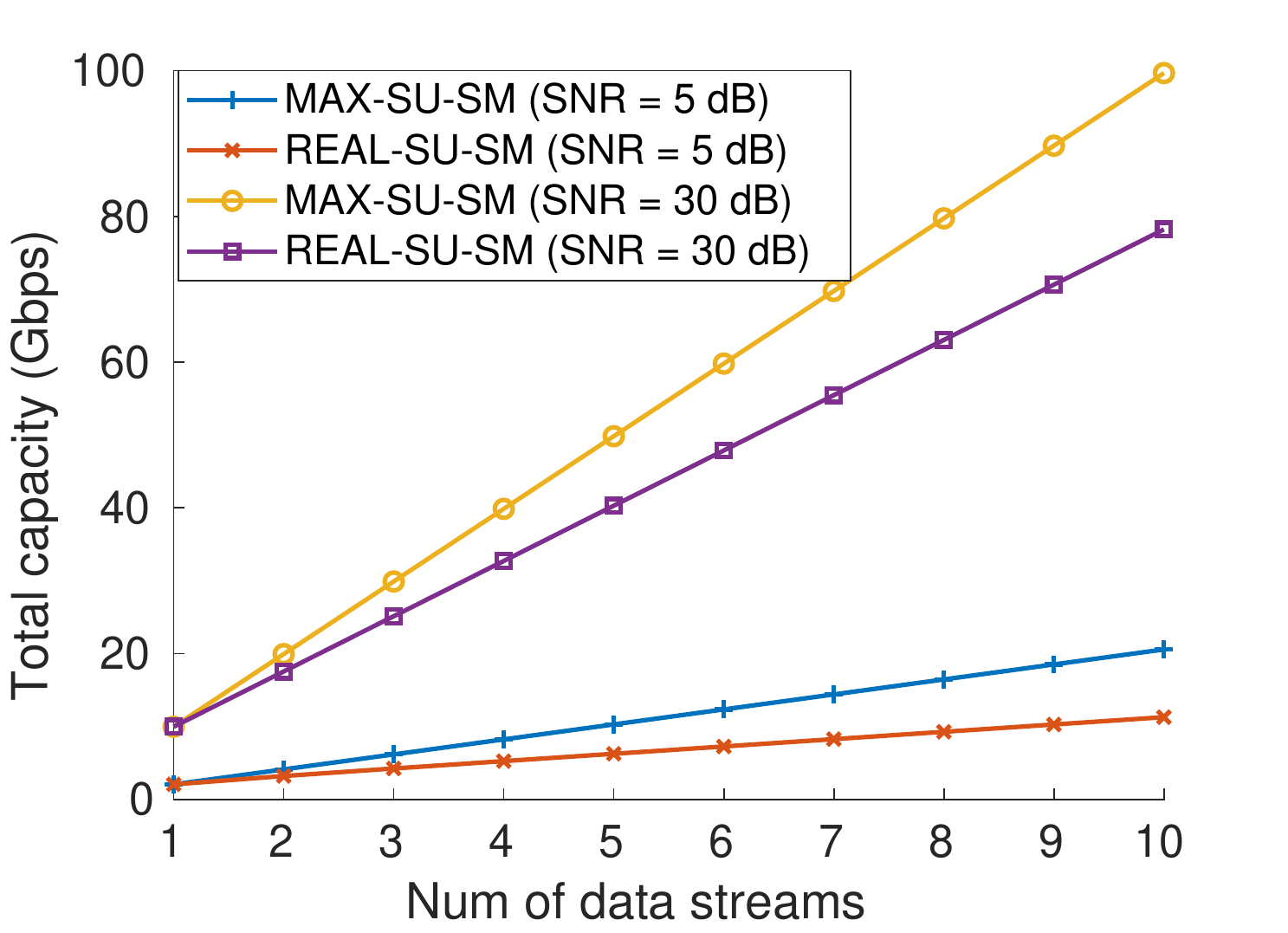}
	\caption{Comparison of the total capacity of parallel data streams for the REAL-SU-SM and MAX-SU-SM models.}
	\label{fig:cap-max-real}	
\end{figure}

\begin{figure}[!htbp]
	\begin{minipage}{0.13\textwidth}
		\begin{figure}[H]
			\includegraphics[width=2cm,left]{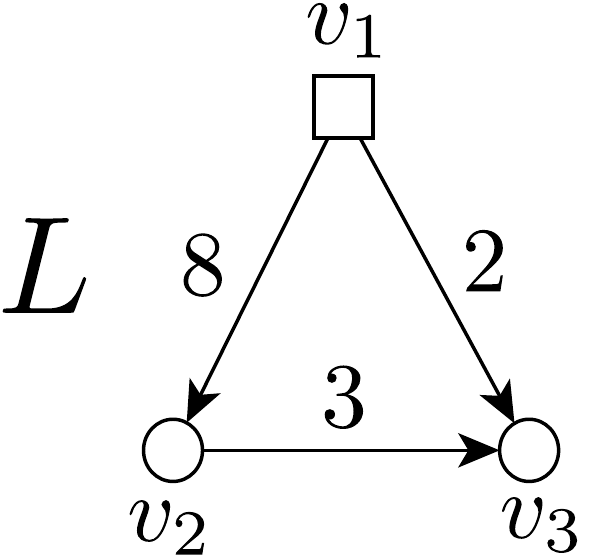}
		\end{figure}
	\end{minipage}
	\begin{minipage}{0.5\textwidth} \footnotesize
		\begin{tabular}{ l | c || c | c } \hline
			algo & OPT-FD-MTFS & \multicolumn{2}{c}{OPT-HD-MTFS} \\ \hline
			slot & \#1 & \#1 & \#2 \\ \hline
			\multirow{3}{*}{sched}  &  $v_1 \rightarrow v_2$  &  $v_1 \rightarrow v_2$  & $v_2 \rightarrow v_3$ \\
			&  $v_1 \rightarrow v_3$  & $v_1 \rightarrow v_2$  &  $v_2 \rightarrow v_3$ \\ 
			& $v_2 \rightarrow v_3$   & \\ \hline
			time & 1 & 0.4286 & 0.5714 \\
			\hline
		\end{tabular}
	\end{minipage}
	\caption{$v_1$ is the macro BS, and $v_2$ and $v_3$ are relay BSs. Each node has 2 RF chains. We assume NI and MAX-SU-SM models.
		The optimal max-min throughput for full-duplex MTFS and half-duplex MTFS problems are $5$ and $3.43$ respectively.}
	\label{fig:exp-diff-fd-hd}	
\end{figure}

\begin{figure*}[!t]
	\centering
	\subfigure[1 macro BS]{\includegraphics[width = 5.5cm]{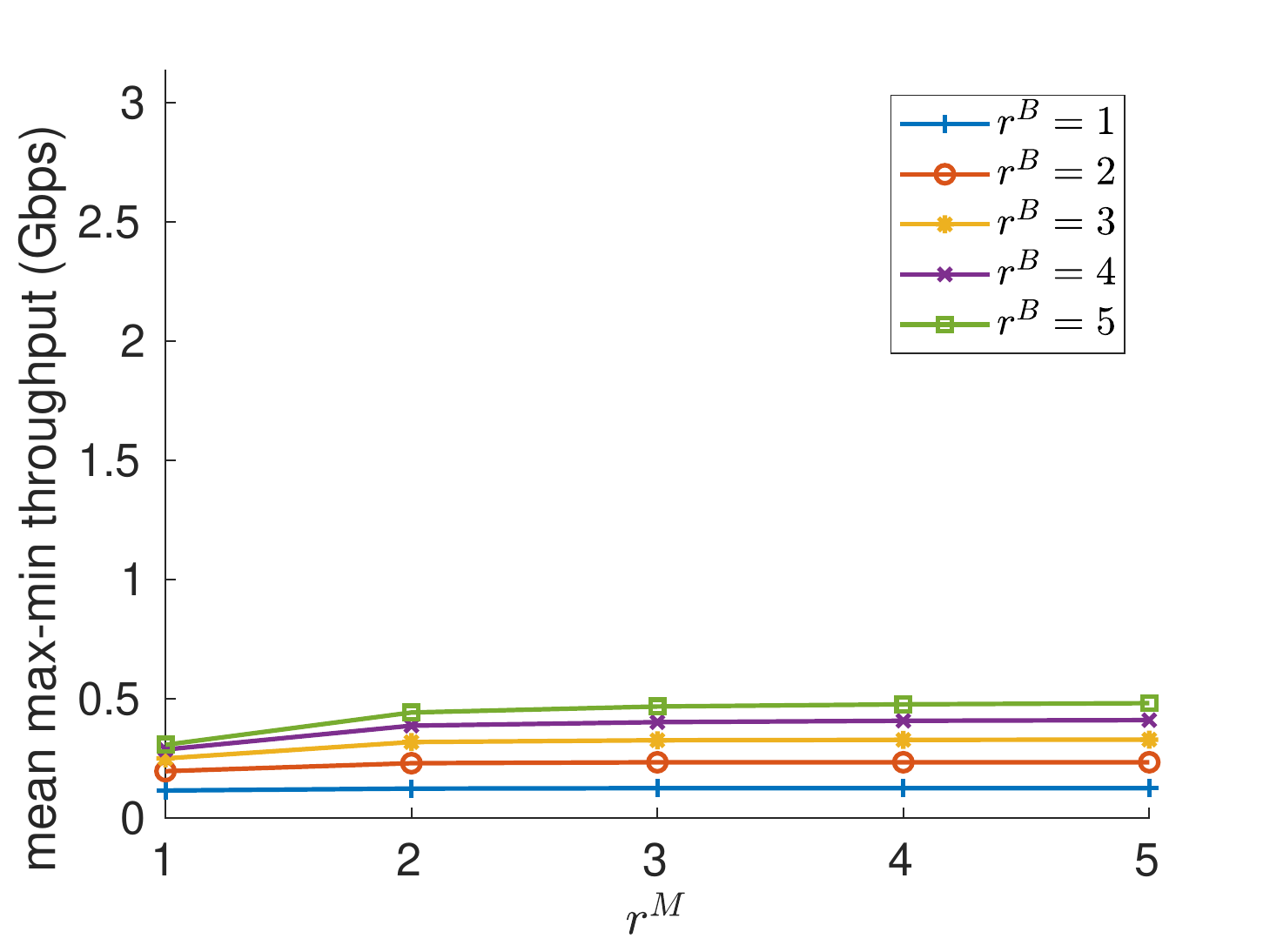}}
	\subfigure[2 macro BSs]{\includegraphics[width = 5.5cm]{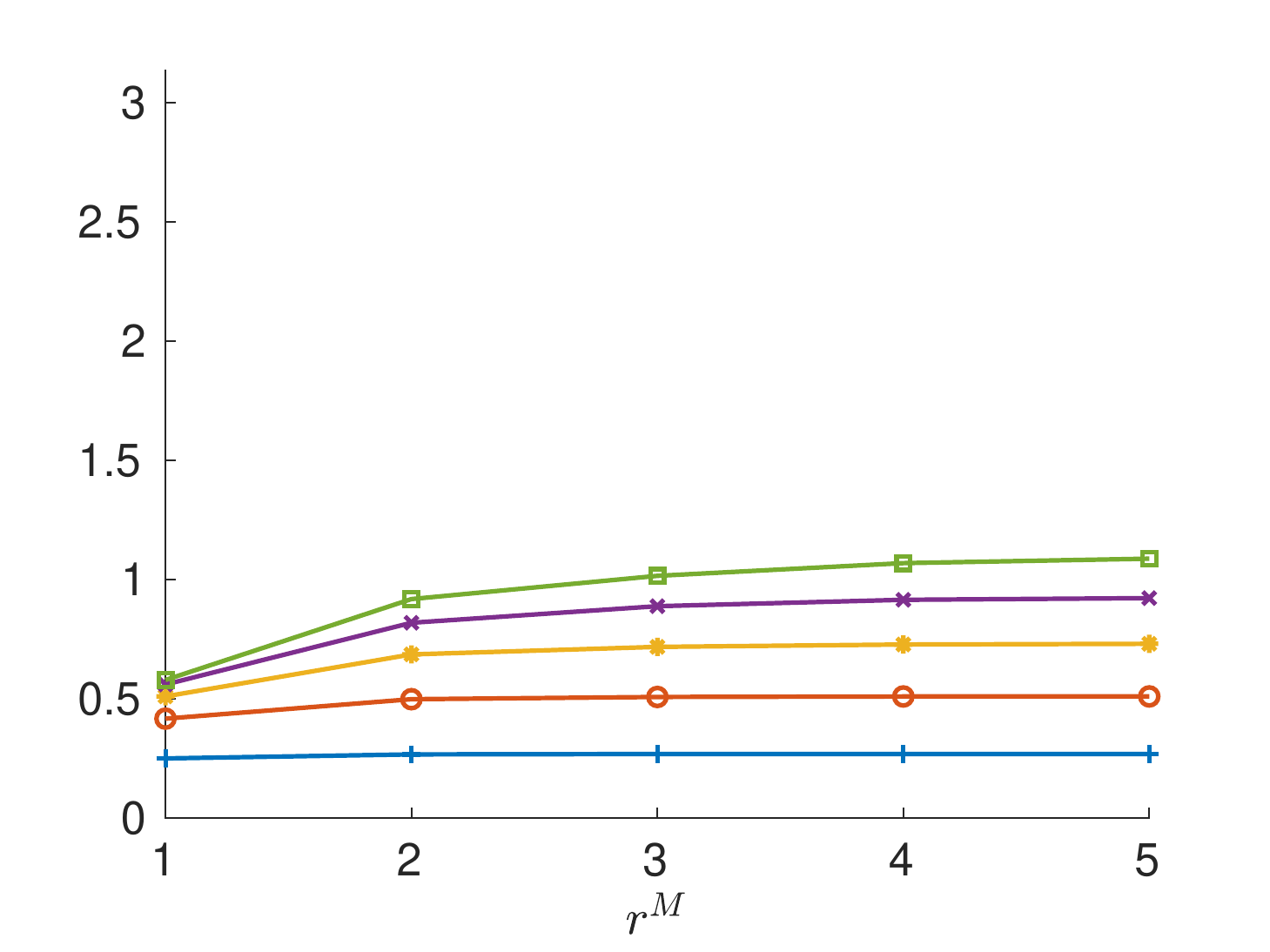}}
	\subfigure[4 macro BSs]{\includegraphics[width = 5.5cm]{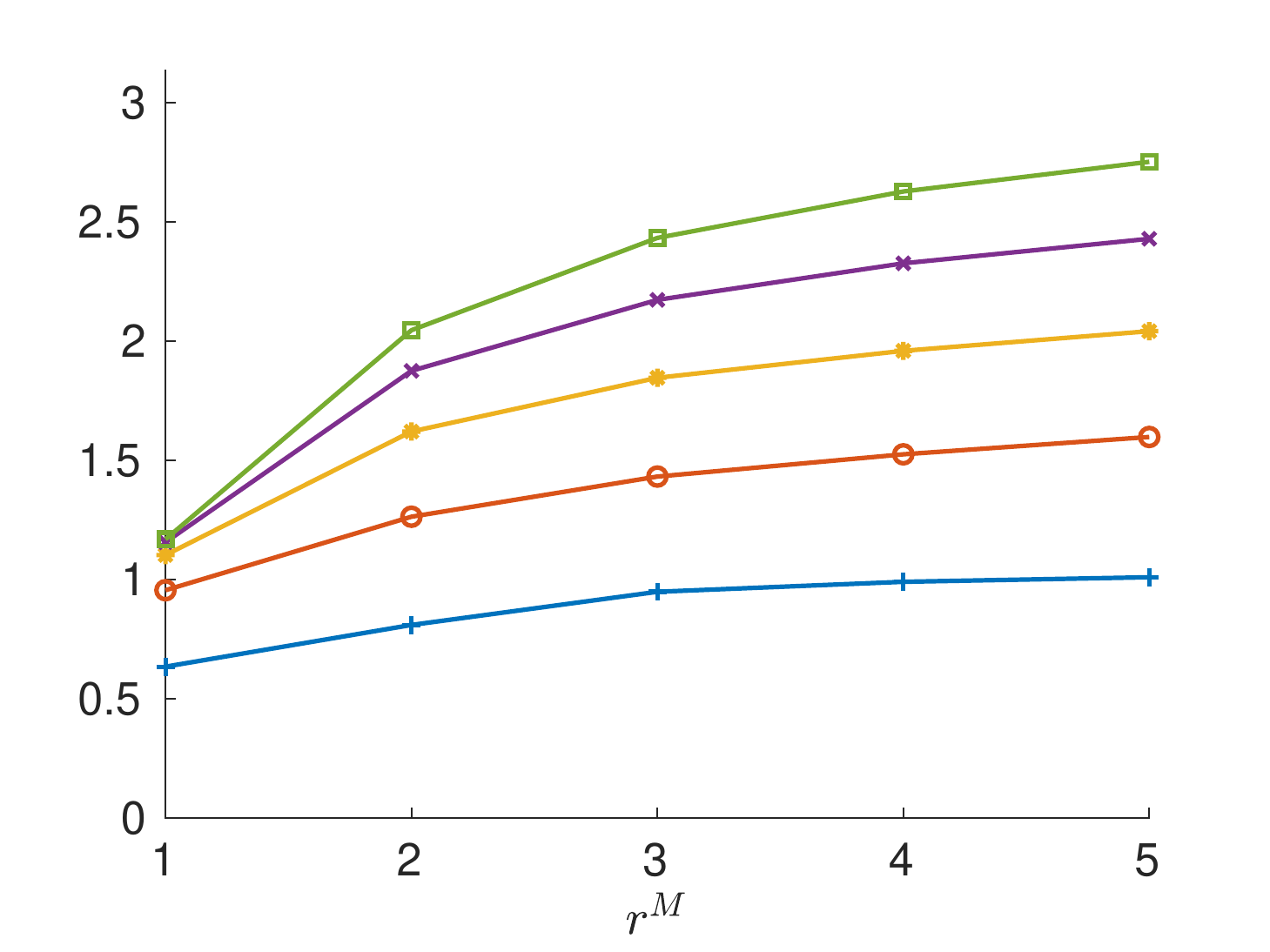}} 
	\caption{Max-min throughput of the OPT-FD-MTFS algorithm for the REAL-SU-SM model and for different number of RF chains.}
	\label{fig:opt-fd-mtfs-theta}
\end{figure*}

The proposed algorithms are implemented in MATLAB, except that we use the C++ program Blossom V for minimum cost perfect matching~\cite{Kolmogorov09} and Gurobi~\cite{gurobi} for linear programming.

We evaluate the optimal algorithms OPT-FD-MTFS and OPT-HD-MTFS as well as three approximation algorithms---F$^3$WC-FAO, F$^3$WC-LSLO, and PDS, for 10 backhaul networks with $10 \times 10$ relay BSs. 
OPT-FD-MTFS works for full-duplex scheduling under the NI model while OPT-HD-MTFS works for half-duplex scheduling of uniform orthogonal backhaul networks. F$^3$WC-FAO and F$^3$WC-LSLO are generally applicable for any combination of half-duplex/full-duplex, NI/PI model and MAX-SU-SM/REAL-SU-SM model while PDS only works for half-duplex scheduling under the NI model.

We place 1, $2 \times 1$ or $2 \times 2$ macro BSs in each network. The macro BSs and the relay BSs have the same number of RF chains $r^B$ and $r^M$ respectively, while $r^B$ and $r^M$ range from 1 to 5.

\begin{figure}
	\centering
	\subfigure[OPT-FD-MTFS]{\includegraphics[width = 4.3cm]{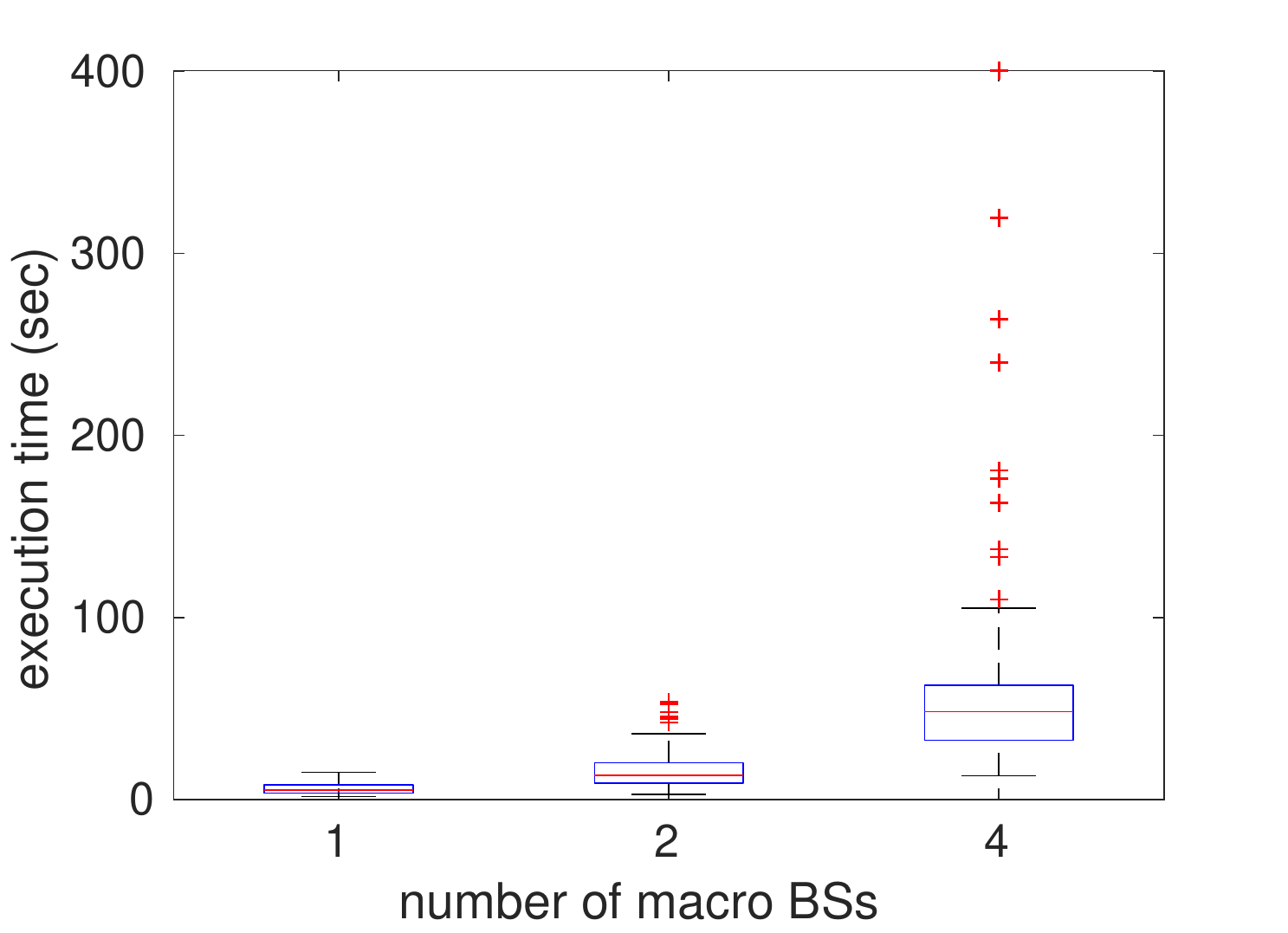}}
	\subfigure[OPT-HD-MTFS]{\includegraphics[width = 4.3cm]{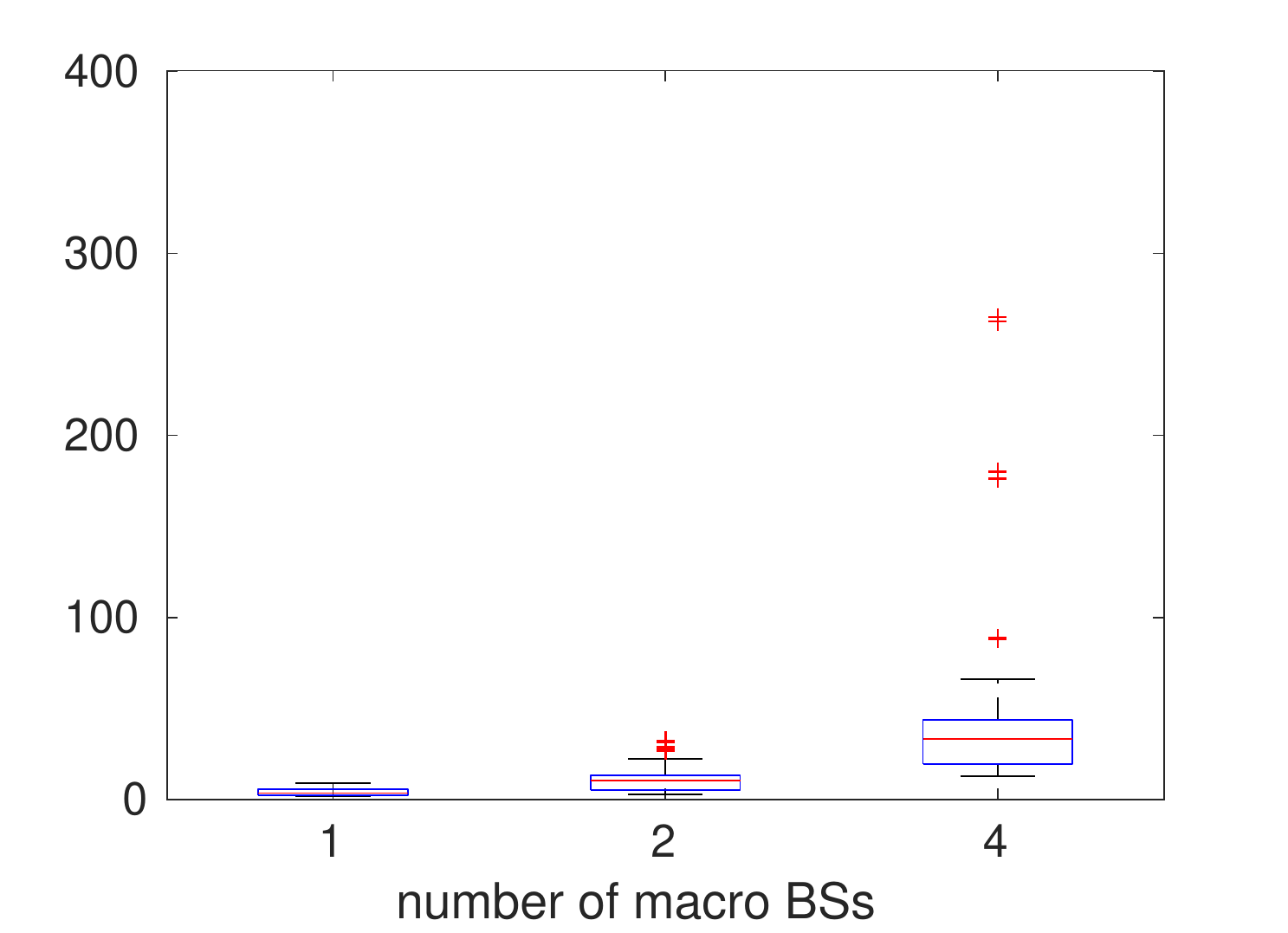}}
	\caption{Execution time of optimal algorithms for uniform orthogonal backhaul networks.}
	\label{fig:fd-hd-exetime}
\end{figure}

\subsection{Optimal Algorithms}
\label{ss:opt-algo}
Both full-duplex and half-duplex optimal schedules can be computed efficiently for uniform orthogonal backhaul networks. 
Surprisingly, for such networks, the max-min throughput of both OPT-HD-MTFS and OPT-FD-MTFS schedules are usually the same. 
 We believe that the close performance of max-min throughput for both half-duplex and full-duplex scheduling is due to the good connectivity of the backhaul network which allows plenty of scheduling possibilities. A simple network in Fig.~\ref{fig:exp-diff-fd-hd} shows that the performance gap can be large.

The max-min throughput $\theta^*$ of OPT-FD-MTFS for the REAL-SU-SM model is shown in Fig.~\ref{fig:opt-fd-mtfs-theta} for various number of macro BSs and RF chains. 
Generally, $\theta^*$ increases with the number of RF chains of relay BS ($r^M$) and of macro BS  ($r^B$), as well as the number of macro BSs. 
When the number of macro BSs and $r^B$ are fixed, $\theta^*$ gradually saturates despite the increase of $r^M$. In such cases, the bottleneck is at the links between macro BSs and relay BSs.
To achieve higher performance in $\theta^*$, we need to increase all three variables.
Yet, adding macro-BSs would be very costly. Adding more RF chains to each macro-BS while increasing the relay BSs that are neighbors to these macro BSs seems like a more cost-effective approach.
Moreover, the average throughput per relay BS is from 1x to 1.96x of the max-min throughput. This shows that in a dense network, we can achieve a rather equal distribution of throughput among relay BSs. 
As expected, $\theta^*$ of OPT-FD-MTFS for the MAX-SU-SM model is greater than or equal to that of the REAL-SU-SM model. The difference increases with $r^M$ and $r^B$ (Tab.~\ref{tab:opt-fd-mtfs-real-max}), which shows that multiple RF chains are especially beneficial to a rich multi-path channel.
\begin{table}[htbp]
	\caption{The rate of $\theta^*$ of MAX-SU-SM to that of REAL-SU-SM.}	
	\begin{tabular}{ c|c|c|c|c|c } 
			Avg. rate & $r^M = 1$ & $r^M= 2$ & $r^M = 3$ & $r^M= 4$ & $r^M= 5$ \\ \hline
			$r^B = 1$& 1.00 & 1.00 & 1.00 & 1.01 & 1.00 \\ \hline
			$r^B = 2$& 1.00 & 1.01 & 1.03 & 1.06 & 1.08 \\ \hline
			$r^B = 3$& 1.00 & 1.01 & 1.05 & 1.08 & 1.10 \\ \hline
			$r^B = 4$& 1.00 & 1.02 & 1.08 & 1.11 & 1.14 \\ \hline
			$r^B = 5$& 1.00 & 1.03 & 1.10 & 1.14 & 1.17 
	\end{tabular}
	\label{tab:opt-fd-mtfs-real-max}
\end{table}

The distributions of execution time of OPT-FD-MTFS and OPT-HD-MTFS for uniform orthogonal backhaul networks are shown in Fig.~\ref{fig:fd-hd-exetime}. OPT-HD-MTFS achieves almost the same performance in max-min throughput and network throughput as OPT-FD-MTFS, yet it runs much faster than the latter, by shortening the execution time by 27\% on average and by 79\% in the best case. The reason is due to the step of merging RF chains in the OPT-HD-MTFS (same as PDS) algorithm which leads to a smaller (in terms of vertices and arcs) graph on which matching is performed.  
Recall that in general cases, the HD-MTFS problem is NP-hard.
In addition, we observe from Fig.~\ref{fig:fd-hd-exetime} that the execution time increases with the number of macro BSs for both algorithms. In addition, the execution time of OPT-FD-MTFS also increases with the number of RF chains at BSs due to the growth of the graph for matching.

\subsection{Full-Duplex Approximation Algorithm}

\begin{figure}
	\centering
	\includegraphics[width=7cm]{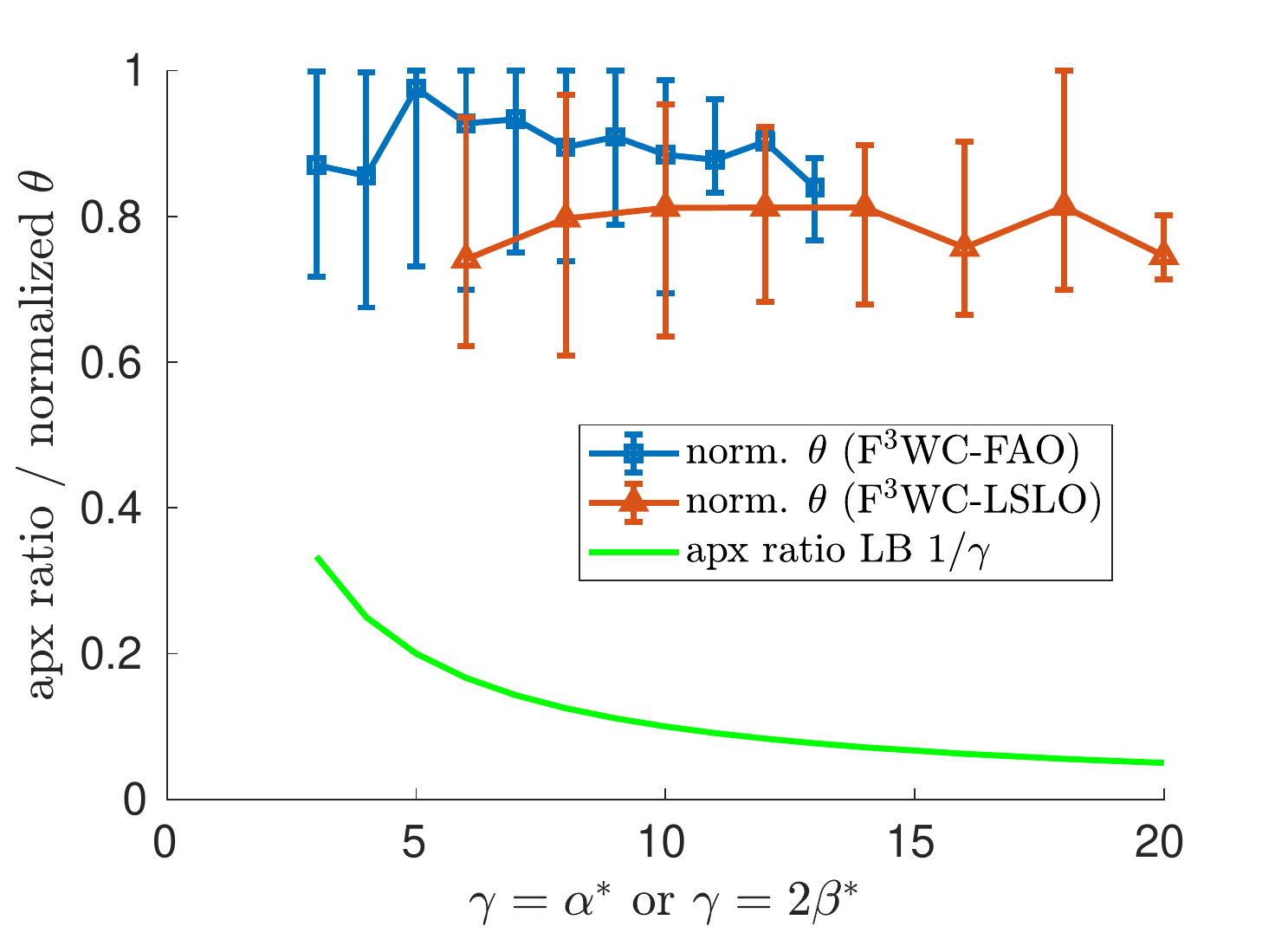}
	\caption{Max-min throughput of full-duplex approximation algorithms for REAL-SU-SM and PI model normalized to that of OPT-FD-MTFS, and the lower bounds of the approximation ratios. Median, $5\%$ and $95\%$ percentiles are shown in the errorbars.}
	\label{fig:fd-apx-theta}	
\end{figure}

If there is mutual interference between links in a backhaul network, we cannot use the optimal full-duplex MTFS scheduling algorithm. However, two fractional weighted coloring based approximation algorithms proposed in~\S\ref{sec:app-fra-color} can be applied. Because the MTFS problem is NP-hard under the PI model, we use the performance of OPT-FD-MTFS as an upper bound.
Fig.~\ref{fig:fd-apx-theta} shows the results for the REAL-SU-SM model. We observe that mmWave backhaul networks are noise-limited instead of interference-limited. 
On average, there are 611 directional links in an evaluated backhaul network, among which only 21 pairs of links are interfering, although we choose a relatively large beamwidth of $20^\circ$.
Despite considering the interference, both algorithms achieve on average more than $70\%$ of the optimal max-min throughput for the ideal interference-free case. In general, F$^3$WC-FAO outperforms F$^3$WC-LSLO in terms of max-min throughput. Besides, the theoretical approximation ratios of Theorem~\ref{thm:perf-f3wc} significantly underestimate the actual performance of the F$^3$WC algorithms. The results for the MAX-SU-SM model are omitted as they are similar.

\begin{figure}
	\centering
	\includegraphics[width=7cm]{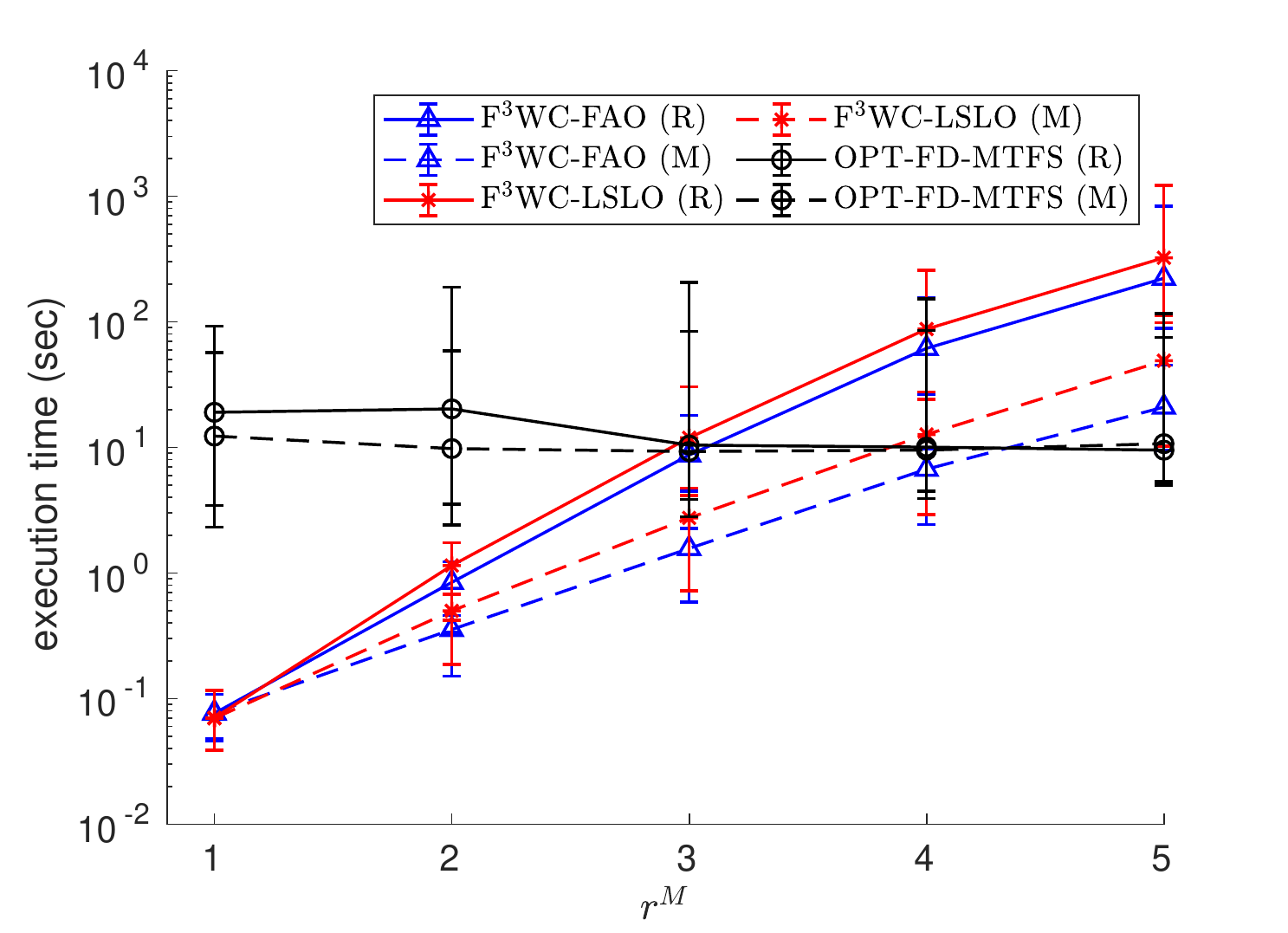}
	\caption{Execution time comparison of approximation algorithms and OPT-FD-MTFS for full-duplex scheduling. Median, $5\%$ and $95\%$ percentiles are shown in the errorbars. (R) and (M) stands for REAL-SU-SM and MAX-SU-SM, respectively.}
	\label{fig:fd-apx-runtime}	
\end{figure}
The execution time of the F$^3$WC algorithms and OPT-FD-MTFS are shown for two SU-SM models in Fig.~\ref{fig:fd-apx-runtime}.
In general, it takes OPT-FD-MTFS less than 100 seconds to schedule a backhaul network with 100 relay BSs and the execution time even decreases with $r^M$. Thus, it is practical to compute the optimal schedule for full-duplex backhauls if interference can be ignored.
The approximation algorithms are more efficient than OPT-FD-MTFS when $r^M$ is small. Yet the execution time goes up quickly with $r^M$, especially for the REAL-SU-SM model.
The reason is due to the large number of vertices in the conflict graph $|V(C)|$ which is equal to the number of arcs in the expanded network $H$ (see \S\ref{sec:conflict_graph}). A F$^3$WC algorithm needs to solve a linear program of $|V(C)| + 1$ variables. For example, with $r^M = 5$ and the REAL-SU-SM model, the linear program has about 30,000 variables, which takes a long time to solve.
For future work it would be interesting to investigate how to shrink the conflict graph, in order to improve the runtime.

\subsection{Half-Duplex Approximation Algorithms}

\begin{figure}
	\centering
	\subfigure[PI and REAL-SU-SM]{\includegraphics[width = 6cm]{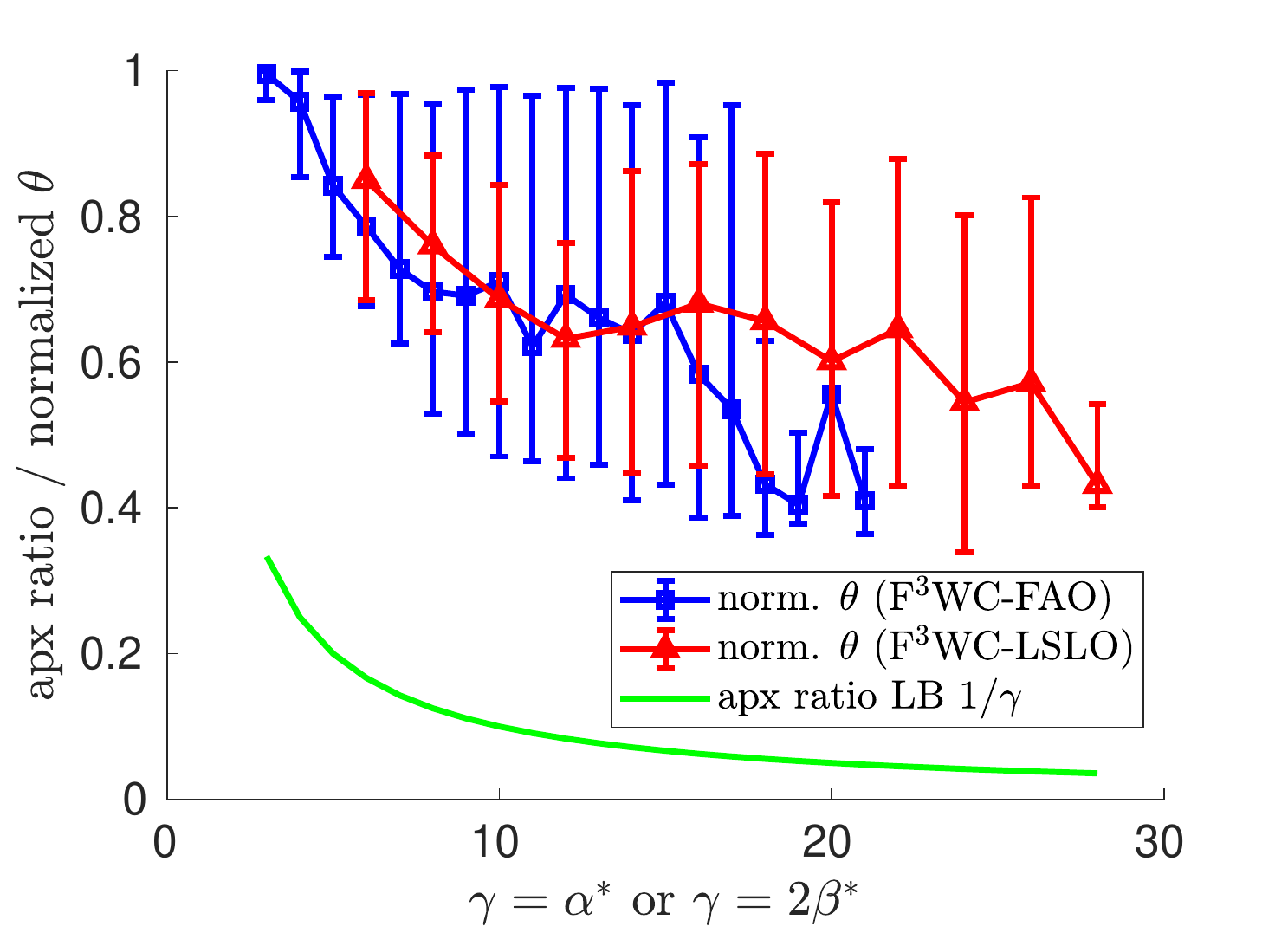}} 
	\subfigure[NI and REAL-SU-SM]{\includegraphics[width = 6cm]{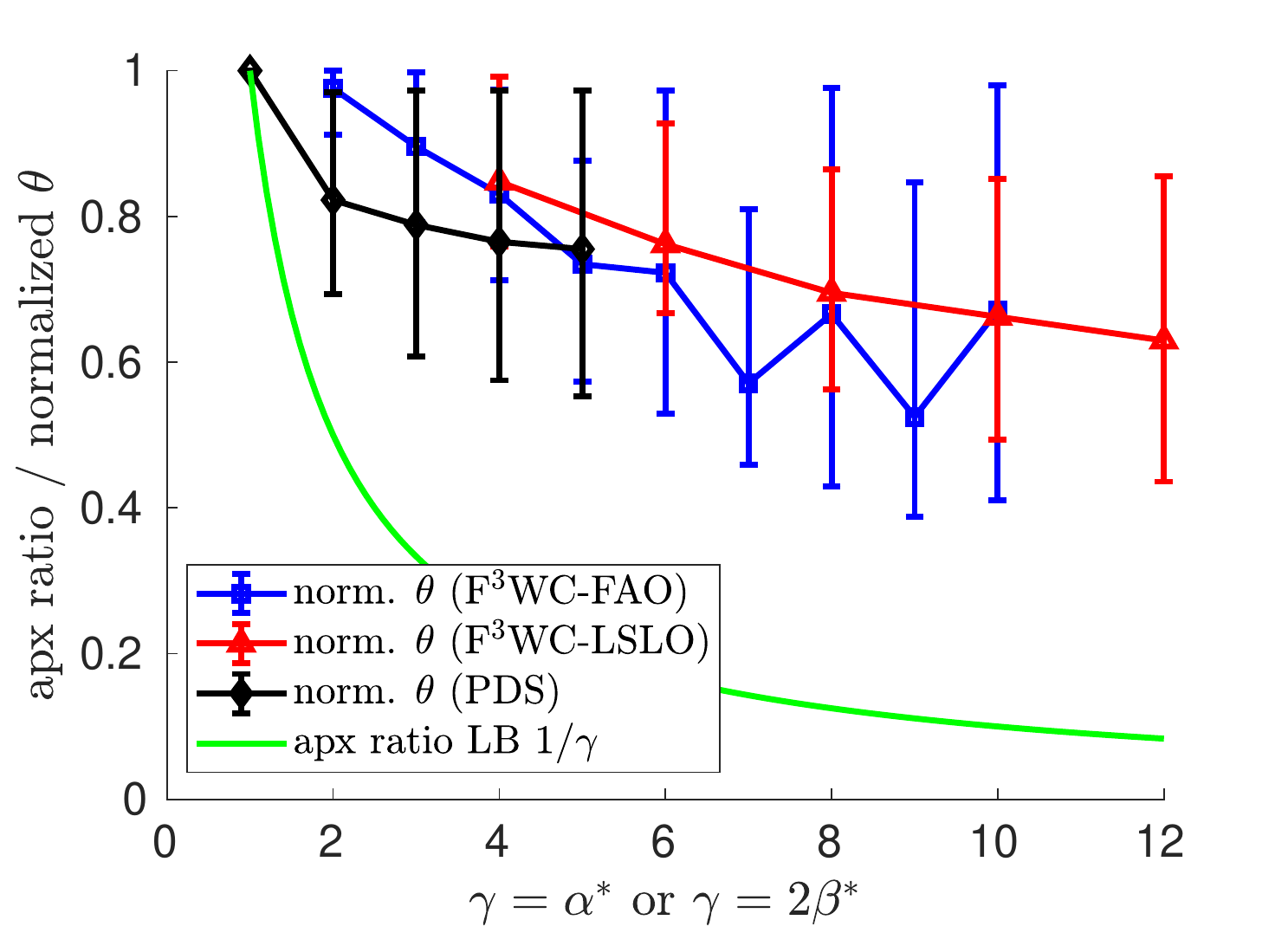}}
	\caption{Max-min throughput of half-duplex approximation algorithms normalized to that of OPT-FD-MTFS, and the lower bounds of the approximation ratios. Median, $5\%$ and $95\%$ percentiles are shown in the errorbars.}
	\label{fig:hd-apx-theta}	
\end{figure}

F$^3$WC-FAO, F$^3$WC-LSLO and PDS are 3 approximation algorithms for half-duplex MTFS scheduling.
The first two work for all cases while PDS only works for the NI model.
We show in \S \ref{ss:opt-algo} that the optimal max-min throughput of the half-duplex MTFS problem is the same or very close to that of full-duplex MTFS for uniform orthogonal backhaul networks. Therefore, we use the max-min throughput of OPT-FD-MTFS as the reference for the evaluation of half-duplex approximation algorithms.
Fig.~\ref{fig:hd-apx-theta}(a) and \ref{fig:hd-apx-theta}(b) show the results for the PI and NI models assuming the REAL-SU-SM model.
All three algorithms attain far better performance than the theoretical lower bounds. 
The two F$^3$WC algorithms have similar performance. Under the NI model, PDS has the best max-min throughput, being higher than $80\%$ on average. The performance of PDS is even better for the MAX-SU-SM model. For example, it is guaranteed to reach the optimal when a backhaul network is uniform orthogonal.

\begin{figure}
	\centering
	\includegraphics[width=7cm]{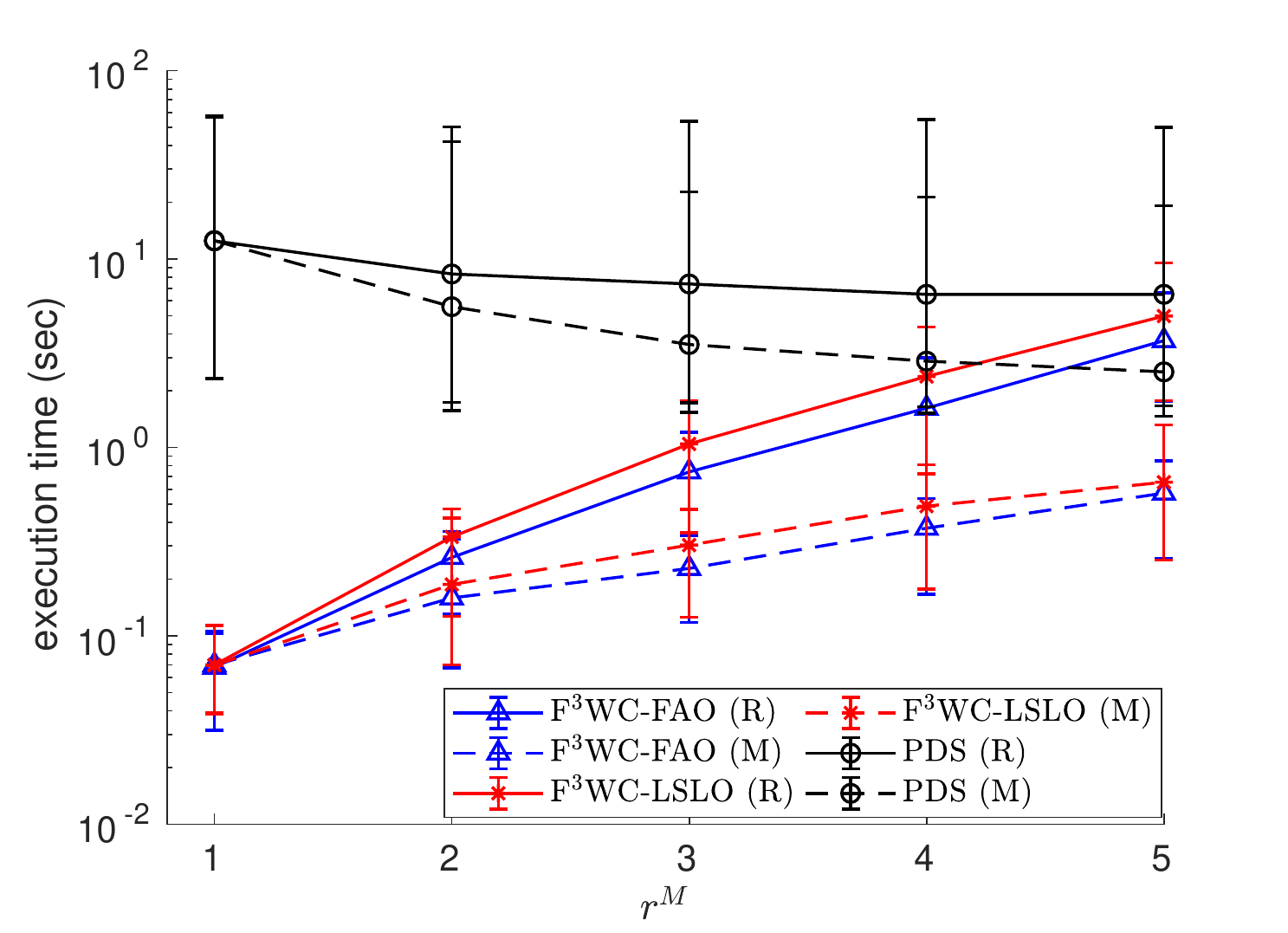}
	\caption{Execution time of the approximation algorithms for half-duplex scheduling. Median, $5\%$ and $95\%$ percentiles are shown in the errorbars. (R) and (M) stands for REAL-SU-SM and MAX-SU-SM, respectively.}
	\label{fig:hd-apx-runtime}	
\end{figure}

Fig.~\ref{fig:hd-apx-runtime} displays the time efficiency of the three approximation algorithms. 
They are all relatively efficient, requiring no more than two minutes. In comparison, F$^3$WC algorithms run faster because we use the property that a directed network can be sparsely expanded under the condition of half-duplex scheduling, which leads to a small conflict graph. We again observe the trend that the execution time of F$^3$WC goes up with $r^M$ while that of PDS goes down. In addition, the execution time of the REAL-SU-SM model is larger than that of the MAX-SU-SM model. This is due to a larger conflict graph for F$^3$WC and an increase in time for maximum weight matching for PDS.

In summary, the evaluation shows that a mmWave backhaul network is generally noise-limited even for a relatively large beamwidth of $20^\circ$. The optimal max-min throughput in practical backhaul networks is quite similar for both full-duplex and half-duplex scheduling. 
PDS is an ideal approximation algorithm for half-duplex scheduling under the NI model as it achieves near optimal performance within practical time. Finally, the two F$^3$WC algorithms have similar max-min throughput. 
They are competitive in execution time for small backhaul networks with a small number of RF chains and half-duplex scheduling.

\section{Conclusion}
\label{s:conclude}
In this article, we studied the scheduling of mmWave backhaul networks assuming a general system model of multiple macro BSs, relay BSs and RF chains as well as interference between links and realistic single-user spatial multiplexing.
Under the assumption of full-duplex radios and interference-free links, we found an optimal joint routing and scheduling method---{\em schedule-oriented optimization} based on matching theory. It can solve any problem formulated as a linear program whose variables are data stream activation durations and QoS metrics. The method is demonstrated to be efficient in practice, capable of solving the maximum throughput fair scheduling (MTFS) problem within a few minutes for a backhaul network of 4 macro BSs, 100 relay BSs and 5 RF chains at each node. 
However, for the more realistic assumption of half-duplex radios or pairwise link interference, we proved that the MTFS problem is NP-hard. 
Subsequently, the paper proposed a number of approximation algorithms with provable performance bounds for the MTFS problem. 
The PDS algorithm works for half-duplex scheduling under the NI (no interference) model. It achieves the optimal performance for uniform orthogonal backhaul networks and about 80\% of the optimum for general backhaul networks.
The F$^3$WC algorithms adapted to our problem are more general than PDS as they support any combination of full-duplex/half-duplex, REAL-SU-SM/MAX-SU-SM model and PI/NI model. Their performance is in general more than half of the optimum. 
In summary, the paper presents optimal and approximation algorithms  that are highly practical for scheduling mmWave cellular networks.

\section*{Acknowledgement}
This work has been performed in the context of the DFG Collaborative Research Center (CRC) 1053 MAKI and the LOEWE center emergenCITY.
It was also supported in part by the Minister of Science and Technology of Taiwan under Grant 104-2911-I-011-503 and the Region of Madrid through TAPIR-CM (S2018/TCS-4496). 
\appendix

\subsection{Proof of Theorem~\ref{thm:MTFS-polynomial-time}}
\label{sec:proof_MTFS-polynomial-time}
\begin{proof}
The proof applies the technique used in~\cite{Nemhauser91} for proving
that fractional edge coloring can be solved in polynomial time by the
ellipsoid algorithm. Specifically, a linear program is solvable in
polynomial time if the separation problem of its dual problem can be
solved in polynomial time. The separation problem of a linear program
$J$ is to determine whether a given solution satisfies all constraints
of $J$ or a violated constraint is identified.

If we can solve both
linear programs of \eqref{eq:mtf-theta} and \eqref{eq:mtf} in
polynomial time, then we can solve the MTFS problem in polynomial
time.
We first prove that \eqref{eq:mtf-theta} can be solved in polynomial
time. The dual of \eqref{eq:mtf-theta} is
\begin{IEEEeqnarray}{lrCl}
  \IEEEyesnumber\label{eq:mtf-theta-dual}\IEEEyessubnumber*
  \textnormal{min} & q & &
  \\
  \text{s.t.}\quad&
  \trans{\vect{p}}\mat{A}^M-q\trans{\vect{1}}& \le &\trans{\vect{0}}\label{eq:mtf-theta-dual-1}
  \\
  & \trans{\vect{p}}\vect{1}& = &1\label{eq:mtf-theta-dual-2}
  \\
  & \vect{p}& \ge &\vect{0}\label{eq:mtf-theta-dual-3}.
\end{IEEEeqnarray}
Let $D$ be the directed network.
Given a solution $(\vect{p}, q)$, \eqref{eq:mtf-theta-dual-2} and
\eqref{eq:mtf-theta-dual-3} can be checked in polynomial time, since
the total number of constraints in \eqref{eq:mtf-theta-dual-2} and
\eqref{eq:mtf-theta-dual-3} is $\card{M(D)}+1$ and $\vect{p}$ contains $\card{M(D)}$
elements.

To check \eqref{eq:mtf-theta-dual-1}, we use the polynomial-time maximum
weighted simple $b$-matching algorithm~\cite[Chap.~33]{Schrijver03}. A constraint of
\eqref{eq:mtf-theta-dual-1} is of the form
$\trans{\vect{p}}\vect{a}^M_k \le q$, where $\vect{a}^M_k$
is the $k$-th column of $\mat{A}^M$ (corresponding to a simple $b$-matching of $D$). 
Define a weight function $w: E(D) \mapsto \mathbb{R}$.
We set the weights to each arc $e = (v_i, v_j)_l \in E(D)$ ($e$ is the $l$-th arc from vertex $v_i$ to vertex $v_j$):
\begin{IEEEeqnarray}{rCl}
  w(e)=
  \begin{cases}
    c(e) (p_j - p_i) & \text{if }v_i \in M(D)
    \\
    c(e) p_j         & \text{otherwise.}
  \end{cases}\IEEEeqnarraynumspace\label{eq:ftr-setw}
\end{IEEEeqnarray}

Then we perform maximum weighted simple $b$-matching on $D$. Let the maximum weight be $w = \max_{k} \trans{\vect{p}}\vect{a}^M_k $.
If $w\le q$, then $(\vect{p}, q)$ satisfies
\eqref{eq:mtf-theta-dual-1}. Otherwise it gives a violated constraint.

According to Theorem~3.10 in~\cite{Groetschel81}, for a linear program
$J$, if we can solve the separation problem of its dual $J^*$ in
polynomial time, then we can solve both $J$ and $J^*$ in polynomial
time with the ellipsoid algorithm. This proves that
\eqref{eq:mtf-theta} can be solved in polynomial time.

Similarly, we next prove that \eqref{eq:mtf} can be solved in
polynomial time. The dual of \eqref{eq:mtf} is
\begin{IEEEeqnarray}{lrCl}
  \IEEEyesnumber\label{eq:mtf-dual}\IEEEyessubnumber*
  \textnormal{min} & \theta^*\trans{\vect{p}}\vect{1}+ q & &
  \\
  \text{s.t.}\quad&
  \trans{\vect{p}}\mat{A}^M+q\trans{\vect{1}}
  & \ge &\trans{\vect{c}}\label{eq:mtf-dual-1}
  \\
  & \vect{p}& \le &\vect{0}.
\end{IEEEeqnarray}
Given a tuple $(\vect{p}, q)$, we set the following weights to
each arc $e = (v_i,v_j)_l \in E(D)$ 
\begin{IEEEeqnarray}{rCl}
  w(e) =
  \begin{cases}
    c(e) (p_i - p_j) & \text{if }v_i \in M(D) \\
    c(e) (1 - p_j)   & \text{otherwise.}
  \end{cases}\IEEEeqnarraynumspace\label{eq:ftr-setw2}
\end{IEEEeqnarray}
Then we perform maximum weighted simple $b$-matching on $D$. Depending on whether
the maximum weight satisfies $w = \max_k (c_k -  \trans{\vect{p}} \vect{a}^M_k) \le q$, the constraints of
\eqref{eq:mtf-dual-1} are satisfied or a violated one is
identified. With the same argument as above, \eqref{eq:mtf} can be
solved in polynomial time. This completes the proof.
\end{proof}

\subsection{Proof of Theorem~\ref{thm:mtfs-intf}}
\label{sec:thm:mtfs-intf}
\begin{proof}
	As is well-known that it is NP-hard to find the {\em fractional chromatic number} $\chi_f(G, \vect{1})$ (minimum fractional weighted coloring assuming each vertex has weight 1) for an arbitrary graph $G$~\cite{Groetschel81}.
	Given a graph $G$, we create a directed network $D$ as follows. $D$ has $2|V(G)|$ vertices and $|V(G)|$ arcs.
	For each $v \in V(G)$, we create a pair of vertices $v^B$ and $v^M$ representing a macro BS and a relay BS, and an arc $(v^B, v^M)$ in $D$.
	For each edge $\{u, v\} \in E(G)$, we specify that the two arcs $(u^B, u^M)$ and $(v^B, v^M)$ in $D$ interfere with each other.
	In addition, we assume that every vertex in $D$ has one RF chain and every arc in $D$ have unit capacity.
	Then it is obvious, that the optimal max-min throughput $\theta^* = 1/\chi_f(G, \vect{1})$. This proves that the MTFS problem is NP-hard under the PI model. This result applies for both half-duplex and full-duplex scheduling as it makes no difference when the RF chain number  is one.
\end{proof}

\subsection{Proof of Lemma~\ref{lem:MCHS-DAG}} \label{sec:proof-lem-MCHS-DAG}
\begin{proof}
We reduce the satisfiability (SAT) problem~\cite{Garey99}, which is NP-hard, to the MWHS problem on a DAG. 
Let $Z = C_1 \wedge \dots \wedge C_K$ be a boolean expression to satisfy.
$Z$ consists of $K$ clauses and each clause $C_k$ is of the form $y_1 \vee \cdots \vee y_J$, 
where $k \in \{1 \dots K\} \eqdef [1:K]$. Note $J$ is the number of literals in $C_k$ and dependent on $k$.
Suppose $Z$ contains in total $L$ boolean variables $x_1 \dots x_L$, then the literals $y_j \in \{x_1, \neg x_1 \dots  x_L, \neg x_L\}$ for $j \in [1:J]$.
We construct a directed network $D$ as follows. Note, $D$ is a strict digraph.
Let $W, Q$ be two disjoint vertex sets with $W = W_1 \cup \cdots \cup W_L$ and $Q = \{q_1 \dots q_K\}$, where $W_l = \{p_l, n_l, r_l\}$, for $l \in [1:L]$. 
Let $V(D) = W \cup Q$, so $D$ has $3L+K$ vertices. Next, we construct the arc set $E(D)$. 
For each clause $C_k$, we define the arc set 
\begin{align*} 
E_k \eqdef &\{(p_l, q_k) \mid \exists y_j \text{ in } C_k \text{ such that } y_j = x_l\}  \\
	  \cup & \{ (n_l, q_k) \mid \exists y_j \text{ in } C_k \text{ such that } y_j = \neg x_l\}.
\end{align*} 
In addition, for each variable $x_l$, we define the arc set
\begin{equation*}
A_l \eqdef \{(r_l, p_l), (r_l, n_l)\}.
\end{equation*}
The arc set of $D$ is
\begin{equation*}
	E(D) = E_1 \cup \cdots \cup E_K \cup A_1 \cup \cdots \cup A_L.
\end{equation*}
The weight is set as $w(e) = 1, \forall e \in E(D)$.
We define the RF chain number function $r$ as: 
\[   
r(v) \eqdef 
\begin{cases}
\max\{\deg^+(v), 1\}, &\text{ if }v = p_l \text{ or }v = n_l, \\
1, &\text{ otherwise,}
\end{cases}
\]
where $\deg^+(v)$ is the outdegree of vertex $v$. 
Obviously, $D$ is a DAG. 
An example for constructing $D$ from a SAT problem is shown in Fig~\ref{fig:sat-2-F}.
To complete the proof, we need to show: \newline
{\bf Claim: } $Z$ is satisfiable if and only if $D$ has a half-duplex subgraph with total weight of $K + L$.

\begin{figure}[htbp]
\centering
	\includegraphics[width=4cm]{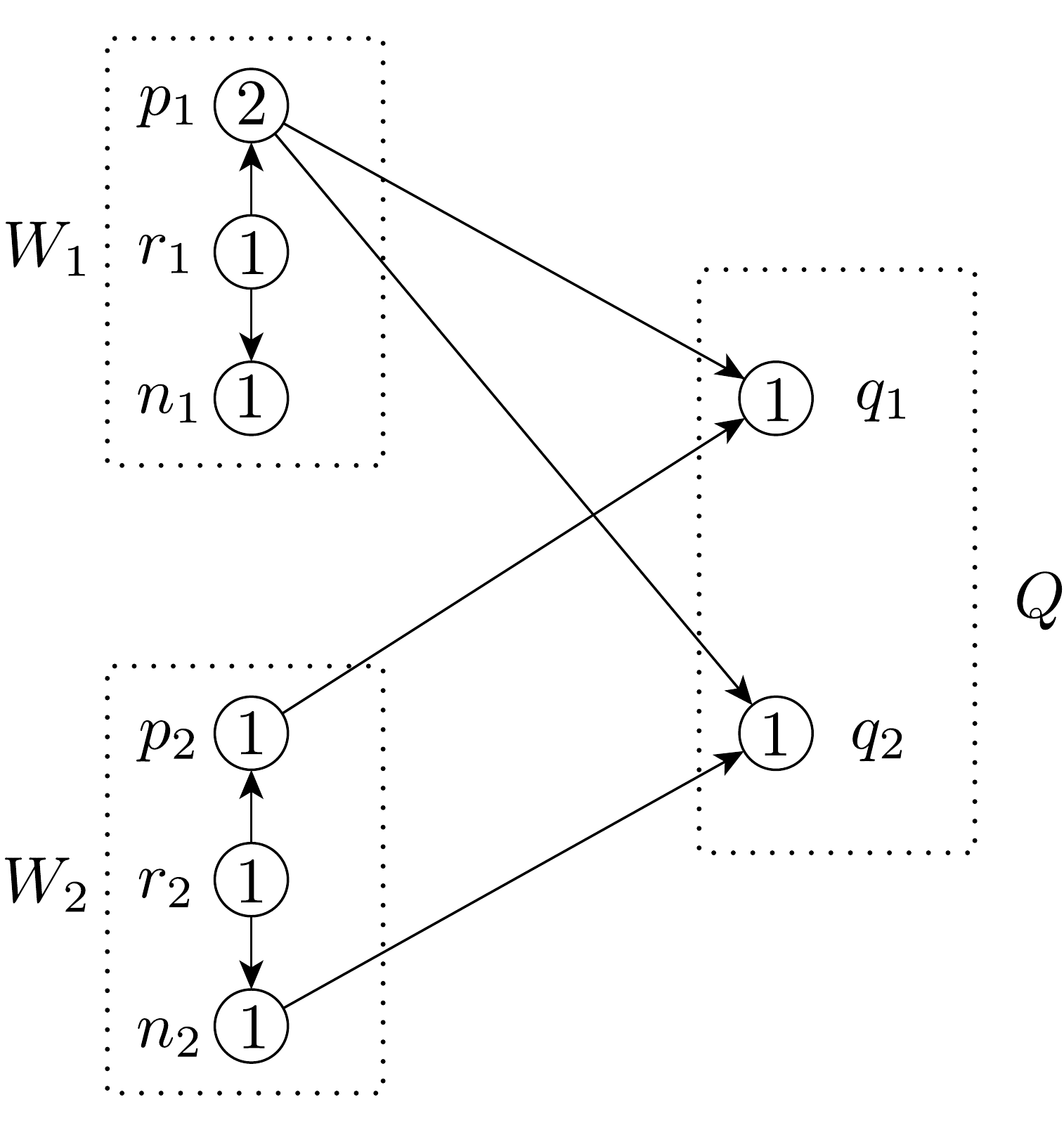}
\caption{The DAG directed network $D$ for $Z = (x_1 \vee x_2) \wedge (x_1 \vee \neg x_2)$. The number in a vertex $v$ is $r(v)$.
$w(e) = 1$ for each arc $e$.}
\label{fig:sat-2-F}	
\end{figure}

Now we prove the claim. Suppose $Z$ is satisfiable. We will select a set of arcs $E \subseteq E(D)$. For each variables $x_l = \text{true}$, we add to $E$ all arcs leaving $p_l$ and the arc $(r_l, n_l)$. For each variable $x_l = \text{false}$, we add to $E$ all arcs leaving $n_l$ and the arc $(r_l, p_l)$. $E$ satisfies the degree and half-duplex constraints on each vertex $w \in W$. Since $Z$ is satisfied, for each $k$, there is at least one arc in $E$ that has one end in $W$ and the other end at $q_k$. We remove arcs from $E$ that are incident to $Q$ until each $q_k$ is incident to exactly one arc. Now $E$ is a half-duplex subgraph of $D$ with total weight $K + L$. 

Conversely, suppose $E$ is a half-duplex subgraph of $D$, then the maximum weight of arcs in $E$ that are between $W$ and $Q$ is $K$ and the maximum weight of arcs in $E$ that are between vertices in $W$ is $L$. If $D$ has a half-duplex subgraph $E$ with total weight $K + L$, then there are exactly $L$ arcs between vertices in $W$, one for each $W_l$. If there is an arc $(r_l, p_l) \in E$, we set $x_l = \text{false}$, otherwise, if there is an arc $(r_l, n_l) \in E$, we set $x_l = \text{true}$. With this assignment $Z$ is satisfied, since $Q$ is incident to exactly $K$ arcs in $E$.
Thus, the MWHS problem is NP-hard on a general directed network that is a DAG.
\end{proof}

\subsection{Proof of Theorem~\ref{thm:HD-MTFS-NP-complete}}
\label{sec:proof-thm-HD-MTFS}
\begin{proof}
Similar to Lemma~\ref{lem:MCHS-DAG}, we prove by reducing the SAT problem~\cite{Garey99}, which is NP-hard, to the full-duplex MTFS problem on a directed network. 
Let $Z = C_1 \wedge \cdots \wedge C_K$ be a boolean expression to satisfy.
$Z$ consists of $K$ clauses and each clause $C_k$ is of the form $y_1 \vee \cdots \vee y_J$, where $k \in \{1 \dots K\} \eqdef [1:K]$.
Note $J$ is the number of literals in $C_k$ and dependent on $k$. Suppose $Z$ contains in total $L$ boolean variables $x_1 \dots x_L$, then the literals
$y_j \in \{x_1, \neg x_1 \dots x_n, \neg x_n\}$ for $j \in [1:J]$.
  We construct a directed network $D$ as follows. Note $D$ is a strict digraph. The construction is more complex than in the proof of Lemma~\ref{lem:MCHS-DAG}, 
  which is necessary for the transformation between a SAT problem and an optimal schedule.
  Let $W, Q$ be two disjoint vertex sets with $W = W_1 \cup \cdots \cup W_L$ and $Q = \{q_1, \dots, q_k\}$, 
  where $W_l = \{p_l^{(1)}, p_l^{(2)}, n_l^{(1)}, n_l^{(2)}, r_l^{(1)}, r_l^{(2)}, r_l^{(3)}, r_l^{(4)}\}$, for $l \in [1:L]$.
  Let $V(D) = W \cup Q$, so $D$ has $8L + K$ vertices. 
  Next, we construct the arc set $E(D)$. 
  Let the $4L$ vertices $r_l^{(m)}, \forall m \in [1:4]$ be macro BSs and all the other vertices be relay BSs.
  For each clause $C_k$, we define the arc set
  \begin{align*} 
  E_k \eqdef  &\{(p_l^{(1)}, q_k), (p_l^{(2)}, q_k)  \mid \exists y_j \text{ in } C_k \text{ such that } y_j = x_l\}  \\
  	\cup  & \{(n_l^{(1)}, q_k), (n_l^{(2)}, q_k)  \mid \exists y_j \text{ in } C_k \text{ such that } y_j = \neg x_l\},
  \end{align*} 
  and set $c(e) = 1, \forall e \in E_k$. In addition, for each variable $x_l$, we define the arc set
  \begin{align*}
  A_l \eqdef \{ (r_l^{(1)}, p_l^{(1)}), (r_l^{(1)}, n_l^{(1)}), (r_l^{(2)}, p_l^{(2)}), (r_l^{(2)}, n_l^{(2)}), \\
  (r_l^{(3)}, p_l^{(2)}), (r_l^{(3)}, n_l^{(1)}), (r_l^{(4)}, p_l^{(1)}), (r_l^{(4)}, n_l^{(2)}) \}.
  \end{align*}
  The capacity of all arcs leaving macro BSs is set to $c(e) = K/2+1, \forall e \in A_l$. The reason for choosing the value $K/2+1$ is that it is a sufficiently large capacity such that the constructed schedule $S$ in the following achieves the optimal max-min throughput of 1.
  The arc set of $D$ is
  \begin{equation*}
	E(D) = E_1 \cup \cdots \cup E_K \cup A_1 \cup \cdots \cup A_L.
  \end{equation*}  
We define the RF chain number function $r$ as: 
\[   
r(v) \eqdef 
\begin{cases}
\max\{\deg^+(v), 2\}, &\text{ if }v = p_l^{(m)} \text{ or }v = n_l^{(m)}, \\
1, &\text{ otherwise,}
\end{cases}
\]
where $\deg^+(v)$ is the outdegree of the vertex $v$.
An example for constructing $D$ from a SAT problem is shown in Fig.~\ref{fig:sat-2-F2}. 
To complete the proof, we need to show: \newline
{\bf Claim: }$Z$ is satisfiable if and only if $D$ has a unit time half-duplex schedule that achieves the max-min throughput $\theta = 1$ 
and the network throughput $\alpha = 4L(K/2+1)$. 
\begin{figure}[htbp]
\centering
	\includegraphics[width=6cm]{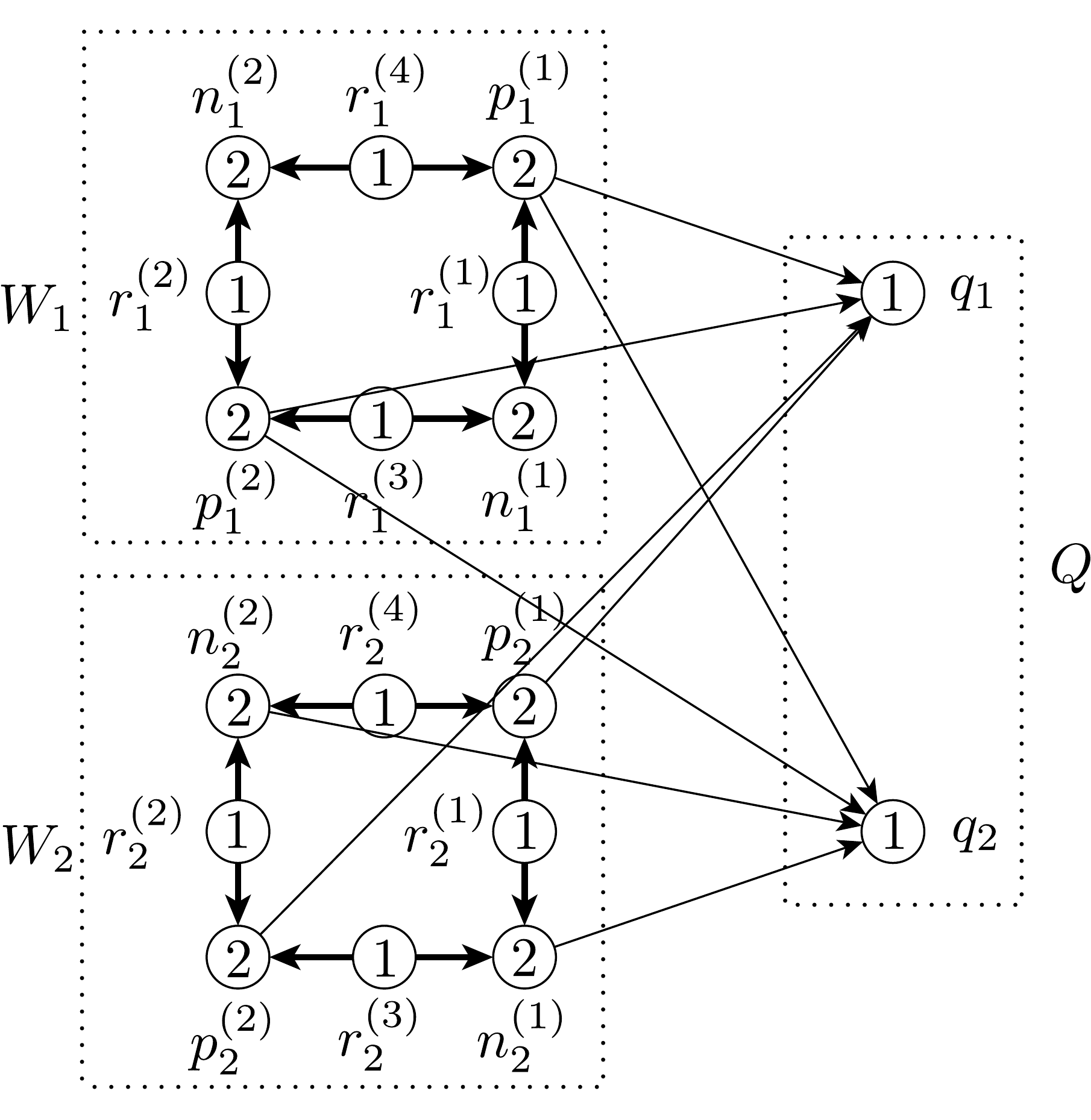}
\caption{The directed network $D$ for $Z = (x_1 \vee x_2) \wedge (x_1 \vee \neg x_2)$. The number in a vertex $v$ is the value $r(v)$. Vertices $r_i^{(j)}$ are macro BSs. The thick arcs $e$ have capacity $c(e) = K/2+1 = 2$ and the thin arcs $e'$ have capacity $c(e') = 1$.}
\label{fig:sat-2-F2}	
\end{figure}

Now we prove the claim. Suppose $Z$ is satisfiable, we create a unit time half-duplex schedule $\set{S}$ that consists of two slots $S_1, S_2 \subseteq E(D)$, each with length 0.5. We first create the arc set $S_1$, For each variable $x_l$, we define the arc set
\[ 
	E_l^{(1)} \eqdef 
	\begin{cases}   \delta^+(p_l^{(1)}) \cup A_l^{(1)}, &\text{if } x_l = \text{true}, \\
                   \delta^+(n_l^{(1)}) \cup B_l^{(1)}, &\text{otherwise,}
    \end{cases}
\]
where $\delta^+(v)$ is the set of arcs that leave vertex $v$, $A_l^{(1)} = \{(r_l^{(1)}, n_l^{(1)}), (r_l^{(2)}, p_l^{(2)}), 
(r_l^{(3)}, p_l^{(2)}), (r_l^{(4)}, n_l^{(2)}) \}$ and $B_l^{(1)} = \{(r_l^{(1)}, p_l^{(1)}), (r_l^{(2)}, n_l^{(2)}),
(r_l^{(3)}, p_l^{(2)}), (r_l^{(4)}, n_l^{(2)}) \}$. Initially,
\begin{equation*}
	S_1 = E_1^{(1)} \cup \cdots \cup E_L^{(1)}.
\end{equation*}
$S_1$ satisfies the degree and half-duplex constraints on each vertex $w \in W$. 
Since $Z$ is satisfied, for each $k$, there is at least one arc in  $S_1$ that has one end in $W$ and the other end at $q_k$. We remove arcs from $S_1$ that are incident to $Q$ until each $q_k$ is incident to exactly one arc. 

$S_2$ is symmetric to $S_1$ in the sense that it can be created from $S_1$: $S_2$ is obtained by scanning the arcs in $S_1$ and
replacing each occurrence of $r_l^{(1)}$ and $r_l^{(2)}$, $r_l^{(3)}$ and $r_l^{(4)}$, $p_l^{(1)}$ and $p_l^{(2)}$, $n_l^{(1)}$ and $n_l^{(2)}$ with each other. 

It is obvious that the schedule $\set{S}$ gives the max-min throughput $\theta = 1$ and the network throughput $\alpha = 4L(K/2+1)$.

Conversely, suppose that a unit time half-duplex schedule $\set{S}'$ achieves the max-min throughput $\theta = 1$ and the network throughput $\alpha = 4L(K/2+1)$. The network throughput $\alpha$ is maximum since $D$ has in total $4L$ single-RF-chain macro BSs and each arc leaving a macro BS has capacity $K/2+1$. So each macro BS must  be always active as a sender in $\set{S}'$. 
In addition, since $\theta = 1$, each relay BS $q_k$ achieves the throughput at least one. Since each $q_k$ has single RF chain and any incoming arc to it has capacity one, $q_k$  must  be always active as a receiver in $\set{S}'$.
We pick an arbitrary slot $S$ from $\set{S}'$. Since among vertices of $W_l$, 4 macro-BS-to-relay-BS arcs are active at any time, it is impossible to have any pair of vertices $p_l^{(m)}$ and $n_l^{(m')}$ ($m, m' \in \{1, 2\}$) active as senders at the same time. 
Otherwise, a macro BS must be inactive which is contradictory to the property of being always active.
Finally, we can set the variables $x_l$ as follows: if none of the 4 vertices $p_l^{(m)}$ and $n_l^{(m')}$ is active as a sender, we set $x_l$ arbitrarily; if one or two of the vertices $p_l^{(m)}$  are active as senders, we set $x_l = \text{true}$; otherwise one or two of the vertices $n_l^{(m')}$  must be active as senders, we set $x_l = \text{false}$. With this assignment $Z$ is satisfied.
Thus, half-duplex MTFS problem is NP-hard for a general directed network.
\end{proof}

\subsection{Proof of Theorem~\ref{thm:hd-mtfs-uniform-rf}}
\label{sec:thm:hd-mtfs-uniform-rf}
To prepare the proof of Theorem~\ref{thm:hd-mtfs-uniform-rf}, let us first prove the following lemma.
\begin{lem} Assume that an undirected loopless multigraph $G$ has the property that between any pair of vertices $u, v \in V(G)$, there are either $R \in \mathbb{N}$ edges of the same  weight $w(\{u, v\})$ or zero edges. The maximum weight biparite subgraph $J \subseteq G$ such that each vertex $v \in V(J)$ has degree $\deg_J(v) \le R$, can be found in polynomial time as follows:
	\begin{enumerate}
		\item Create a simple graph $G'$ for $G$: between each pair of vertices $u, v \in V(G)$, if there are $R$ edges, we remove $R-1$ of them. Let the resulting graph be $G'$.
		\item Find the maximum weight matching $M$ of $G'$ with the weight function $w$. $J$ is a graph whose edge set is the multiset $(M, R)$, i.e., $R$-time repetition of $M$.
	\end{enumerate}
	\label{lem:max-bipartite-subgraph}
\end{lem}
\begin{proof}
	Since $M$ is a matching of $G'$, then $M$ is a bipartite graph such that $\deg_M(v) = 1, \forall v \in V(M)$. Since $J$ is a graph whose edge set is the multiset $(M, R)$, $J$ is a bipartite subgraph of $G$ such that $\deg_J(v) \le R, \forall v \in V(J)$.
	
	Let $K$ be a bipartite subgraph of $G$ such that $\deg_K(v) \le R, \forall v \in V(K)$. Since $K$ is bipartite, its vertices have a bipartition $[U, V]$. We assume without loss of generality that $\card{U} \ge \card{V}$. Then we add $\card{U} - \card{V}$ new vertices to $V$, and add edges between $U$ and $V$ to $K$ until we get a $R$-regular bipartite graph $K'$. {\em Regular} means that each vertex has the same degree, $\deg_{K'}(v) = R, \forall v \in V(K')$. 
	Since $K'$ is a $R$-regular bipartite graph, any subset $S \subseteq U$ is connected with at least $\card{S}$  vertices in $V$ according to the pigeonhole principle. Then according to the Hall's marriage theorem~\cite{Wilson96}, $K'$ contains a matching $N$ with cardinality $\card{U}$. Removing $N$ from $K'$, we get a $(R-1)$-regular bipartite graph. Inductively, we have proved that $K'$ can be decomposed into $R$ matchings. 
	Therefore, $K$, a subgraph of $K'$, can be decomposed into at most $R$ matchings.
	Each matching is a subgraph of $G'$.  Since $M$ is a maximum weight matching of $G'$, $J$ is a maximum weight bipartite subgraph of $G$ such that each vertex $v \in V(J)$ satisfies $\deg_J(v) \le R$. The algorithm is polynomial-time because the maximum weight matching on a graph can be solved in polynomial time~\cite{Edmonds65b}.
\end{proof}

\subsubsection*{Proof of Theorem~\ref{thm:hd-mtfs-uniform-rf}}
\begin{proof}
The linear program formulation of the half-duplex MTFS problem is as follows.
\begin{IEEEeqnarray}{lrCl}
	\IEEEyesnumber\label{eq:hd-mtf-theta}\IEEEyessubnumber*
	\textnormal{max} & \theta & &
	\\
	\textnormal{s.t.}\quad& \mat{L}^{M}\vect{t}^S& \geq & \vect{1}\theta
	\\
	& \trans{\vect{1}}\vect{t}^S& = &1
	\textnormal{ and }
	\vect{t}^S \geq \vect{0},
\end{IEEEeqnarray}
\begin{IEEEeqnarray}{lrCl}
	\IEEEyesnumber\label{eq:hd-mtf}\IEEEyessubnumber*
	\textnormal{max} &\trans{\vect{c}}\vect{t}^S & &
	\\
	\textnormal{s.t.}\quad& \mat{L}^{M}\vect{t}^S& \geq & \vect{1} \theta^*
	\\
	& \trans{\vect{1}}\vect{t}^S& = &1
	\textnormal{ and }
	\vect{t}^S \geq \vect{0},
\end{IEEEeqnarray}	
We prove by solving \eqref{eq:hd-mtf-theta} and \eqref{eq:hd-mtf}, which give the optimal schedule for the half-duplex MTFS problem. 
The method is similar to that of Alg.~\ref{alg:mtf-theta} and Alg.~\ref{alg:mtf}.

The first step is to find an initial basic feasible solution to \eqref{eq:hd-mtf-theta}.
We use the method for the full-duplex MTFS problem in \S\ref{sec:solve-mtfs}.
Suppose the result is a schedule $S_0$. Then we define $S_0'$ to be $R$ copies of $S_0$ running in parallel. 
Obviously, $S_0'$ is an initial basic feasible solution to \eqref{eq:hd-mtf-theta}.

To compute the max-min throughput, Alg.~\ref{alg:mtf-theta} and Alg.~\ref{alg:mtf} need to be modified.
In Line \ref{alg:mtf-a1} of both Alg.~\ref{alg:mtf-theta} and Alg.~\ref{alg:mtf}, we replace "Do max weight simple $b$-matching on $D$" with "Solve the MWHS problem on $D$".

The MWHS problem on $D$ can be solved as follows. Let $D'$ be a subgraph of $D$ that contains only positive arcs.
The solution of MWHS on $D$ is the same as that on $D'$.
Note that $D'$ satisfies the condition that if there is an arc $(u, v) \in E(D')$, no opposite arcs $(v, u)$ are contained in $D'$. The reason is as follows.
If both $u, v$ are relay BSs and $(u, v) \in E(D')$, then $w((v, u)) = -w((u, v)) < 0$ and the $(v, u)$ arcs will be removed.
 If $u$ is a macro BS, then $(v, u)$ are not contained in $D'$. Moreover, between any two vertices in $D'$, there are either $R$ equivalent arcs (same head, tail and weight) or zero arcs. So, $D'$ can be considered as a weighted undirected loopless multigraph of the uniform edge multiplicity $R$. 
 Because a half-duplex subgraph of $D'$ must be a bipartite subgraph with degree constraint $R$, the maximum weight bipartite subgraph $B \subseteq D'$ with degree constraint $R$ has weight greater than or equal to that of the maximum weight half-duplex subgraph of $D'$. 
 From Lemma~\ref{lem:max-bipartite-subgraph}, $B$ is also a half-duplex subgraph of $D'$. So it is also the maximum weight half-duplex subgraph of $D'$ and $D$.
 Therefore, the optimal schedule for the half-duplex MTFS problem $S^*$ consists of $R$ copies of the same schedule $S$ executed in parallel since each iteration in Alg.~\ref{alg:mtf-theta} and Alg.~\ref{alg:mtf} produces such a schedule. $S$ must be a unit time schedule for $D$ assuming that each node has one RF chain. Consequently, the optimal schedule for the half-duplex MTFS problem is obtained by the algorithm in Theorem~\ref{thm:hd-mtfs-uniform-rf}.
\end{proof}

\subsection{Proof of Theorem~\ref{thm:sparse-exp-net}}
\label{sec:thm:sparse-exp-net}
\begin{proof}
	\begin{figure}[!bthp]
		\centering
		\includegraphics[width=5cm]{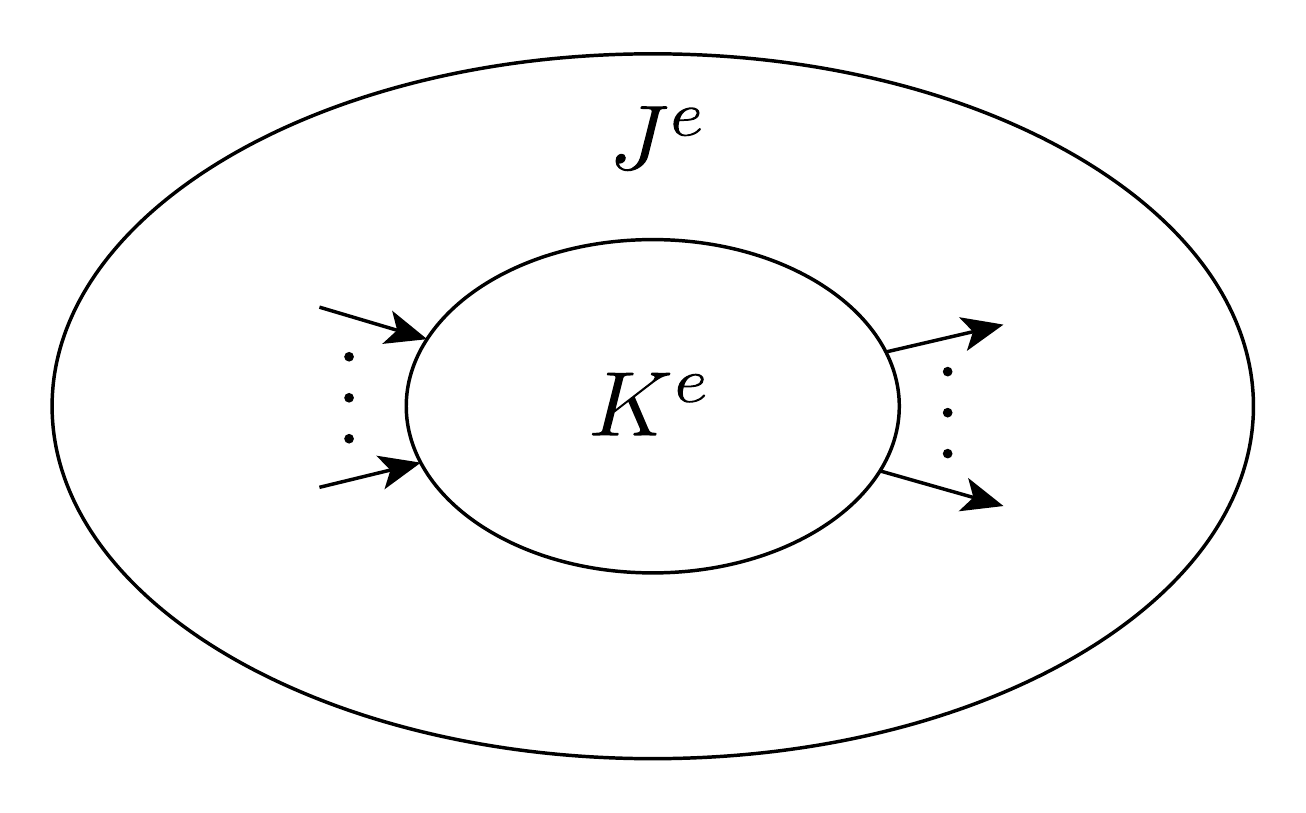}
		\caption{Arcs of $J^e$ entering and leaving $K^e$.}	
		\label{fig:max-ind-subg}
	\end{figure}
	Let an arbitrary half-duplex subgraph of the directed network $D$ be $J$. Obviously, $J$ corresponds to a certain matching $J^e$ (which are the data streams scheduled in a timeslot) in the expanded network $H_{\text{FD}}^{\text{M}}$ as it is fully expanded. We need to prove that $J^e$ is equivalent to $J^s$ which is a matching of the sparsely expanded network $H_{\text{HD}}^{\text{M}}$.
	
	Initially, we set all vertices of $J$ as untagged and let $J^s = \emptyset$.
	Starting from an untagged vertex $v$ of $J$ (suppose it has RF chain number $r(v)$), we find the {\em maximal induced subgraph with $r(v)$ RF chains} $K$, which is defined as a connected (two vertices are connected if there is an arc between them) induced subgraph of $J$ that has the largest number of vertices of exactly $r(v)$ RF chains. $K$ is a bipartite graph with maximum vertex degree of at most $r(v)$ because the vertices in $K$ can be divided into the sender and receiver sets.
	
	Let the {\em link network} of $D$ be $L$.
	According to the K\H{o}nig's Theorem~\cite{Wilson96} for the edge coloring of bipartite graphs, $K$ can be decomposed into at most $r(v)$ matchings in $G$, where $G$ is the induced subgraph of $L$ by the vertex set $V(K)$. 
	In the sparsely expanded network $H_{\text{HD}}^{\text{M}}$, $G$ is expanded into $G^s$ which is $r(v)$ copies of $G$.
	Suppose that $K$ corresponds to a graph $K^e \subseteq H_{\text{FD}}^{\text{M}}$. By rearranging senders and receivers, $K^e$ is equivalent to a $K^s \subseteq G^s$. We add $K^s$ to $J^s$.
	Suppose that $J^e$ has an arc $e$ that goes from a vertex $u^{(i)}$ out of $K^e$ to a vertex $w^{(j)}$ inside $K^e$ (see Fig.~\ref{fig:max-ind-subg}). Then in $K^s$ there is at least one vertex $w^{(k)}$ ($k$ may be different from $j$) which is unconnected. We map $(u^{(i)}, w^{(j)})$ in $J^e$ to $(u^{(i)}, w^{(k)})$ and add the latter arc to $J^s$. This works in the same way for outgoing arcs. After we have processed all arcs entering and leaving $K^e$, we tag all vertices in $K$. Then we go on to process untagged vertices in $J$.
	After we have tagged all vertices, we get a matching $J^s \subseteq H_{\text{HD}}^{\text{M}}$. 
\end{proof}

\subsection{Proof of Theorem~\ref{thm:perf-f3wc}}
\label{sec:thm:perf-f3wc}
Before proving Theorem~\ref{thm:perf-f3wc}, we need to first prove two lemmas:
\begin{lem}
	Let $D$ be a directed network and $L$ be the corresponding link network. We have $Q \subseteq P \subseteq \alpha^*Q$.
	For any $\vect{t} \in Q$, the coloring of $(C, \vect{t})$ by F$^3$WC-FAO has weight at most 1. Furthermore,	 
	\begin{itemize}
		\item $\alpha^* \le \max \Big( 1, \max_{l \in E(L)} \big( \sum_{l' | \mathrm{intf}(l', l) = 1} d(l') \big) \Big) + 2$ for the case of full-duplex network, PI and REAL-SU-SM model.
		\item $\alpha^* \le \max_{l \in E(L)} \big( \sum_{l' | \mathrm{intf}(l', l) = 1} d(l') \big) + 2$ for the case of full-duplex network, PI and MAX-SU-SM model.
		\item $\alpha^* \le \max_{l \in E(L)} \big(r(l) +  \sum_{l' | \mathrm{intf}(l', l) = 1} d(l') \big)$ for the case of half-duplex network and PI model where $r(l) = r(u) + r(v)$ for $l = (u,v) \in E(L)$.
	\end{itemize}
	\label{lem:Q}
\end{lem}
\begin{proof}
	Let $H$ be the expanded network\footnote{There are four possibilities of $H$ depending on the modeling.} of $D$. Let $e$ be an arc in $H$. 
	Correspondingly, $e$ is a vertex in the conflict graph $C$. We check the neighbors of $e$ in $C$. 
	If any independent set of $C$ contained in $e$ and its neighbors, has a size at most $N$, then the {\em inductive independence number} of $C$ is $\alpha^* \le N$, where $\alpha^*$ is defined to be the maximum size of any independent set of $C$ contained in some $V_i$, for $1 \le i \le n$. 
	$Q \subseteq P \subseteq \alpha^*Q$ follows directly from Corollary 5.2 of~\cite{Wan09}.
	By Theorem 5.1 of~\cite{Wan09}, for any $\vect{t} \in Q$, the coloring of $(C, \vect{t})$ by F$^3$WC-FAO has weight at most $\max_{1 \le i \le n}\{ t(V_i) \} \le 1$. Now we look at four backhaul network modelings.
	\begin{figure}[!htbp]
		\centering
		\includegraphics[width=7cm]{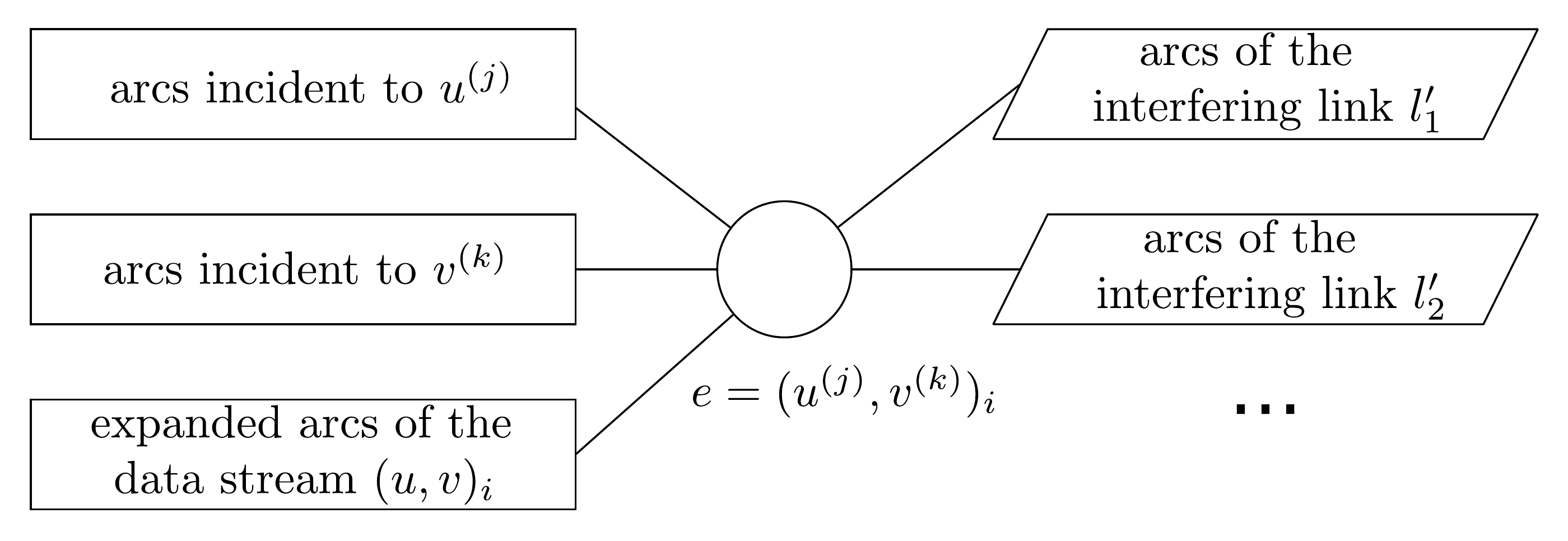}
		\caption{The neighbors of a vertex $e = (u^{(j)}, v^{(k)})_i \in V(C)$.
			A rectangle is a complete graph that contributes only one vertex to the independent set while	 
			a parallelogram may contribute multiple.}	
		\label{fig:fd-r-cft-nbr}
	\end{figure}
	
	\begin{figure}[!htbp]
		\centering
		\includegraphics[width=7cm]{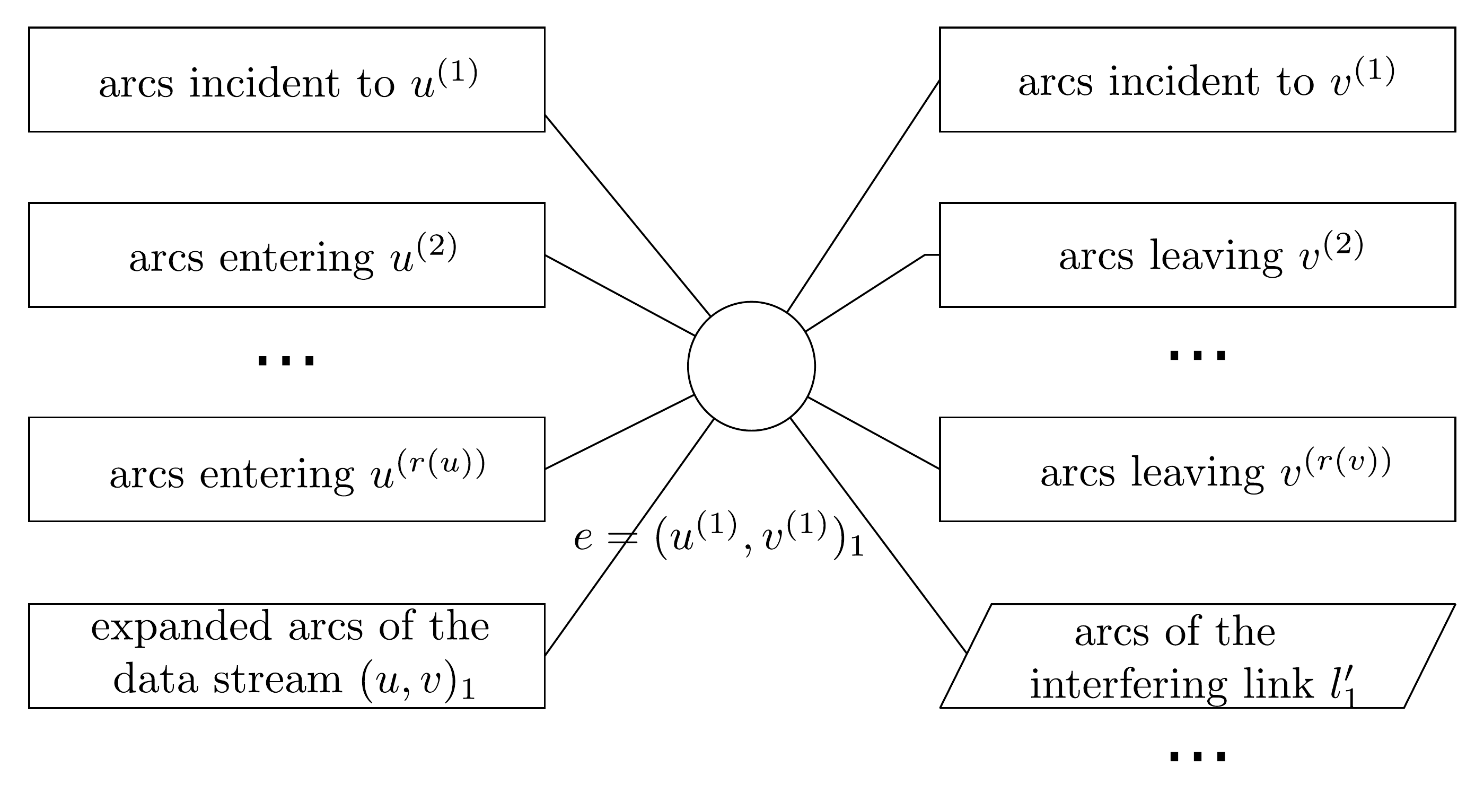}
		\caption{The neighbors of a vertex $e = (u^{(1)}, v^{(1)})_1 \in V(C)$.}	
		\label{fig:hd-r-cft-nbr}
	\end{figure}
	
	1. We check the case of full-duplex, PI and REAL-SU-SM model.
	Let $e = (u^{(j)}, v^{(k)})_i$ be an arc in $H_{\text{FD}}^{\text{R}}$. The neighbors of $e$ in $C$ is shown in Fig.~\ref{fig:fd-r-cft-nbr}. The arcs incident to $u^{(j)}$ form a complete graph in $C$, so they contributes only one vertex to the independent set. This applies also be the arcs incident to $v^{(j)}$. 
	If there are no links that are interfering with $(u, v)$. Then the independent set can be extended by one due to the expanded arcs of $(u, v)_i$. Otherwise, it can be extended by $\sum_{l' | \mathrm{intf}(l', l) = 1} d(l')$ because each interfering link $l'$ may contribute at most $d(l')$. Because these arcs are in conflict with the expanded arcs of $(u, v)_i$, the latter does not contribute to the independent set.
	
	2. We check the case of full-duplex, PI and MAX-SU-SM model. The neighbors of $e = (u^{(j)}, v^{(k)})$ is similar to Fig.~\ref{fig:fd-r-cft-nbr} except that we don't have the neighbors that are the expanded arcs of a data stream.
	
	3. We check the case of half-duplex, PI and REAL-SU-SM model. Let $e = (u^{(1)}, v^{(1)})_1$ be an arc in $H_{\text{HD}}^{\text{R}}$ without loss of generality. The neighbors of $e$ in $C$ is shown in Fig.~\ref{fig:hd-r-cft-nbr}. The arcs incident to $u^{(1)}$, entering $u^{(j)}$, incident to $v^{(1)}$ and
	leaving $v^{(k)}$ for any $j \ne 1$ and any $k \ne 1$, each contributes one vertex to the independent set.
	If an expanded arc, say $(u^{(l)}, v^{(m)})_1$ is add to the independent set, then we need to remove two vertices from the independent set that belong to the arcs entering $u^{(l)}$ and those leaving $v^{(m)}$.
	So it is not worth to do that.
	The increase of the independent set due to the interfering links is the same as that of full-duplex scheduling.
	
	4. The case of half-duplex, PI and MAX-SU-SM model is the same as 3.
\end{proof}

\begin{lem}
	Let $D$ be the directed network and $L$ be the corresponding link network. We have $Q' \subseteq P \subseteq 2 \beta^* Q'$.
	For any $\vect{t} \in Q'$, the coloring of $(C, \vect{t})$ by F$^3$WC-LSLO has weight at most 1.
	Furthermore,
	\begin{itemize}
		\item $\beta^* \le \max \Big( 1, \max_{l \in E(L)} \big( \sum_{l' |l' < l, \mathrm{intf}(l', l) = 1} d(l') \big) \Big) + 2$ for the case of full-duplex network, PI and REAL-SU-SM model.
		\item $\beta^* \le \max_{l \in E(L)} \big( \sum_{l' |l' < l, \mathrm{intf}(l', l) = 1} d(l') \big) + 2$ for the case of full-duplex network, PI and REAL-SU-SM model.
		\item $\beta^* \le \max_{l = (u, v) \in E(L)} \big(r(u) +  \sum_{l' | l' < l, \mathrm{intf}(l', l) = 1} d(l') \big) \\ +1$ 
		for the case of half-duplex network and PI model.
	\end{itemize}
	\label{lem:Q'}
\end{lem}
\begin{proof}
	The {\em local independence number} of $C^d$ is defined to be the maximum size of any independent set of $C$ contained in $N^{in}[u]$ for any $u \in V(C^d)$ and is denoted by $\beta^*$. $Q' \subseteq P \subseteq 2\beta^* Q'$ follows directly from Corollary 5.4 of~\cite{Wan09}. By Lemma 5.3 of~\cite{Wan09}, for any $\vect{t} \in Q'$, the coloring of $(C, \vect{t})$ by F$^3$WC-LSLO has weight at most $2  \max_{u \in V(C^d)} t(N^{in}[u]) \le 1$. Now we check the local independence number for four backhaul network modeling.
	
	1. We check the case of full-duplex, PI and REAL-SU-SM model. Let $e = (u^{(j)}, v^{(k)})_i$ be an arc in $H_{\text{FD}}^{\text{R}}$. 
	By observing Fig.~\ref{fig:fd-r-cft-nbr}, we see that one  of the arcs incident to $u^{(j)}$ and one of the arcs incident to $v^{(k)}$ may be in-neighbors of $e$ according to the edge orientation rule and be added to the independent set.
	That also applies to one of the expanded arcs of $(u, v)_i$ if there are no interfering arcs that are also in-neighbors of $e$. 
	Otherwise, we will add the maximum number of interfering arcs (their links must be smaller than $(u, v)$) but not an expanded arc of $(u, v)_i$ because they are in conflict.
	
	2. The case of full-duplex, PI and MAX-SU-SM model is similar to case 1 except that we don't have expanded arcs of a data stream being modeled in $H_{\text{FD}}^{\text{M}}$. 
	
	3. We check the case of half-duplex, PI and REAL-SU-SM model. Let $e = (u^{(j)}, v^{(k)})_i$ be an arc in $H_{\text{HD}}^{\text{R}}$. 
	By observing Fig.~\ref{fig:hd-r-cft-nbr}, we see that one arc each from the arcs incident to $u^{(j)}$, entering $u^{(l)}$ and incident to $v^{(k)}$ for all $l \ne j$ may be an in-neighbor of $e$ and be added to the independent set together. Now no expanded arc of $(u, v)_i$ can be added to the independent set because it is in conflict with the current independent set. Also no arc $e'$ leaving $v^{(m)}$ for any $m \ne k$ can be added to the independent set. Otherwise there must be an oriented edge $(e', e)$. According to the orientation rule, $e'$ must end with a vertex $u^{(n)}$. Thus $e'$ is in conflict with the current independent set. The argument regarding the interfering arcs of $e$ is the same as in 1.
	
	4. The result of 3. also applies to the case of half-duplex, PI and MAX-SU-SM model.
\end{proof}
\subsubsection*{Proof of Theorem~\ref{thm:perf-f3wc}}
\begin{proof}
	We prove for the first sorting method (FAO). The second (LSLO) can be proved in the same way.
	The optimal max-min throughput $\theta^*$ can be obtained by solving the linear program
	\begin{IEEEeqnarray}{lrClr}
		\label{eq:opt-max-min-tput}
		\IEEEyesnumber \IEEEyessubnumber* 
		\textnormal{max} & \theta & & & \\
		\textnormal{s.t.} &\sum_{\mathclap{e\in\delta_{H}^-(U(v))}} c(e) t_e
		-\sum_{\mathclap{e\in\delta_{H}^+(U(v))}} c(e) t_e& \ge & \theta
		\quad\forall\,v\in M(D)\IEEEeqnarraynumspace 
		\\
		& \vect{t} & \in & P \IEEEeqnarraynumspace  
	\end{IEEEeqnarray}
	and the approximate max-min throughput $\theta$ is obtained by solving \eqref{eq:q0app1} with $Q^\circ = Q$.
	Suppose that the optimal solution to \eqref{eq:opt-max-min-tput} is $(\theta^*, \vect{t}^*)$, we will see that
	$(\theta^*/\alpha^*, \vect{t}^*/\alpha^*)$ is a solution to \eqref{eq:q0app1}.
	Since $\vect{t}^* \in P$, we have $\vect{t}^*/\alpha^* \in P /\alpha^* \subseteq Q$ due to Lemma~\ref{lem:Q}.
	In addition, we verify that $(\theta^*/\alpha^*, \vect{t}^*/\alpha^*)$ satisfies~\eqref{eq:q0app-maxmin-tput}.
	Therefore $\theta \ge \theta^*/\alpha^*$. In addition, because the  schedule $S$ produced by F$^3$WC-FAO without the scaling step satisfies $\vect{t} \in Q$, we have that the length of $S$ is at most 1 due to Lemma~\ref{lem:Q}. 
	The scaling step makes the schedule length exactly one and the final max-min throughput $\theta' \ge \theta$.
\end{proof}

\subsection{Proof of Theorem~\ref{thm:pls-perf}}
\label{sec:thm:pls-perf}
\begin{proof}
	Assuming REAL-SU-SM model, 
	let $r_{max}^{D^e}$ be the maximum number of RF chains of any vertex in $D^e$. 
	Given $D^e$, let the optimal max-min throughput of the half-duplex MTFS problem be $\theta_{min}$ when all vertices of $D^e$ have 1 RF chain (each link $l$ has one data stream of the largest capacity $c(l_1)$). Let the optimal max-min throughput be $\theta_{max}$ when all vertices of $D^e$ have $r_{max}^{D^e}$ and each link has $r_{max}^{D^e}$ data streams of the largest capacity $c(l_1)$. According to Theorem~\ref{thm:hd-mtfs-uniform-rf},
	$\theta_{min} \le \theta \le \theta^* \le \theta_{max} = r_{max}^{D^e} \theta_{min}$.
	We have $\theta \ge \theta_{min} = \frac{\theta_{max}}{r_{max}^{D^e}}  \ge \frac{\theta^*}{r_{max}^{D^e}}$. 
	$r_{max}^{D^e} = \max(r_{max}^M, \max_{v \in B(D)} m(v))$. In addition, $m(v) \le 2d_{min} - 1$.

	Assuming MAX-SU-SM model, $D$ and $D^e$ has the same optimal max-min throughput for the MTFS problem $\theta^*$.
	Let $r_{min}^{D^e}$ and $r_{max}^{D^e}$ be the minimum and maximum number of RF chains of any vertex in $D^e$.
	Given $D^e$, let the optimal max-min throughput of the half-duplex MTFS problem be $\theta_{min}$ 
	and $\theta_{max}$ when all vertices of $D^e$ have $r_{min}^{D^e}$ and $r_{max}^{D^e}$ number of RF chains, respectively.
	According to Theorem~\ref{thm:hd-mtfs-uniform-rf}, $\theta_{min} \le \theta \le \theta^* \le \theta_{max} = \frac{r_{max}^{D^e}}{r_{min}^{D^e}} \theta_{min}$.
	We have $\theta \ge \theta_{min} = \frac{r_{min}^{D^e}}{r_{max}^{D^e}} \theta_{max} \ge \frac{r_{min}^{D^e}}{r_{max}^{D^e}} \theta^*$.  Moreover, $r_{min}^{D^e} = d_{min} = r_{min}$. On the other hand, $r_{max}^{D^e} = \max(r_{max}^M, \max_{v \in B(D)} m(v))$. In addition, $m(v) \le 2r_{min} - 1$, which completes the proof.
\end{proof}

\ifCLASSOPTIONcaptionsoff
  \newpage
\fi



%
\bibliographystyle{IEEEtran}


\bibliography{biblio}





\end{document}